\documentclass[journal]{IEEEtran}

\IEEEoverridecommandlockouts

\ifCLASSINFOpdf
   \usepackage[pdftex]{graphicx}
  \graphicspath{{figures/}}
  \usepackage{epstopdf}
  \DeclareGraphicsExtensions{.pdf,.jpeg,.png}
\else
  \usepackage[dvips]{graphicx}
   \graphicspath{{figures/}}
   \DeclareGraphicsExtensions{.eps}
\fi

\usepackage{pgfplots}
\pgfplotsset{compat=newest}
\usetikzlibrary{plotmarks}
\usetikzlibrary{arrows.meta}
\usepgfplotslibrary{patchplots}
\usepackage{grffile}
\pgfplotsset{plot coordinates/math parser=false}
\newlength\figureheight
\newlength\figurewidth

\usepackage{pdfpages} 

\usepackage{stfloats}
\usepackage{nicefrac}
\usepackage{amsthm}
\usepackage{amssymb}
\usepackage{pifont}

\usepackage{latexsym}
\usepackage{times}
\usepackage{bm} 
\usepackage{array}
\usepackage[cmex10]{amsmath}
\usepackage{multicol}
\usepackage{urwchancal}

\usepackage{booktabs}

\usepackage{tikz}
\usepackage{textcomp}
\usepackage{hyperref}
\usepackage{lipsum}

\newcommand\copyrighttext{%
  \footnotesize \textcopyright 2020 IEEE. Personal use of this material is permitted.
  Permission from IEEE must be obtained for all other uses, in any current or future
  media, including reprinting/republishing this material for advertising or promotional
  purposes, creating new collective works, for resale or redistribution to servers or
  lists, or reuse of any copyrighted component of this work in other works.
  DOI: 10.1109/TWC.2020.3027345
  }
\newcommand\copyrightnotice{%
\begin{tikzpicture}[remember picture,overlay]
\node[anchor=south,yshift=10pt] at (current page.south) {\fbox{\parbox{\dimexpr\textwidth-\fboxsep-\fboxrule\relax}{\copyrighttext}}};
\end{tikzpicture}%
}

\usepackage{eso-pic}

\usepackage{cleveref}

\usepackage{support-caption}
\usepackage[subrefformat=parens,labelformat=parens]{subcaption}

\usepackage{multirow}
\usepackage{xcolor}

\usepackage[nolist,printonlyused]{acronym}      

\usepackage{optidef} 
\usepackage{algpseudocode, algorithm}

\makeatletter
\algrenewcommand\ALG@beginalgorithmic{\footnotesize}
\makeatother
\algrenewcommand\algorithmiccomment[2][\normalsize]{{#1\hfill\(\triangleright\) #2}}
\usepackage{comment}
\usepackage{cite}
\usepackage{titlecaps}
\usepackage{setspace}

\usepackage{thmtools}

\newtheorem{theorem}{Theorem}

\newtheorem{lemma}{Lemma}
\newtheorem{definition}{Definition}

\usepackage{xpatch}
\iftrue 
\makeatletter
\xpatchcmd{\proof}{\@addpunct{.}}{\normalfont\,\@addpunct{:}}{}{}
\makeatother
\fi

\DeclareMathAlphabet\mathbfcal{OMS}{cmsy}{b}{n}

\begin{acronym}[LTE-Advanced]
  \acro{2G}{Second Generation}
  \acro{3G}{3$^\text{rd}$~Generation}
  \acro{3GPP}{3$^\text{rd}$~Generation Partnership Project}
  \acro{4G}{4$^\text{th}$~Generation}
  \acro{5G}{5$^\text{th}$~Generation}
  \acro{5GPPP}{5G Infrastructure Public Private Partnership}
  \acro{ACLR}{adjacent channel leakage ratio}
  \acro{ADMM}{alternating direction method of multipliers}
  \acro{ARQ}{automatic repeat request}
  \acro{BER}{bit error rate}
  \acro{BLER}{block error rate}
  \acro{BPC}{binary power control}
  \acro{BPSK}{binary phase-shift keying}
  \acro{BRA}{balanced random allocation}
  \acro{BS}{base station}
  \acro{BW}{bandwidth}
  \acro{CADMM}{\ac{CSI}-aware \ac{ADMM}}
  \acro{CAP}{combinatorial allocation problem}
  \acro{CAPEX}{capital expenditure}
  \acro{CBF}{coordinated beamforming}
  \acro{CP}{cyclic prefix}
  \acro{CS}{coordinated scheduling}
  \acro{CPOCS}{\ac{CSI}-aware \ac{POCS}}
  \acro{CRUISE-DAYS}{Convex RedUctIon of unwanted Spectral Emission with Davis And Yin Splitting}
  \acro{cruise}{convex reduction of unwanted spectral emission}
  \acro{CRUISE-SSP}{\ac{CRUISE}-\ac{SSP}}
  \acro{CSI}{channel state information}
  \acro{CSIT}{channel state information at the transmitter}
  \acro{D2D}{device-to-device}
  \acro{DCA}{dynamic channel allocation}
  \acro{DCI}{downlink control information}
  \acro{DE}{differential evolution}
  \acro{DFT}{discrete Fourier transform}
  \acro{DIST}{Distance}
  \acro{DL}{downlink}
  \acro{DMA}{double moving average}
  \acro{DMRS}{demodulation reference signal}
  \acro{D2DM}{D2D mode}
  \acro{DMS}{D2D mode selection}
  \acro{DROPE}{Douglas-Rachford-based \ac{OOBE} power reduction with \ac{EVM} constraint}
  \acro{DRS}{Douglas-Rachford splitting}
  \acro{DYS}{Davis-Yin splitting}
  \acro{DPC}{dirty paper coding}
  \acro{DR}{Douglas-Rachford}
  \acro{DRA}{dynamic resource assignment}
  \acro{DSA}{dynamic spectrum access}
  \acro{eMBB}{enhanced mobile broadband}
  \acro{eV2X}{enhanced vehicle-to-everything}
  \acro{EADMM}{\ac{EVM}-constrained \ac{ADMM}}
  \acro{EIRP}{equivalent isotropically radiated power}
  \acro{EPOCS}{\ac{EVM}-constrained \ac{POCS}}
  \acro{ETSI}{European telecommunications standards institute}
  \acro{EVM}{error vector magnitude}
  \acro{EMSP}{EVM-constrained \ac{MSP}}
  \acro{ENSP}{EVM-constrained \ac{NSP}}
  \acro{ESP}{EVM-constrained spectral precoding}
  \acro{ESSP}{EVM-constrained \ac{SSP}}
  \acro{FD}{fully digital}
  \acro{FDD}{frequency division duplexing}
  \acro{GFBS}{generalized forward-backward splitting}
  \acro{HARQ}{hybrid automatic repeat request}
  \acro{HB}{hybrid beamforming}
  \acro{KKT}{Karush-Kuhn-Tucker}
  \acro{KPI}{key performance indicator}
  \acro{IFFT}{Inverse (Fast) Discrete Fourier Transform}
  \acro{InC}{in-coverage}
  \acro{IoT}{internet of things}
  \acro{ITS}{intelligent transportation systems}
  \acro{LDPC}{low-density parity-check code}
  \acro{LS}{large-scale}
  \acro{LS-MSP}{large-scale MSP}
  \acro{LTE}{long term evolution}
  \acro{MAC}{medium access control}
  \acro{MCS}{modulation and coding scheme}
  \acro{METIS}{Mobile Enablers for the Twenty-Twenty Information Society}
  \acro{MIMO}{multiple-input multiple-output}
  \acro{MIMO-OFDM}{\ac{MIMO}-\ac{OFDM}}
  \acro{MISO}{multiple-input single-output}
  \acro{MRC}{maximum ratio combining}
  \acro{MS}{mobile station}
  \acro{MSE}{mean squared error}
  \acro{MSP}{mask-compliant spectral precoder}
  \acro{MMSE}{minimum mean squared error}
  \acro{MTC}{machine type communications}
  \acro{MU-MIMO-OFDM}{multi-user \ac{MIMO}-\ac{OFDM}}
  \acro{mMTC}{massive machine type communications}
  \acro{cMTC}{critical machine type communications}
  \acro{NR}{New Radio}
  \acro{NSP}{notching spectral precoder}
  \acro{NSPS}{national security and public safety}
  \acro{NWC}{network coding}
  \acro{OFDM}{orthogonal frequency division multiplexing}
  \acro{OOB}{out-of-band}
  \acro{OOBE}{out-of-band emissions}
  \acro{OoC}{out-of-coverage}
  \acro{PDCCH}{physical downlink control channel}
  \acro{PDSCH}{physical downlink shared channel}
  \acro{POCS}{projection on convex sets}
  \acro{PPG}{proximal-proximal-gradient}
  \acro{PRB}{physical resource block}
  \acro{PSBCH}{physical sidelink broadcast channel}
  \acro{PSFCH}{physical sidelink feedback channel}
  \acro{PSCCH}{physical sidelink control channel}
  \acro{PSP}{proximal splitting based mask compliant spectral precoding}
  \acro{PSSCH}{physical sidelink shared channel}
  \acro{PHY}{physical}
  \acro{PLNC}{physical layer network coding}
  \acro{PRB}{physical resource block}
  \acro{PSD}{power spectral density}
  \acro{QAM}{quadrature amplitude modulation}
  \acro{QCQP}{quadratic constrained quadratic programming}
  \acro{QoS}{quality of service}
  \acro{QPSK}{quadrature-phase shift keying}
  \acro{PaC}{partial coverage}
  \acro{RAISES}{reallocation-based assignment for improved spectral efficiency and satisfaction}
  \acro{RAN}{radio access network}
  \acro{RA}{resource allocation}
  \acro{RAT}{radio access technology}
  \acro{RB}{resource block}
  \acro{RF}{radio frequency}
  \acro{RSRP}{reference signal received power}
  \acro{Rx}{receive}
  \acro{RX}{receive}
  \acro{RxEVM}{receive-EVM}
  \acro{Rx-TxEVM}{received \ac{TxEVM}}
  \acro{SC-FDM}{single carrier frequency division modulation}
  \acro{SEM}{spectral emission mask}
  \acro{SFBC}{space-frequency block coding}
  \acro{SCI}{sidelink control information}
  \acro{SIC}{successive interference cancellation}
  \acro{SINR}{signal-to-interference-plus-noise ratio}
  \acro{SISO}{single-input single-output}
  \acro{SL}{sidelink}
  \acro{SP}{Spectral precoding}
  \acro{SNR}{signal-to-noise ratio}
  \acro{SSP}{semi-analytical spectral precoding}
  \acro{STC}{space-time coding}
  \acro{SU}{single-user}
  \acro{SU-MIMO}{single-user \ac{MIMO}}
  \acro{SU-MIMO-OFDM}{single-user \ac{MIMO}-\ac{OFDM}}
  \acro{TDD}{time division duplexing}
  \acro{Tx}{transmit}
  \acro{TX}{transmit}
  \acro{TxEVM}{transmit-EVM}
  \acro{UE}{user equipment}
  \acro{UL}{uplink}
  \acro{UP}{uniform power}
  \acro{URLLC}{ultra-reliable and low latency communications}
  \acro{VUE}{vehicular user equipment}
  \acro{V2N}{vehicular-to-network}
  \acro{V2X}{vehicle-to-everything}
  \acro{V2V}{vehicle-to-vehicle}
  \acro{V2P}{vehicle-to-pedestrian}
  \acro{WMMSE}{weighted minimum mean squared error}
  \acro{wMMSE}{weighted minimum mean squared error}
  \acro{ZF}{zero forcing}
  \acro{ZMCSCG}{zero mean circularly symmetric complex Gaussian}
\end{acronym}

\newcommand{\squeezeup}{\vspace{-2.5mm}}
\usepackage{etoolbox}

\allowdisplaybreaks

\begin{document}

\title{EVM-Constrained and Mask-Compliant \\ MIMO-OFDM Spectral Precoding}
\author{Shashi~Kant,~\IEEEmembership{Student Member,~IEEE,} Mats~Bengtsson,~\IEEEmembership{Senior Member,~IEEE,}  Gabor~Fodor,~\IEEEmembership{Senior Member,~IEEE,} Bo~G\"oransson,~\IEEEmembership{Member,~IEEE,} and Carlo~Fischione,~\IEEEmembership{Senior~Member,~IEEE}
\thanks{S.\ Kant, B.\ G\"oransson, and G.\ Fodor are with Ericsson AB and KTH Royal Institute of Technology, Stockholm, Sweden (e-mail: \{shashi.v.kant, bo.goransson, gabor.fodor\}@ericsson.com)}
\thanks{M.\ Bengtsson and C.\ Fischione are with KTH Royal Institute of Technology, Stockholm, Sweden (e-mail: \{mats.bengtsson@ee.kth.se, carlofi@kth.se\})}
\thanks{The work of S.\ Kant was supported in part by the Swedish Foundation for Strategic Research under grant ID17-0114.}
\thanks{Part of this paper on EVM-unconstrained OFDM spectral precoding was presented at IEEE SPAWC 2019~\cite{KantSPAWC2019}.}
}

\maketitle
\copyrightnotice


\renewcommand\qedsymbol{$\blacksquare$}

\newcounter{subeqsave}
\newcommand{\savesubeqnumber}{\setcounter{subeqsave}{\value{equation}}%
\typeout{AAA\theequation.\theparentequation}}
\newcommand{\recallsubeqnumber}{%
  \setcounter{equation}{\value{subeqsave}}\stepcounter{equation}}

\iftrue
\renewcommand{\vec}[1]{\ensuremath{\boldsymbol{#1}}}
\newcommand{\mat}[1]{\ensuremath{\boldsymbol{#1}}}
\newcommand{\herm}{{\rm H}}
\newcommand{\tran}{{\rm T}}
\newcommand{\trans}{{\rm T}}
\newcommand{\trace}{{\rm Tr}}
\newcommand{\diag}{{\rm diag}}
\newcommand{\Diag}{{\rm Diag}}
\newcommand{\rank}{{\rm rank}}
\newcommand{\EVM}{{\rm EVM}}
\newcommand{\SNR}{{\rm SNR}}
\newcommand{\SINR}{{\rm SINR}}
\newcommand{\expect}{\mathbb{E}}
\newcommand{\Cm}{\mathbb{C}}
\newcommand{\Rm}{\mathbb{R}}
\newcommand{\CN}{\mathcal{CN}}
\newcommand{\be}{\begin{equation}}
\newcommand{\ee}{\end{equation}}
\newcommand{\pdf}{\mathcal{P}}
\newcommand{\prox}{\ensuremath{\boldsymbol{\rm{prox}}}}
\newcommand{\rprox}{\ensuremath{\boldsymbol{\operatorname{rprox}}}}
\newcommand{\proj}{\ensuremath{\boldsymbol{\rm{proj}}}}
\newcommand{\dom}{\rm{dom}}
\newcommand{\epi}{\rm{epi}}
\newcommand{\ie}{\textit{i.e.}}
\newcommand{\eg}{\textit{e.g.}}
\newcommand{\cf}{\textit{cf.}}
\newcommand{\avgTxEVM}{\ensuremath{\epsilon_{\rm avg}}}
\newcommand{\TxEVM}{\ensuremath{\bm{\epsilon}}}
\newcommand{\Ind}{\mathcal{X}} 

\newcommand{\vecOp}{\rm{vec}}
\newcommand{\unvecOp}{\rm{unvec}}

\newcommand{\NR}{\ensuremath{N_{\rm R}}}
\newcommand{\NT}{\ensuremath{N_{\rm T}}}
\newcommand{\NL}{\ensuremath{N_{\rm L}}}
\newcommand{\NCP}{\ensuremath{N_{\rm CP}}}
\newcommand{\Nsc}{\ensuremath{N_{\rm SC}}}
\newcommand{\RxEVM}{\ensuremath{\bm{\varsigma}}}

\newcommand{\A}{\ensuremath{\boldsymbol{A}}}
\renewcommand{\a}{\ensuremath{\boldsymbol{a}}}
\newcommand{\B}{\ensuremath{\boldsymbol{B}}}
\renewcommand{\b}{\ensuremath{\boldsymbol{b}}}
\newcommand{\C}{\ensuremath{\boldsymbol{C}}}
\renewcommand{\c}{\ensuremath{\boldsymbol{c}}}

\newcommand{\D}{\ensuremath{\boldsymbol{D}}}
\renewcommand{\d}{\ensuremath{\boldsymbol{d}}}
\newcommand{\E}{\ensuremath{\boldsymbol{E}}}

\newcommand{\F}{\ensuremath{\boldsymbol{F}}}

\newcommand{\K}{\ensuremath{\boldsymbol{K}}}

\newcommand{\M}{{\mathbf{M}}}
\newcommand{\I}{\ensuremath{\boldsymbol{I}}}

\newcommand{\R}{\ensuremath{\boldsymbol{R}}}

\renewcommand{\S}{\ensuremath{\boldsymbol{S}}}

\newcommand{\T}{\ensuremath{\boldsymbol{T}}}
\renewcommand{\t}{\ensuremath{\boldsymbol{t}}}

\newcommand{\U}{\ensuremath{\boldsymbol{U}}}
\renewcommand{\u}{\ensuremath{\boldsymbol{u}}}

\newcommand{\V}{\ensuremath{\boldsymbol{V}}}
\renewcommand{\v}{\ensuremath{\boldsymbol{v}}}

\newcommand{\X}{\ensuremath{\boldsymbol{X}}}

\newcommand{\Y}{\ensuremath{\boldsymbol{Y}}}

\newcommand{\0}{\ensuremath{\boldsymbol{0}}}
\fi

\def\baselinestretch{1}
\def\skred#1{\textcolor{red}{#1}}
\def\skblue#1{\textcolor{blue}{#1}}
\def\skteal#1{\textcolor{teal}{#1}}
\def\skblack#1{\textcolor{black}{#1}}
\def\gf#1{\textcolor{blue}{#1}}
\def\bg#1{\textcolor{green}{#1}}
\def\mb#1{\textcolor{magenta}{#1}}
\def\cf#1{\textcolor{teal}{#1}}
\newcommand{\mx}[1]{\mathbf{#1}}

\DeclareRobustCommand{\BigOh}{%
  \text{\usefont{OMS}{cmsy}{m}{n}O}%
}

\def \hfillx {\hspace*{-\textwidth} \hfill}

\algblock{ParFor}{EndParFor}
\algnewcommand\algorithmicparfor{\textbf{parfor}}
\algnewcommand\algorithmicpardo{\textbf{do}}
\algnewcommand\algorithmicendparfor{\textbf{end\ parfor}}
\algrenewtext{ParFor}[1]{\algorithmicparfor\ #1\ \algorithmicpardo}
\algrenewtext{EndParFor}{\algorithmicendparfor}

\begin{abstract} 
Spectral precoding is a promising technique to suppress out-of-band emissions and comply with {leakage constraints} over adjacent frequency channels and with mask requirements on the unwanted emissions.  However, spectral precoding may distort the original data vector, which is formally expressed as the error vector magnitude (EVM) between {the} precoded and original data vectors. Notably, EVM has a deleterious impact on the performance of  
multiple-input multiple-output orthogonal frequency division {multiplexing-based} systems. 
In this paper we propose a novel spectral precoding approach which constrains the EVM while complying with the mask requirements. 
We first formulate and solve the EVM-unconstrained mask-compliant spectral precoding problem,
which serves as a springboard to the design of two EVM-constrained spectral precoding schemes.
The first scheme takes into account a wideband EVM-constraint which limits the average in-band
distortion. 
The second scheme takes into account frequency-selective EVM-constraints, and consequently,
limits the signal distortion at the subcarrier level.
Numerical examples illustrate that both proposed schemes outperform previously developed schemes in 
terms of important performance indicators such as block error rate and system-wide throughput
while complying with spectral mask and EVM constraints.


\IEEEpeerreviewmaketitle

\end{abstract}

\squeezeup
\section{Introduction}

Modern wireless communication {systems}, including fifth-generation (5G) {\ac{NR}}, {adopt} orthogonal frequency division multiplexing \acused{OFDM}(\ac{OFDM}) {with cyclic prefix} \cite{Dahlman5GNR2018}. 
The reasons are that \ac{OFDM} has several attractive characteristics such as robustness to the negative effects of time dispersive channels and multi-path fading, simplicity in terms of equalization 
and flexibility in terms of supporting both low and high symbol rates and thereby supporting
a variety of quality of service requirements. Further, \ac{OFDM}-based systems  can facilitate dynamic spectrum~sharing~\cite{Setoodeh:17}. 

Unfortunately, \ac{OFDM} suffers from high \ac{OOBE} 
due to the
discontinuities at the boundaries of the rectangular window and the high sidelobes
associated with the sinc functions  of the \ac{OFDM} signal \cite{Prasad:04}. 
The \ac{OOBE} must be adequately suppressed since high \ac{OOBE} causes significant interference to the neighbouring adjacent channels. 
In practice, all \ac{OFDM} systems are designed to 
comply not only with \ac{OOBE} requirements, in terms of \ac{ACLR} and spectral emission mask, but also in-band requirements, in terms of \ac{EVM} and other signal demodulation/detection requirements~\cite{3GPPTS38.1042018NRReception}. {Simply, the \ac{EVM} describes the distortion/noise incurred to the useful (transmit) symbols and \ac{ACLR} represents the amount of undesired power that exists in the neighbouring carriers relative to the desired carrier power---the detailed description of these metrics are in the subsequent~sections.}

\iftrue
There are a plethora of techniques to suppress/reduce \ac{OOBE} for the cyclic-prefix based \ac{OFDM}---see, \eg, \cite{Huang_OOBE_unified:2015} and its references---which can be categorized into time and frequency domain methods. Amongst them, the methods are, namely, guard band inclusion, filtering \cite{Faulkner:2000}, windowing \cite{Bala_windowing:2013}, cancellation carriers \cite{Yamaguchi:2004, Brandes:2005, extended_AIC__Qu:2010}, and spectral precoding \cite{Cosovic:2006, DeBeek2009SculptingPrecoder,Chen_spectral_precoding:2011,Tom2013MaskShaping,Kumar2015WeightedRadio,Mohamad:18,KantSPAWC2019,Mohamad_PhDThesis2019}.

{Spectral precoding is one of the promising {bandwidth efficient} techniques for \ac{OOBE} reduction.} It spectrally precodes the data symbols before \ac{OFDM} modulation~\cite{DeBeek2009SculptingPrecoder,Mohamad:18,Mohamad_PhDThesis2019}, which reduces the \ac{OOBE}
without {using the extra spectral resources contrary to cancellation carriers and without} increasing the delay/time dispersion or penalizing the cyclic prefix of the transmitted signal {unlike filtering/windowing. In Section~\ref{sec:related_works_and_contributions} we review related works in addition to our contribution.}

In this paper, we develop spectral precoding schemes that comply with \ac{EVM}-constraints
and simultaneously meet \ac{OOBE} requirements in terms of mask and \ac{ACLR}, without affecting signal processing at the receiver. 
Our objective is to design low complexity spectral precoding algorithms that operate well 
in wide-band \ac{MIMO}-\ac{OFDM} systems. In the first part of the paper, we consider the problem of designing mask-compliant spectral precoding
without \ac{EVM} constraints. {In the second part of the paper, we propose to incorporate wideband and frequency-selective \ac{EVM} constraints in addition to mask compliance.} 
{Our goal is to improve key performance indicators, including
out-of-band and in-band performance metrics, because highly spectral efficient \ac{MIMO}-\ac{OFDM}-based systems require low \ac{EVM}.}
Accordingly, we introduce a wideband \ac{EVM} constraint, which 
restricts the wideband average in-band distortion due to spectral precoding.
Finally, the     
frequency-selective \ac{EVM} constraint offers more flexibility---{but at the expense of increased complexity}---than using wideband constraints, 
because the frequency-selective constraints limit the in-band distortion power at the subcarrier (or group of subcarriers) level.
As we will show and discuss, these design approaches represent {at least} three distinct ways of addressing the complexity-performance
trade-off in the design of \ac{MIMO}-\ac{OFDM} spectral precoders.

{The rest of the paper is structured as follows. The next section{, which can be skipped by the reader familiar with the subject,} discusses related work {including impact of \ac{EVM}} and states our contribution. 
Section~\ref{sec:prelim} discusses preliminaries and introduce the system model. Next, Section~\ref{sec:evm_unconstrained_fom_algos} formulates the \ac{EVM}-unconstrained problem and develops an \ac{ADMM}-based algorithm \cite{Combettes2011, Boyd2011} and a coordinate-descent-based algorithm to solve this unconstrained problem. Section~\ref{sec:tx_evm_constrained__csi_unaware__ls_msp__mimo_ofdm} formulates the \ac{EVM}-constrained problem and develops an \ac{ADMM}-based solution and yet an alternative approach based on the Douglas-Rachford algorithm \cite{Eckstein_DRS:92,Combettes2011}. Section \ref{sec:simulation_results} presents 
simulation results that compare the performance of the proposed schemes with that of benchmarking schemes, and finally Section~\ref{sec:conclusion_future_work} concludes the~paper.}   
\section{{Related Works} and Contributions} \label{sec:related_works_and_contributions}
{In this section, we first give an overview of the related frequency domain techniques, then the impact of \ac{EVM} on the system performance, and finally highlight our original contribution.  } 

{
{\subsection{Cancellation Carrier Techniques}}}
{
The cancellation carrier techniques, such as active interference cancellation \cite{Yamaguchi:2004}, cancellation carrier with power constraint \cite{Brandes:2005}, and extended active interference cancellation \cite{extended_AIC__Qu:2010}, utilize the non-data bearing subcarriers solely for \ac{OOBE} reduction. These methods generally offer good sidelobe suppression with relatively low complexity, at the expense of transmit power wastage and spectral resources. Besides \ac{SNR} degradation due to transmit power sharing with the non-data bearing cancellation carriers, the extended active interference cancellation \cite{extended_AIC__Qu:2010} induces intersymbol and intercarrier interferences to the desired/useful signal due to the placement of the non-orthogonal carriers. Additionally, there are several practical implementation challenges, \eg, many cancellation carriers are required to maintain appropriate \ac{PSD} of the composite signals to avoid the negative impact of intermodulation due to nonlinear components \cite{Pozar_rf_microwave:2011}, \eg, power amplifier, and consequently failing to meet other \ac{OOBE} requirements---in terms of mask and \ac{ACLR} defined in \ac{NR}-like standards~\cite{3GPPTS38.1042018NRReception}.   
}

{\subsection{Spectral Precoding Techniques}}
There are {many} variants of spectral precoding methods in the literature to suppress \ac{OOBE}. We briefly review some of these techniques subsequently. 

{\subsubsection{Unconstrained Spectral Precoding in \ac{SISO} Systems}}
In the notching spectral precoder \cite{DeBeek2009SculptingPrecoder}, the spectrally precoded signal essentially nulls/notches the \ac{OOBE} at given discrete frequencies in the out-of-band spectrum of the \ac{OFDM} signal. The precoder is obtained in the closed-form by solving equality constrained least-squares optimization problem. Notching the spectrum at well-selected discrete frequencies often suppresses the whole signal spectrum and consequently leads to \ac{OOBE} reduction~\cite{DeBeek2009SculptingPrecoder}. However, such a notching approach has a deleterious impact on the in-band performance, in terms of increased \ac{EVM} and   block  error  rate, which thereby penalizes the system-wide throughput. Notably, the edge subcarriers have very high \ac{EVM}---see Fig.~\ref{fig:fig1__txevm_unconstrained_aclrevm_vs_iter__psd_evm_final_iter}\subref{fig:fig4_evm_in_percent_vs_prb_sem2_sem1_final_iter} in the numerical results Section~\ref{subsec:simulation_results}---which decreases the throughput. Recognizing these problems, in~\cite{Kumar2015WeightedRadio} a weighted-notching spectral precoder method is proposed to reduce the \ac{EVM} at the edge subcarriers by spreading the total distortion over all the allocated subcarriers. Another extension of {the} notching spectral precoder, proposed in \cite{Mohamad:16}, performs precoding jointly over several consecutive \ac{OFDM} symbols which improves the in-band performance slightly over the single \ac{OFDM} symbol-based precoding but penalizes the latency since a large number of \ac{OFDM} symbols have to be buffered.

\subsubsection{Constrained Spectral Precoding}
On the contrary, \ac{MSP} improves the in-band performance over notching spectral precoder while only meeting a target mask instead of creating nulls at the selected discrete frequencies at the expense of increased complexity~\cite{Kumar2016,KantSPAWC2019}. Because \ac{MSP} is posed as an inequality constrained convex optimization formulation, it typically does not yield a closed-form solution. Furthermore, the authors of \cite{Tom2013MaskShaping} suggest utilizing a generic optimization solver, which generally employs interior-point methods \cite{Boyd2004ConvexOptimization}.
Therefore, the authors in~\cite{Kumar2016,KantSPAWC2019},~propose computationally efficient schemes to obtain mask-compliant spectrally precoded data symbols for \ac{SISO}-\ac{OFDM} systems. The authors in \cite{VanDeBeek2009EVM-constrainedEmission}  propose three \ac{EVM}-constrained precoders for \ac{SISO} systems, whereby two precoders are developed heuristically, and one of them is posed as a convex optimization problem. {The optimization-based precoder is obtained by minimizing the $\ell_2$-norm of the \ac{OOBE} at a chosen set of discrete frequency points subject to the \ac{EVM} constraint---yielding no closed-form solution. The former precoders are ad~hoc, but can be seen as a scaled-form of the notching spectral precoder such that they meet the \ac{EVM} constraint but penalize the \ac{OOBE} in an unsystematic way.} Furthermore, the authors show numerically that the ad~hoc precoders render superior performance compared with the optimization-based precoder in suppressing \ac{OOBE} under \ac{EVM} constraint. 

\subsubsection{Spectral Precoding in \ac{MIMO} Systems}
In \cite{Pitaval:17b} several linear receivers are investigated when notching spectral precoders is employed in a \ac{MIMO}-\ac{OFDM} system. In \cite[Paper~F]{Mohamad_PhDThesis2019}, the authors extend the \ac{SISO}-\ac{OFDM} spectral precoding to massive MIMO-OFDM for a joint spatial and notching spectral precoding, which exploits full downlink channel state information at the transmitter due to channel reciprocity in time division duplexing, to improve the in-band performance at the receiver. {Recent works analyzed the impact of the \ac{EVM}
at the receiver or over-the-air specifically in massive \ac{MIMO} context. 
These works reveal that in beamforming systems 
the \ac{EVM} is also beamformed along the same direction as the data, see \eg,~\cite{Moghadam_distortion_corr_mimo:12,Sienkiewicz_spatially_dependent:14,Larsson:18,Moghadam:18}.}   

\subsection{{Impact of \ac{EVM} on the System Performance}}
{
In practical transceivers, the hardware imperfections stem from both the digital and the analogue components. 
These may include clip noise, filter ripple or distortion, in-phase and quadrature (IQ) mismatch, non-linearity, local oscillator inducing phase noise, sampling clock offsets, timing and frequency error, see, \eg, \cite{Fettweis_dirty_rf:2005, Schenk:2008}. 
These imperfections have detrimental impact on the \ac{MIMO}-\ac{OFDM} performance, notably in high data rate achieving systems \cite{Goransson:2008,Studer_mimo_residual_txRF:2010, Studer_evm_system_level_mimo:2011,  Emil_maMIMO_book:2017, E_Bjornson__maMimo__non_ideal_hw:2014, Boulogeorgos:2019}. 
In practice, \ac{EVM} metric evaluation provides insightful and useful information on the link quality, in terms of \ac{SNR}, seen at the receiver due to the aggregated digital and analogue hardware imperfections~\cite{Hassun:1997, Mahmoud:2009, Tobias:2013}. {The \ac{NR}-like standards mandate (transmit) \ac{EVM} requirements, {see Table~\ref{table:nr_wbevm_requirements},} on the composite inband distortion emanating from the various sources but measured at the (ideal) receiver by using pilots or reference signals as described in, \eg,~\cite[Appendix~B]{3GPPTS38.1042018NRReception}.}
}

{
To minimize the impact of \ac{EVM} onto the system throughput performance, in general, there are at least two approaches: 
1) constraining the \ac{EVM} at the transmitter and 
2) mitigating or cancelling the (transmit) \ac{EVM} seen at the receiver.   
We pursue the former approach because \ac{NR}-like standards \cite{3GPPTS38.1042018NRReception, 3GPPTS38.101-4_NRUEdemod} stipulate the minimum \ac{EVM} requirements according to the considered modulation alphabet, which must be satisfied by the base station or user equipment devices.   
Hence, in this paper, we constrain the \ac{EVM} of spectral precoding in addition to mask compliance (which implicitly may ensure the fulfilment of minimum \ac{ACLR} requirements).
}

\subsection{Contribution of the Paper}
{In this paper, we design computationally efficient algorithms for the mask-compliant spectral precoding with (un)constrained \ac{EVM} for the \ac{MIMO}-\ac{OFDM}-based systems that scale linearly with the number of supported transmit antennas and do not require additional signal processing at the receiver. More specifically, }
\begin{itemize}
    \item {w}e propose and accomplish a wideband and frequency-selective \ac{EVM}-constrained and mask-compliant spectral precoding formulation. 
    \item {w}e develop two highly efficient algorithms by decomposing the large-scale spectral precoding optimization problems for both unconstrained and constrained \ac{EVM} into subproblems, where each subproblem yields closed-form or efficient solution, which has a very low computational complexity compared to the general purpose solver.  
    {We propose solutions for (un)constrained \ac{EVM}: 1) \ac{ADMM}-based algorithms, referred to as \ac{ADMM}/\ac{EADMM}  and 2) specialized algorithms, dubbed as \ac{SSP}/\ac{ESSP}. Note that part of \cite{KantSPAWC2019} is used for the \ac{EVM}-unconstrained part of the present paper}.
    \item {w}e finally present exhaustive simulations using a 5G NR (Release-15 compliant) {inhouse} link-level simulator \cite{KantSPAWC2019}. 
\end{itemize}

\vspace{-3.5mm}
\subsection{Notation}
\vspace{-1.5mm}
Let the set of complex  and real numbers be denoted by $\mathbb{C}$ and $\mathbb{R}$, respectively. $\Re\{x\}$ denotes the real part of a complex number $x$. The $i$-th element of a vector $\vec{a} \! \in \! \mathbb{C}^{m \times 1}$ is denoted by ${a}[i] \! \in \! \mathbb{C}$, and element in the $i$-th row and $j$-th column of the matrix $\mat{A} \! \in \! \mathbb{C}^{m \times n}$ is denoted by $\mat{A}\left[i,j\right] \! \in \! \mathbb{C}$. The $i$-th row and $j$-th column vector of a matrix $\mat{A} \in \mathbb{C}^{m \times n}$ are represented as $\mat{A}\left[i,:\right] \! \in \! \mathbb{C}^{1 \times n}$ and $\mat{A}\left[:, j\right] \! \in \! \mathbb{C}^{m \times 1}$, respectively. An $i$-th higher order vector and matrix are denoted  as $\vec{x}[i] \! \in \! \mathbb{C}^{m \times 1}$ or $\vec{x}_i \! \in \! \mathbb{C}^{m \times 1}$ and $\mat{X}[i] \! \in \! \mathbb{C}^{m \times n}$. We form a matrix by stacking the set of higher order vectors $\left\{ \vec{a}[n]\! \in \! \mathbb{C}^{M \times 1} \right\}_{n=1}^N$ and $\left\{ \vec{b}[m] \! \in \! \mathbb{C}^{1 \times N} \right\}_{m=1}^M$ column-wise and row-wise as {$\mat{A} \!= \! \left[ \vec{a}[1],\ldots,\vec{a}[N] \right] \! \in \! \mathbb{C}^{M \times N}$ and $\mat{B} \! = \! \left[ \vec{b}[1];\ldots;\vec{b}[M] \right] \! \in \! \mathbb{C}^{M \times N} $}, respectively. The dimensions of the vectors/matrices are equally applicable for both $\mathbb{C}$ and $\mathbb{R}$. The transpose and conjugate transpose of a vector or matrix are denoted by $\left(\cdot\right)^{\rm T}$ and $\left(\cdot\right)^{\herm}$, respectively. The complex conjugate is represented by $\left(\cdot\right)^*$. The $K \! \times \! K$ identity matrix is written as $\vec{I}_K$. The expectation operator is denoted by $\expect\{\cdot\}$. An $i$-th iterative update is denoted by $(\cdot)^{(i)}$.

\section{Preliminaries} \label{sec:prelim}
In this section, we introduce the downlink \ac{MIMO}-\ac{OFDM} system model followed by performance metrics useful for the spectral precoding design.

\subsection{System Model and Out-of-Band Emissions} \label{sec:system_model}

We consider the \ac{OFDM}-based single-user \ac{MIMO} downlink, where the base station is equipped with $\NT$ \ac{Tx} antennas, and the \ac{UE} is equipped with $\NR$ \ac{Rx} antennas as depicted in Fig. \ref{fig:block_diag_mimo_ofdm_tx_rx}. Additionally, we reckon spatial multiplexing transmission scheme with $\NL \leq \min\left\{\NT,\NR\right\}$ spatial layers. 
\begin{figure*}[tp!]
    \centering
    \scalebox{0.625}{\includegraphics[trim=0mm 0mm 0mm 0.1mm,clip]{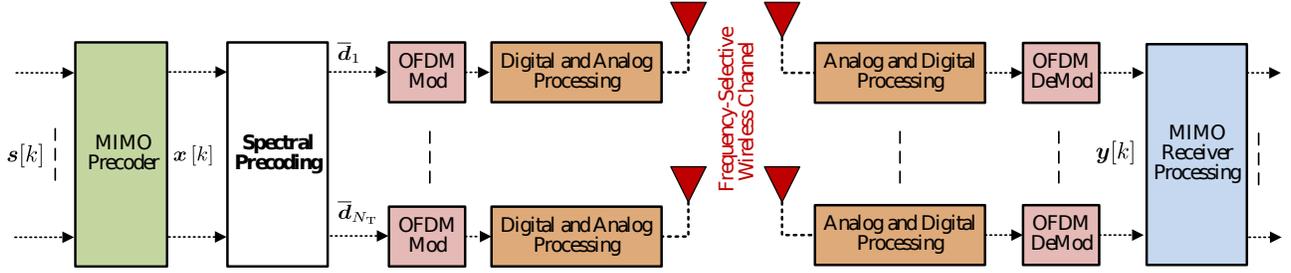}}
    \caption{\footnotesize Simplified block diagram of a single-user MIMO-OFDM transceiver with spectral precoding.}
    \label{fig:block_diag_mimo_ofdm_tx_rx}
    \vspace{-4mm}
\end{figure*}

A spatially precoded symbol vector ${\vec{x}}[k] $ at $k$-th subcarrier for a given \ac{OFDM} symbol can be formed by
\vspace{-1mm}
\be \label{eqn:general_tx_signal_model_per_re_without_distortion}
  {  {\vec{x}}[k]\, = \vec{W}[k] \vec{s}[k] \in \Cm^{\NT \times 1} }, 
  \vspace{-1mm}
\ee

\noindent where, the vector $\vec{s}[k] \in \mathcal{S}^{\NL\times 1}$ belongs to a complex-valued finite-alphabet set $\mathcal{S}$, \eg, corresponding to a $2^Q$-\ac{QAM} constellation with $Q\in\{2,4,6,8\}$. The given MIMO precoder $\vec{W}[k] \in \Cm^{\NT \times \NL}$ is chosen, \eg, from a codebook in 3GPP NR~\cite{3GPPTS38.214NR}. 

Towards this end, we introduce a frequency-domain (spatially-precoded) data matrix $\mat{X}$ by stacking each (column) vector $\vec{x}\left[k\right]$ (cf. Eq. \eqref{eqn:general_tx_signal_model_per_re_without_distortion}) column-wise for all the subcarriers within a given $N$-point inverse discrete Fourier transform such that 
\vspace{-1.5mm}
\be \label{eqn:tx_signal_freq_domain_matrix_X__without_distortion}
  \mat{X} \coloneqq \left[ \vec{x}\left[1\right], \ldots, \vec{x}\left[N\right] \right] \coloneqq \left[ \vec{d}_1^\trans; \ldots; \vec{d}_{\NT}^\trans \right] \in \Cm^{\NT \times N},  
  \vspace{-2mm}
\ee
\noindent where we define a spatially precoded vector $\vec{d}_j \!\in\! \Cm^{N \times 1}$ for $j$-th transmit antenna as
\vspace{-2.5mm}
\be \label{eqn:defnition_of_d_jth_ofdm_signal_in_freq_domain}
    \vec{d}_j \coloneqq \left(\mat{X}\left[j,: \right]\right)^\trans \in \Cm^{N \times 1} .
    \vspace{-2mm}
\ee
Similarly, we define a \textit{spectrally precoded} frequency-domain data matrix $\overline{\mat{X}}$  as
\vspace{-1.5mm}
\be \label{eqn:tx_signal_freq_domain_matrix_Xbar__without_distortion}
  \overline{\mat{X}} \coloneqq \left[  \overline{\vec{x}}\left[1\right], \ldots,  \overline{\vec{x}}\left[N\right] \right] \coloneqq \left[ \overline{\vec{d}}_1^\trans; \ldots; \overline{\vec{d}}_{\NT}^\trans \right]  \in \Cm^{\NT \times N},
  \vspace{-1.5mm}
\ee
where the spectrally (and spatially) precoded data symbol at $k$-th subcarrier is represented as $\overline{\vec{x}}\left[k\right]$. Further, we define a column vector $\overline{\vec{d}}_j \in \Cm^{N \times 1}$ of a spectrally precoded vector corresponding to $j$-th transmit antenna branch as
\vspace{-1.5mm}
\begin{align} \label{eqn:defnition_of_dbar_jth_ofdm_signal_in_freq_domain}
    \overline{\vec{d}}_j \coloneqq \left(\overline{\mat{X}}\left[j,: \right]\right)^\trans \in \Cm^{N \times 1} .
\end{align}

\subsection{Performance Metrics} \label{sec:performance_metrics}
{We utilize two figure-of-merits, \ie, \ac{OOBE} and inband distortions, for the spectral precoding design and its performance. 
The inband distortions are evaluated in terms of \ac{EVM} and also \ac{BLER} (or equivalently throughput). More specifically, we have employed \ac{EVM} metric for the spectral precoding design. Nonetheless, there exists a nonlinear mapping between \ac{EVM} and \ac{BLER} or equivalently received \ac{SNR}, see, \eg, \cite{Mahmoud:2009, Shafik_evm_ber_and_snr:2006}. Therefore, \ac{EVM} metric is sufficient for the spectral precoding design and performance evaluation.  
The \ac{OOBE} requirements are typically characterized by the spectral mask and \ac{ACLR}, among others \cite{3GPPTS38.1042018NRReception, 3GPPTS38.101-4_NRUEdemod}. Based on our practical experience, if the base station fulfils the spectral mask requirement with a suitable implementation margin then the \ac{ACLR} requirement is also achieved; and the converse is also true. Therefore, we have solely considered an appropriately discretized mask as a spectral precoding design parameter. However, for the numerical performance evaluation, we also evaluate \ac{ACLR} for completeness. }

\subsubsection{Out-of-Band Emissions} \label{sec:performance_metrics_oobe}
The \ac{OOBE} is typically quantified in terms of the operating band unwanted emissions and (conducted) \ac{ACLR}, whose definitions are given below. 

\begin{definition} [Operating band unwanted emissions {\cite[Section~6.6.4]{3GPPTS38.1042018NRReception}}, referred to as mask]
Unwanted emissions that are immediately outside the base station channel bandwidth resulting from the modulation process and non-linearity in the transmitter but excluding spurious emissions.
\end{definition}
The unwanted \ac{OOBE} due to the \ac{OFDM} frequency-domain signal $\vec{d}_j$ at the $M$ considered discrete frequency points $\vec{\nu}=\left[\nu_1,\ldots,\nu_{M}\right]$ can be described by $\vec{p}(\vec{\nu})\, = \mat{A}\vec{d}_j.$ 

We now define $\mat{A}[m,:] \coloneqq \vec{a}\left(\nu_m\right)^\trans \in \Cm^{1 \times N }$, 
where $\mat{A}\left[m,k\right] \coloneqq a(\nu_m,k)$ 
can be derived in discrete form as \cite{DeBeek2009SculptingPrecoder}: 
\vspace{-1.5mm}
\begin{align} 
\label{eqn:definition_of_a_nu_k_simplified_without_diric_function}
a(\nu_m,k) =&\left(\frac{1}{\sqrt{N}} \right) \exp\left(j\pi\frac{\left(\nu_m-k\right)}{N} \left(N_{\rm CP}-N+1\right)\right) \cdot \nonumber \\
		 &\cdot \frac{\sin{\left(\pi\frac{\left(\nu_m-k\right)}{N}\left(N+N_{\rm CP}\right)\right)}}{\sin{\left(\pi\frac{\left(\nu_m-k\right)}{N}\right)}} ,
\vspace{-1.5mm}
\end{align}
{where $N_{\rm CP}$ corresponds to cyclic prefix length in samples.}

\begin{definition} [Adjacent channel leakage ratio (ACLR) {\cite[Section~6.6.3]{3GPPTS38.1042018NRReception}}]
\ac{ACLR} is the ratio of the filtered mean power centred on the assigned channel frequency to the filtered mean power centred on an adjacent channel frequency{. The (worst-case) \ac{ACLR}  can mathematically be expressed as
\squeezeup
\begin{align*}
    {\rm \! ACLR} \!\coloneqq\! \frac{\!\int_{  -{\rm BW}/2}^{ {\rm BW}/2}  S_{\rm desired channel}\!\left( f \right) \: df \!}{\! \max\!\left\{\! \int\limits_{\frac{-3{\rm BW}}{2}}^{ \frac{-{\rm BW}}{2}} \! S_{\rm left channel}\!\left(\! f \right)  df \!, \!  \int\limits_{ \frac{{\rm BW}}{2}}^{ \frac{3{\rm BW}}{2} } \! S_{\rm right channel}\!\left(\! f \right) df \!\right\} \!}\!,
\end{align*}
where $S_{\rm desired channel}$ is the \ac{PSD} in the desired carrier having ${\rm BW}$ bandwidth including the guard band; and similarly, $S_{\rm left channel}$  and $S_{\rm right channel}$ correspond to the \ac{PSD} on the left and the right side of the desired carrier having same bandwidth ${\rm BW}$ as the desired carrier, respectively.}
\end{definition}

We would like to accentuate that we do not directly use \ac{ACLR} for the spectral precoding design, but rather we use the \ac{OOBE} power at the considered discrete frequency points.

\subsubsection{In-Band Distortion}
The considered in-band distortion for the spectral precoding design is \ac{EVM}, which can be quantified as a loss in the demodulated signal quality. It can be described mathematically per $j$-th transmit antenna as\footnote{\ac{EVM} measurements in NR standard are stipulated with a zero-forcing receiver, see, \eg, \cite[Annex~B]{3GPPTS38.1042018NRReception}, which can be seen as an equalized \ac{EVM} measure. On contrast, the \ac{EVM} for the spectral precoding design can be defined as an unequalized \ac{EVM} since the unequalized \ac{EVM} value would be conservative, \ie, it can not be less than the equalized \ac{EVM}, see, \eg, \cite{Karl_PhDThesis2019}.\label{footnote:evm_equalized_vs_unequalized}}
\begin{align}
{\rm EVM}_j \coloneqq \TxEVM_j = \frac{\expect{\left\{ \left\| \vec{d}_j - \overline{\vec{d}}_j \right\|_2^2  \right\}}}{\expect{\left\{ \left\| \vec{d}_j \right\|_2^2 \right\}}} .    
\end{align}

For the spectral precoding design, we assume that the average transmit signal power per transmit antenna branch $\expect{\left\{ \left\| \vec{d}_j \right\|_2^2 \right\}}$ is fixed.

\section{\ac{EVM}-Unconstrained \ac{LS-MSP} in MIMO-OFDM} \label{sec:evm_unconstrained_fom_algos}

In this section, we first consider a previously proposed mask-compliant spectral precoding without \ac{EVM} constraint that utilizes a generic convex optimization solver. It is known that it suffers from high computational complexity, notably in large-scale systems. Afterwards, to mitigate the complexity of computing the \ac{LS-MSP}, we propose a divide-and-conquer approach that breaks the original problem into smaller rank~1 quadratic-constraint problems, where each small problem yields a closed-form solution. Lastly, based on these solutions, we develop two specialized first-order low-complexity algorithms. In particular, the first one is based on the consensus \ac{ADMM} while the second one is derived by employing the coordinate descent scheme  of a dual variable and capitalizing on {closed-form} of the rank~1 constraint.

\squeezeup
\subsection{{EVM}-Unconstrained Proposed Problem Formulations}

The work in \cite{DeBeek2009SculptingPrecoder} was extended in \cite{Tom2013MaskShaping} for single-antenna \ac{OFDM}, referred to as \ac{MSP}, such that only the mask constraint needs to be fulfilled, \ie,
\vspace{-1mm}
\begin{align} \label{eqn:original_msp_optimization_problem}
\underset{\overline{\vec{d}}_j  \in \Cm^{N \times 1} }{\text{minimize}} 
\ \left\| \vec{d}_j - \overline{\vec{d}}_j \right\|_2^2 
\quad \text{subject to}
\ \left|\mat{A} \: \overline{\vec{d}}_j \right|^2 \preceq \vec{\gamma} , &
\end{align}
where the inequality constraint is element-wise and the target mask $\vec{\gamma} \in \mathbb{R}^{M \times 1} $ is given. The solution could not be expressed in (semi) closed-form and thereby the authors proposed to solve this \ac{MSP} problem via a generic quadratic programming solver~\cite{Tom2013MaskShaping}, \eg, CVX~\cite{Grant2014CVX:Beta}. 

Based on our key observation, we rewrite the problem \eqref{eqn:original_msp_optimization_problem} as described below, without any loss of convexity. More specifically, the constraint 
in \eqref{eqn:original_msp_optimization_problem} can be decomposed into 
$M$ rank~1 constraints such that it becomes
\vspace{-1mm}
\begin{alignat}{3}  \label{eqn:msp_problem_with_rank1_constraints}
&\underset{\overline{\vec{d}}_j}{{\text{minimize}}} 
\quad \left\| \vec{d}_j - \overline{\vec{d}}_j \right\|_2^2 \nonumber \\
& \text{{subject to}}
\quad \overline{\vec{d}}_j^\herm \; \overline{\mat{A}}_m \; \overline{\vec{d}}_j \leq \gamma_m,  \ \forall m \!=\!1\ldots,M,
\end{alignat}
\noindent where 
\begin{align} \label{eqn:definition_rank1_Am_matrix}
    { \overline{\mat{A}}_m = \vec{a}\left(\nu_m\right)^* \vec{a}\left(\nu_m\right)^\trans \in \Cm^{N \times N}  }
\end{align}
and $\rank\{\overline{\mat{A}}_m \} \!=\! 1$. Problem \eqref{eqn:msp_problem_with_rank1_constraints} is referred to as \ac{LS-MSP} that facilitates large-scale optimization.

\subsection{Efficient Algorithms for \ac{EVM}-Unconstrained \ac{LS-MSP}}

The proximal operator, $\prox$, is used for the algorithm design, {whose definition is given~below.}
\begin{definition}[proximal mapping~\cite{Parikh2013,Beck2017}] \label{definition:prox_operator}
Given a proper closed convex function $f: \dom_f \mapsto \left(-\infty\right., \left. +\infty \right]$, then the proximal mapping of $f$ is the operator given by 
\iftrue
\be
    \prox_{\lambda f}\left( \vec{x}\right) = \arg\min_{\vec{z} \in  \dom_f} \left\{ f(\vec{z}) + \frac{1}{2 \lambda} \left\|\vec{x} - \vec{z} \right\|_2^2\right\} 
    \vspace{-1mm}
\ee
\else
$\prox_{\lambda f}\left( \vec{x}\right) = \arg\min_{\vec{z} \in  \dom_f} \left\{ f(\vec{z}) + \frac{1}{2 \lambda} \left\|\vec{x} - \vec{z} \right\|_2^2\right\}$
\fi
for any $\vec{x} \in  \dom_f$, where $\dom_f$ corresponds to the {domain of function} $f$ and $\lambda > 0$.
\end{definition}

\begin{definition}[proximal mapping of the indicator function~\cite{Parikh2013,Beck2017}]
\label{definition:prox_operator_on_indicator_function}
Let $f: \dom_f \mapsto \left(-\infty\right., \left. +\infty \right]$ be an indicator function,  
$f(\vec{x}) \!\coloneqq\! \Ind_{\mathcal{C}}\left(\vec{x}\right)$ where $\mathcal{C}$ is a nonempty set $\Ind_{\mathcal{C}}\left(\vec{x}\right) \!=\! 0$ if $\vec{x} \! \in \! \mathcal{C}$ otherwise $\Ind_{\mathcal{C}}\left(\vec{x}\right) \!=\!  +\infty$,
then the proximal mapping of a given set $\mathcal{C}$ is an orthogonal projection operator {$\proj_{\mathcal{C}}$} onto the same set, \ie, 
\begin{align*}
    \prox_{\lambda \Ind_{\mathcal{C}}}\left( \vec{x}\right) 
    &= \arg\min_{\vec{z} \in  \dom_f} \left\{ \Ind_{\mathcal{C}}\left(\vec{x}\right) + \frac{1}{2 \lambda} \left\|\vec{x} - \vec{z} \right\|_2^2\right\} \\
    &= \arg\min_{\vec{z} \in  \mathcal{C}} \left\{\frac{1}{2} \left\|\vec{x} - \vec{z} \right\|_2^2\right\} = \proj_{\mathcal{C}}\left( \vec{x}\right).
\end{align*}
\end{definition}

\setlength{\textfloatsep}{3mm}
\begin{algorithm} [t]
\caption{\ac{ADMM}} \label{alg:consensus_admm_msp} 
 \begin{algorithmic}[1] 
 	\algrenewcommand\algorithmicrequire{\textbf{Inputs:}}
    \algrenewcommand\algorithmicensure{\textbf{Output(s):}}
 	\Require $\left\{ \vec{d}_j \in \mathbb{C}^{N \times 1} \right\}_{j=1}^{\NT}$, $\left\{ \gamma_m; \ \vec{a}\left(\nu_m\right) \right\}_{m=1}^M$,  {$\rho \in \mathbb{R}_{>0} $}
 	\Ensure  $\overline{\vec{d}}_j^{\left(I\right)} \in \mathbb{C}^{N \times 1} \quad \forall j=1,\ldots,\NT$ 
 \\ \textit{Initialization}: $\vec{y}_m^{(0)} = \vec{0}_{N \times 1}$ and $\vec{z}_m^{(0)} = \vec{0}_{N \times 1}$ 
  \For {$i = 1,2,\ldots,I$ } 
  \Statex
  \vspace{-4mm}
  \begin{subequations} \label{eqn:admm_algo_ver1}
  \begin{flalign} \label{eqn:psp_admm_algo_update_delta_d_ver1}
  	\overline{\vec{d}}_j^{\left(i\right)} \!&=\! \frac{1}{\left(1 \! + \! \rho M\right)} \left[ \vec{d}_j + \rho \sum_{m=1}^M \left(\vec{y}_m^{\left(i-1\right)} + \vec{z}_m^{\left(i-1\right)} \right)  \right] 
  	\end{flalign}
  	\Statex 
  	\vspace{-2mm}
    \ParFor {$m = 1,\ldots,M$ } \Comment{\% {run parallel}} 
  	\begin{flalign}
   \label{eqn:psp_admm_algo_update_delta_ym_ver1}
  	\vec{y}_m^{\left(i\right)} 
  	   &=  {\proj}_{\mathcal{C}_m}\left(\overline{\vec{d}}_j^{\left(i\right)} - \vec{z}_m^{\left(i-1\right)} \right) \\ 
  	\label{eqn:psp_admm_algo_update_delta_zm_ver1}   
  \vec{z}_m^{\left(i\right)}  &= \vec{z}_m^{\left(i-1\right)} + \vec{y}_m^{\left(i\right)} - \overline{\vec{d}}_j^{\left(i\right)}  
  \end{flalign} 
  \EndParFor
  \end{subequations}
  \EndFor
 \State \textbf{return}  $\overline{\vec{d}}_j^{(I)}$
\end{algorithmic} 
\end{algorithm}

If $M=1$, the orthogonal projection onto the rank~1 quadratic constraint is obtained in closed-form as described in the following theorem:

\begin{theorem}[projection onto the $\rank$ 1 quadratic constraint] \label{theorem:proj_rank1_quadratic_constraint}
Let $\mathcal{C} \subseteq \Cm^{N \times 1}$ and $\mathcal{C} \neq \emptyset$ given by $\mathcal{C} = \left\{\vec{x} \in  \Cm^{N \times 1}  \bm{:} \vec{x}^\herm {\widetilde{\mat{A}}} \vec{x} - b \leq  0 \right\}$, where {$ {\widetilde{\mat{A}}} \! = \! \vec{u} \vec{u}^\herm \in \Cm^{N \times N} $} is $\rank$ 1 matrix and $b \in \Rm_{\geq 0}$, then
the proximal operator 
\begin{align} 
\label{eqn:projection_operator_rank1_quadratic_constraint}
    &{\prox}_{\mathcal{X}_\mathcal{C}}\left( \vec{x} \right) = \proj_{\mathcal{C}} \left( \vec{x} \right)   \nonumber \\ 
    &\!=\! 
    \left\{ 
		 \begin{matrix}
		  \vec{x} \!+\! \left( \frac{\sqrt{b} \!-\! \left| \vec{u}^\herm \vec{x} \right|}{\left\| \vec{u} \right\|_2^2 \left| \vec{u}^\herm \vec{x} \right|} \right) \vec{u} \left( \vec{u}^\herm \vec{x} \right), \!       &  \text{if} \  \vec{x}^\herm {\widetilde{\mat{A}}} \vec{x} > b    \\
		  \vec{x},    \!    &  \text{if} \  \vec{x}^\herm {\widetilde{\mat{A}}} \vec{x} \leq b.    
		\end{matrix} 
	\right. 
\end{align}
\end{theorem}
\begin{proof}
See Appendix \ref{sec:proof_of_proj_onto_rank1_matrix}.
\end{proof}

If $M > 1$, then no closed-form is known yet. Hence, in the sequel, we propose low-complexity algorithms that essentially {break} down the \ac{LS-MSP} problem into smaller subproblems, where each subproblem admits closed-form solution capitalizing on Theorem \ref{theorem:proj_rank1_quadratic_constraint}. 

\subsubsection{\ac{EVM}-Unconstrained \ac{ADMM} 
\ac{LS-MSP} (referred to as \ac{ADMM})} \label{sec:admm_algorithm}

In our first proposal, we utilize \ac{ADMM} with consensus optimization to solve the \ac{LS-MSP} problem by rewriting Problem~\eqref{eqn:msp_problem_with_rank1_constraints}  
\vspace{-2.5mm}
\begin{alignat*}{3}
&\underset{\overline{\vec{d}}_j, \vec{y}_m \in \mathbb{C}^{N \times 1}}{\text{minimize}} 
\quad f\left(\overline{\vec{d}}_j\right) + \sum_{m=1}^M \Ind_{\mathcal{C}_m}\left(\vec{y}_m\right) \nonumber \\ 
&{\text{subject to}}
\quad \quad \vec{y}_m = \overline{\vec{d}}_j  \quad \forall m\! =\! 1,\ldots,M \! ,
\end{alignat*}
where the (non-equalized) squared \ac{EVM} $f\left(\overline{\vec{d}}_j\right) \coloneqq \left\| \vec{d}_j - \overline{\vec{d}}_j \right\|_2^2$ is a convex and differentiable function.  
The non-differentiable indicator function $\mathcal{X}_{\mathcal{C}_m}\left(\vec{y}_m\right)$  with the rank~1 constraint set is given by $\mathcal{C}_m = \left\{\vec{y}_m : \vec{y}_m^\herm \; \overline{\mat{A}}_m \; \vec{y}_m - \gamma_m \leq  0 \right\}$. 

Algorithm~\ref{alg:consensus_admm_msp} summarizes the proposed recipe for the \ac{ADMM}-based spectral precoding, cf.~\cite{KantSPAWC2019} for details, where $I$ denotes the total number of iterations. {The convergence analysis of the \ac{ADMM} algorithm is given in Appendix~\ref{sec:proof_of_consensus_admm}.}

\subsubsection{\ac{EVM}-Unconstrained \ac{SSP} \ac{LS-MSP} (referred to as \ac{SSP})}
In our second proposal, we derive an optimal semi-analytical algorithm, dubbed as \ac{SSP}, based on the \ac{KKT} conditions \cite{Boyd2004ConvexOptimization} for the constrained optimization \eqref{eqn:msp_problem_with_rank1_constraints}. 

We form {the} Lagrangian of \eqref{eqn:msp_problem_with_rank1_constraints} by introducing the Lagrange multipliers $\{\mu_m\}$ as follows:
\begin{equation} \label{eqn:lagrangian_for_ssp_algo}
L\left(\overline{\vec{d}}_j, \left\{\mu_m\right\}\right) 
\!=\! \left\| \vec{d}_j - \overline{\vec{d}}_j \right\|_2^2 \! + \! \sum \limits_{m=1}^{M} \mu_m \! \left(  \overline{\vec{d}}_j^{\herm} \overline{\mat{A}}_m  \overline{\vec{d}}_j - \gamma_m \right).
\end{equation}

\noindent Utilizing the \ac{KKT} conditions, the stationarity condition yields \eqref{eqn:primal_solution_ssp__d_bar} and the Lagrange multipliers $\{ \mu_m \}$ are obtained iteratively in a coordinate descent fashion \cite{Luo1992} as outlined in Algorithm~\ref{alg:solution_ssp_ver1}---see Appendix \ref{sec:derivation_of_ssp_algorithm} for a detailed derivation. 

\begin{algorithm}[t] 
\caption{\ac{SSP}} \label{alg:solution_ssp_ver1} 
 \begin{algorithmic}[1] 
 	\algrenewcommand\algorithmicrequire{\textbf{Inputs:}}
    \algrenewcommand\algorithmicensure{\textbf{Output(s):}}
 	\Require $\left\{ \vec{d}_j \in \mathbb{C}^{N \times 1} \right\}_{j=1}^{\NT}$, $\left\{\gamma_m \! \in \! \Rm; \ \lambda_1^m = \|\vec{a}\left(\nu_m\right)\|_2^2\right\}_{m=1}^{M}$.
 	\Ensure  $\overline{\vec{d}}_j^{\left(I\right)} \in \mathbb{C}^{N \times 1} \quad \forall j=1,\ldots,\NT$ 
 \\ \textit{Initialization}:\\  $\left\{\mu_{m} = \frac{1}{\lambda_1^m} \left(\left| \vec{a}\left(\nu_m\right)^{\tran} \vec{d}_j \right| \sqrt{\left( \frac{\lambda_1^m}{\gamma_m} \right)} - 1 \right)\right\}$ $\forall m \!=\! 1,\ldots,M$ \texttt{(cf. Lemma~\ref{lemma:mu_value_for_rank1_closed_form})}; $\phi = 0$.
  \For {${i = 1,\ldots,I}$} 
  \For {$m = 1,\ldots,M$}
  \begin{subequations} \label{eqn:ssp_algo_ver1}
  \Statex 
  \vspace{-2.5mm}
  \begin{flalign} 
    \label{eqn:ssp_Gm_inverse_sum_o_rank1}
    \mat{G}_{\backslash m}^{-1}  &= \left( \mat{I}_N + \sum \limits_{n=1; n \neq m}^{M} \mu_n \overline{\mat{A}}_n  \right)^{-1} \quad \texttt{(using Lemma \ref{lemma:matrix_inversion_sum_of_rank_one_matrices})} \\ \nonumber
    \alpha_1 &= \vec{a}\left(\nu_m\right)^{\trans} \mat{G}_{\backslash m}^{-1} \ \vec{d}_j; \   \alpha_2 = \vec{a}\left(\nu_m\right)^{\trans}  \mat{G}_{\backslash m}^{-1} \ \vec{a}\left(\nu_m\right)^{*} 
  \end{flalign} 
  \Statex 
  \vspace{-2.5mm}
  \begin{flalign}\label{eqn:ssp_mu_computation_step}   
   \mu_m &\leftarrow \Re \left\{ \frac{\alpha_1 \exp\left(-\iota \phi\right) - \sqrt{\gamma_m}}{\sqrt{\gamma_m} \ \alpha_2}  \right\}
   \savesubeqnumber
  \end{flalign}    
  \end{subequations}
  \EndFor
  \EndFor
 \State \textbf{return}   
 \vspace{-4mm}
  \addtocounter{equation}{-1}
 \begin{subequations}
 \recallsubeqnumber
\begin{flalign} \label{eqn:primal_solution_ssp__d_bar} \overline{\vec{d}}_j 
= \left( \mat{I}_N + \sum \limits_{m=1}^{M} \mu_m \  \overline{\mat{A}}_m \right)^{-1} \vec{d}_j \quad \texttt{(using Lemma \ref{lemma:matrix_inversion_sum_of_rank_one_matrices})} \end{flalign} 
\end{subequations}
\end{algorithmic} 
\vspace{-2mm}
\end{algorithm}

\begin{table*}[!htbp] 
\begin{center}
\caption{Comparison of online complexity of various \ac{EVM}-unconstrained and mask-compliant algorithms}\label{table:complexity_comparison}
\vspace*{-2.5mm}
{
\scalebox{0.95}{
\resizebox{\textwidth}{!}{%
	\begin{tabular}{|l||c|c|} \hline
	\textbf{Method}   & \textbf{Complexity: Real Multiplications} & \textbf{Complexity: Real Additions}\\ \hline
    \ac{MSP} \cite{Tom2013MaskShaping}   & $\BigOh\left(N^{4.5}\NT\right)$ \cite{Kumar2016} & --  \\ \hline
    \ac{ADMM}                            & {$I\left(9MN \!+\! M \!+\!N\!+\! 1\right)\NT $}   & {$I\left(14MN \!+\! 2N \!+\!M\!+\!1\right)\NT$} \\ \hline
    \ac{SSP}                            & $ \left( \left[5MN \!+\! I \!\left( 4M^2N^2 \!+\! 12MN^2 \!+\! 5M^2N \!+\! 3MN \!+\!M \right) \right]   \right) \NT  $  & $ \left( \left[2MN \!+\! I \!\left( 2M^2N^2 \!+\! 6MN^2 \!+\! 6M^2N \!-\! 2MN \!+\!M^2 \right) \right]   \right) \NT  $ \\ \hline
	\end{tabular}
}
}
}
\end{center}
\vspace*{-2mm}
\end{table*}

\vspace*{-3mm}
\subsection{Complexity Analysis} \label{sec:complexity_analysis_of_evm_unconstrained}
\vspace*{-1.5mm}
We analyze the run-time complexity, in terms of required {real-valued} multiplications {and real-valued additions} but ignore the offline complexity.  
Notice, {we convert all the complex multiplications and additions into equivalent real-valued multiplications and additions, \ie, $1$ complex multiplication is equivalent to $4$ real multiplications and $2$ real additions, and $1$ complex addition corresponds to $2$ real additions. Furthermore,} we assume that all the subcarriers are allocated, which will give the worst-case complexity analysis. Table~\ref{table:complexity_comparison} summarizes the complexity of the mask-compliant precoding schemes.

{
\textit{\ac{ADMM}}: The initialization step requires no multiplications/additions. In step \eqref{eqn:psp_admm_algo_update_delta_d_ver1}, there are {$\left(M\!+\!1\right)N$ complex additions, $N\!+\!1$ real multiplications, and $1$ real addition per iteration and transmit antenna}. The dominating online algebraic complexity is in the computation of the $\prox$ operator in step \eqref{eqn:psp_admm_algo_update_delta_ym_ver1}, which is in the order of {{$2MN$} complex multiplications, $2MN$ complex additions, $M\left(N\!+\!1\right)$ real multiplications, and $M$ real additions per iteration and transmit antenna}. 
However, due to distributed nature of consensus \ac{ADMM}, the $M$ subiterations can run in parallel per iteration cycle at the expense of increased memory requirements. {In step~\eqref{eqn:psp_admm_algo_update_delta_zm_ver1}, we need $2MN$ complex additions.}  {Thus}, ignoring parallelization, the total run-time complexities for $I$ iterations and $\NT$ transmit antennas 
{are $I\left(9MN \!+\! M \!+\!N\!+\! 1\right)\NT $ and $I\left(14MN \!+\! 2N \!+\!M\!+\!1\right)\NT$ in terms of real multiplications and real additions, respectively.}
}

{
\textit{\ac{SSP}}: The initialization step needs {$N$} complex {and real multiplications, respectively} for each $m$-th frequency point and antenna, where $\{\lambda_1^m\}$ can be computed offline.  
The main computational complexity of step \eqref{eqn:ssp_Gm_inverse_sum_o_rank1} is due to the matrix inversion, but no online inversion is necessary due to the sum of rank~1 matrices---see Lemma \ref{lemma:matrix_inversion_sum_of_rank_one_matrices}, which {are $\left(M\!-\!1\right)N^2 \!+\! NM$ complex multiplications, $2\left(M\!-\!1\right)N$ complex additions, $\left(M\!-\!1\right)N$ real multiplications, and $\left(M\!-\!1\right)$ real additions} for each frequency point, iteration, and transmit antenna. The complexity of the computation of both $\alpha_1$ and $\alpha_2$ are $\left(2N^2\right)$ for each frequency point, iteration, and transmit antenna. {The step~\eqref{eqn:ssp_mu_computation_step} need $N$ complex multiplications, $1$ real multiplication, and $1$ real addition.} Hence, the total online {real multiplications and real additions for $I$ iterations and $\NT$ transmit antennas are $\left( \left(\! \left[ 5M\!N \!+\! I \!\left( 4M^2\!N^2 \!+\! 12M\!N^2 \!+\! 5M^2\!N \!+\! 3M\!N \!+\!M \right) \! \right] \!  \right)\! \NT  \right)$ and $\!\left( \! \left( \! \left[2M\!N \!+\! I \!\left( 2\!M^2\!N^2 \!+\! 6\!M\!N^2 \!+\! 6\!M^2\!N \!-\! 2\!M\!N \!+\!M^2 \right) \!\right]  \! \right)\! \NT  \right)$, respectively.} 
}

\section{\ac{EVM}-Constrained \ac{LS-MSP} in MIMO-OFDM} \label{sec:tx_evm_constrained__csi_unaware__ls_msp__mimo_ofdm}
In this section, we firstly extend the \ac{LS-MSP} Problem \eqref{eqn:msp_problem_with_rank1_constraints} such that the spectrally precoded symbol yearns to keep the \ac{EVM} below the desired level by sacrificing \ac{OOBE} performance in terms of \ac{ACLR}, referred to as \ac{EMSP}. {Unfortunately, adding \ac{EVM} constraint in addition to the mask constraint poses challenges to develop a computationally efficient algorithm. Thanks to \ac{ADMM}, which offers a divide-and-conquer approach, incorporating \ac{EVM} constraint becomes easy and the algorithm is referred to as \ac{EADMM}. However, \ac{SSP} becomes prohibitively complex to support \ac{EVM} constraint. Therefore, we have proposed an alternative operator splitting framework based on the Douglas-Rachford algorithm~\cite{Combettes2011,Eckstein_DRS:92}, referred to as \ac{EVM}-constrained \ac{SSP} (\ac{ESSP}), that employs the iterative \ac{SSP} algorithm internally for the mask constraint and the outer loop of Douglas-Rachford supports the~\ac{EVM}~constraint}.

\subsection{{EVM}-Constrained Proposed Problem Formulations}
We pose wideband and frequency-selective \ac{EVM} constrained mask-compliant spectral precoding optimization problems and then develop low-complexity algorithms.
We firstly perform the epigraph transformation \cite{Chi_convex_sigcom_book:2017} of Problem \eqref{eqn:msp_problem_with_rank1_constraints}, without losing convexity, such that the proposed wideband \ac{EVM}-constrained optimization problem reads as:
\iftrue
\begin{subequations}  \label{eqn:wbevm_constrained_msp_problem_with_rank1_constraints__cvx} 
\begin{alignat}{3} 
 \label{eqn:evm_constrained_msp__common_delta_t__cost_function__cvx}
& \underset{\overline{\mat{X}} \in \Cm^{\NT \times N}, \Delta t \in \Rm }{\text{minimize}} 
& & \quad \Delta t  \\  \label{eqn:wbevm_constrained_msp__wbavgtxevm_constraint__cvx}
& \text{subject to}
& &  \left\|\overline{\mat{X}} -  {\mat{X}} \right\|_F \leq \avgTxEVM  \\
\label{eqn:wbevm_constrained_msp__msp_constraint__cvx}
& & & \left| \mat{A}  \overline{\mat{X}}^\trans \right|^{2} \preceq \bm{\Gamma} \odot \Delta t \bm{1}_{M \times \NT},
\end{alignat}
\end{subequations}
\fi
\noindent {where $\odot$ represents element-wise multiplication, and} $\avgTxEVM$ denotes the desired wideband averaged \ac{EVM} over all the allocated subcarriers (frequency-domain) and also over all the transmit antennas\footnote{It is straightforward to modify the problem to support wideband \ac{EVM}-constraint per transmit antenna branch.}. The target mask is {$\bm{\Gamma} \in \Rm^{M \times \NT}$}, \eg, $\bm{\Gamma}\left[:,j\right] = \left[\gamma_1;\ldots;\gamma_m \right]$ {and $\bm{1}_{M \times \NT}$ is an all-ones $M \times \NT$ matrix. Note that the mask constraints can be the same for all the transmit antenna branches.} {Furthermore, the constraint \eqref{eqn:wbevm_constrained_msp__msp_constraint__cvx} can be decomposed for each $j$-th transmit antenna and each $m$-th discrete frequency point such that the constraint can be read as $\overline{\vec{d}}_j^\herm \; \overline{\mat{A}}_m \; \overline{\vec{d}}_j \! \leq \! \Delta t \gamma_m$.}

To support more flexibility in terms of defining different \ac{EVM} constraints to different subcarriers, we extend the wideband \ac{EVM}-constrained Problem \eqref{eqn:wbevm_constrained_msp_problem_with_rank1_constraints__cvx} that offers different \ac{EVM} constraint for each subcarrier. A frequency-selective \ac{EVM}-constrained optimization   
problem is 
\vspace{-1.5mm}
\iftrue
\begin{subequations}  \label{eqn:fsevm_constrained_msp_problem_with_rank1_constraints__cvx} 
\begin{align} 
& \underset{\overline{\mat{X}} \in \Cm^{\NT \times N}, \Delta t \in \Rm }{\text{minimize}} 
& &  \Delta t \nonumber \\ 
\label{eqn:fsevm_constrained_msp_fstxevm_constraint__cvx}
& \text{subject to}
& & \left\|  \overline{\mat{X}}  \! \left[:, k\right] \! - \!   {\mat{X}}  \! \left[:, k\right] \right\|_2 \leq \TxEVM\left[ k\right]   \  \forall k \! \in \! \mathcal{T} \\
& & & \left| \mat{A}  \overline{\mat{X}}^\trans \right|^{2} \preceq \bm{\Gamma} \odot \Delta t \bm{1}_{M \times \NT}, \nonumber 
\end{align}
\end{subequations}
\fi
where $\TxEVM \left[ k\right]$ is the desired \ac{EVM} at $k$-th subcarrier and the set $\mathcal{T}$ denotes all activated subcarriers.

The main benefit of Problem \eqref{eqn:fsevm_constrained_msp_problem_with_rank1_constraints__cvx} is that the \ac{EVM} constraint per subcarrier or the group of subcarriers can be defined by the upper layers depending on the channel link quality and/or the allocation of the data/pilots appropriately. In other words, such frequency-selective \ac{EVM} constraint may have a high threshold (or high allowable \ac{EVM}) for the lower supported modulation alphabet, \eg, for {QPSK} modulation{---cf. Table~\ref{table:nr_wbevm_requirements}}. In contrast, the \ac{EVM} constraint may have a low threshold for the higher modulation alphabet, \eg, for 64\ac{QAM}. 

An optimal solution to both problems \eqref{eqn:wbevm_constrained_msp_problem_with_rank1_constraints__cvx} and \eqref{eqn:fsevm_constrained_msp_problem_with_rank1_constraints__cvx} can be obtained via a general purpose optimization solver, \eg, CVX \cite{Grant2014CVX:Beta}. However, as motivated in the previous section, such general purpose algorithms employ interior-point-based methods whose complexity is {prohibitively high}. Hence, in order to develop efficient algorithms for the \ac{EVM}-constrained and mask-compliant spectral precoding problems, we now instead transform the problems \eqref{eqn:wbevm_constrained_msp_problem_with_rank1_constraints__cvx} and \eqref{eqn:fsevm_constrained_msp_problem_with_rank1_constraints__cvx} into the feasibility problems by omitting the $\Delta t$ variable, \ie, problems \eqref{eqn:wbevm_constrained_msp_feasibility_problem_with_rank1_constraints} and \eqref{eqn:persubcarrier_evm_constrained_msp_feasibility_problem_with_rank1_constraints}, respectively. {If $\Delta t \leq 1$ in problems \eqref{eqn:wbevm_constrained_msp_problem_with_rank1_constraints__cvx} and \eqref{eqn:fsevm_constrained_msp_problem_with_rank1_constraints__cvx}, then the respective problems~\eqref{eqn:wbevm_constrained_msp_feasibility_problem_with_rank1_constraints} and \eqref{eqn:persubcarrier_evm_constrained_msp_feasibility_problem_with_rank1_constraints} are feasible.} {In other words, we make an assumption that the problem is feasible or has at least one solution, which implies $\Delta t \leq 1$, then the  
mask constraint \eqref{eqn:wbevm_constrained_msp__msp_constraint__cvx} can be expressed as following:}
\vspace{-1mm}
{
\begin{equation*}
     \left| \mat{A}  \overline{\mat{X}}^\trans \right|^{2} \preceq \bm{\Gamma} \odot \Delta t \bm{1}_{M \times \NT} \overset{\Delta t \leq 1}{\preceq}  \bm{\Gamma}.
     \vspace{-1mm}
\end{equation*}
Consequently, we omit $\Delta t$ from the the mask constraint in \eqref{eqn:wbevm_constrained_msp_problem_with_rank1_constraints__cvx}/\eqref{eqn:fsevm_constrained_msp_problem_with_rank1_constraints__cvx} for the feasibility problem.
}
A wideband \ac{EVM} constraint 
Problem \eqref{eqn:wbevm_constrained_msp_problem_with_rank1_constraints__cvx} can be posed as the following feasibility problem:
\begin{subequations}  \label{eqn:wbevm_constrained_msp_feasibility_problem_with_rank1_constraints} 
\begin{align} 
& {\text{find}} 
& & \hspace{-1cm} \overline{\mat{X}}  \in \Cm^{\NT \times N }  \nonumber \\ \label{eqn:wbevm_constrained_msp__wbavgtxevm_constraint}
& \text{subject to} 
& & \hspace{-1cm} \left\|\overline{\mat{X}} -  {\mat{X}} \right\|_F \leq \avgTxEVM  \\
\label{eqn:wbevm_constrained_msp__msp_constraint}
& & & \hspace{-1cm} \left| \mat{A}  \overline{\mat{X}}^\trans \right|^{2} \preceq \bm{\Gamma}.
\end{align}
\end{subequations}

Similarly, Problem \eqref{eqn:fsevm_constrained_msp_problem_with_rank1_constraints__cvx}, \ie, a frequency-selective \ac{EVM} constraint with mask-compliant problem, can be posed as the following feasibility problem:
\begin{subequations}  \label{eqn:persubcarrier_evm_constrained_msp_feasibility_problem_with_rank1_constraints} 
\begin{align} 
& {\text{find}} 
& & \overline{\mat{X}}  \in \Cm^{\NT \times N }  \nonumber \\ \label{eqn:persubcarrier_evm_constrained_msp__wbavgtxevm_constraint}
& \text{subject to}
& &  \left\|  \overline{\mat{X}}  \! \left[:, k\right] \! - \!   {\mat{X}}  \! \left[:, k\right] \right\|_2 \leq \TxEVM\left[ k\right]  \ \  \forall k \in \mathcal{T} \\
\nonumber 
& & & \left| \mat{A}  \overline{\mat{X}}^\trans \right|^{2} \preceq \bm{\Gamma}.
\end{align}
\end{subequations}

\noindent Now, the respective mask and {wideband} and frequency-selective \ac{EVM} constraint sets are defined as following. The $m$-th mask constraint set corresponds to the rank~1 quadratic inequality, \ie,
\be \label{eqn:ls_msp_constraint_set__mimo_X}
    \mathbfcal{C}_m \coloneqq \left\{\overline{\mat{X}} \bm{:} \overline{\vec{d}}_j^\herm {\overline{\mat{A}}_m}\overline{\vec{d}}_j - \gamma_m \leq  0; \forall j=1,\ldots,\NT \right\},
\ee
and the wideband \ac{EVM} constraint set can be described by
\be \label{eqn:wb_evm_constraint_set}
    \mathbfcal{E}_{\rm wb} \! \coloneqq \! \left\{\overline{\mat{X}} \bm{:} \left\|\overline{\mat{X}} -  {\mat{X}} \right\|_F - \avgTxEVM \leq 0 \right\},
\ee
whereas
the frequency-selective \ac{EVM} set can be expressed as
\be \label{eqn:fs_evm_constraint_set}
    \mathbfcal{E}_{\rm fs} \! \coloneqq \! \left\{\overline{\mat{X}} \! \bm{:} \! \left\|  \overline{\mat{X}}  \! \left[:, k\right] \! - \!   {\mat{X}}  \! \left[:, k\right] \right\|_2 - \TxEVM\left[ k\right] \! \leq \! 0; \ \forall k \in \mathcal{T}  \right\}.
\ee
We will denote the \ac{EVM} constraint set as $\mathbfcal{E}$, which can be wideband $\mathbfcal{E}_{\rm wb}$ and/or frequency-selective $\mathbfcal{E}_{\rm fs}$, unless stated otherwise.

Now, we rewrite the feasibility problems \eqref{eqn:wbevm_constrained_msp_feasibility_problem_with_rank1_constraints} and \eqref{eqn:persubcarrier_evm_constrained_msp_feasibility_problem_with_rank1_constraints} as the following unconstrained problem amenable to the latter proposed efficient algorithms 
\begin{align}  \label{eqn:generic_min_problem_with_sum_of_indicator_functions_evm_and_sum_of_rank1_qc}
 \underset{ \overline{\mat{X}}  \in \Cm^{\NT \times N }}{\text{minimize}} 
\quad \mathcal{F}\left( \overline{\mat{X}}  \right) \coloneqq \left\{ \mathcal{X}_{\mathbfcal{E}}\left( \overline{\mat{X}} \right) + \sum_{m=1}^M \mathcal{X}_{\mathbfcal{C}_m}\left( \overline{\mat{X}} \right) \right\},
\end{align}
where the composite function $\mathcal{F}\left( \overline{\mat{X}}  \right)$ is a sum of  non-differentiable indicator functions, \ie,  $\Ind_{\mathbfcal{E}}\left(\cdot\right)$ and $\Ind_{\mathbfcal{C}_m}\left(\cdot\right)$,  of constraint sets corresponding to the wideband or frequency-selective \ac{EVM}, and mask defined in \eqref{eqn:wb_evm_constraint_set} or \eqref{eqn:fs_evm_constraint_set} and \eqref{eqn:ls_msp_constraint_set__mimo_X}, respectively. Hence, we seek efficient methods to solve such a problem.

Prior to developing the algorithms for the constrained \ac{EVM}, we present the following theorems which are utilized for the {subsequent} algorithm development.
\begin{theorem} \label{theorem:separable_prox_operators}
    If a function $f\left( \mat{X} \right) \!=\! \sum_{k=1}^n f_i\left( \vec{x}_k \right)$ is separable across the variables column-wise $\mat{X} \!=\! \left[\vec{x}_1, \ldots, \vec{x}_n \right]$ or row-wise $\mat{X} \! =\! \left[\vec{x}_1; \ldots; \vec{x}_{\NT} \right]$, then the respective $\prox$ operators can be shown as
    $\prox_{f}\left( \mat{X} \right) \!=\! \left[ \prox_{f_1} \left( \vec{x}_1 \right), \ldots, \prox_{f_n} \left( \vec{x}_n \right)  \right]$ or $\prox_{f}\left( \mat{X} \right) \!=\! \big[ \prox_{f_1} \left( \vec{x}_1 \right); \ldots; \allowbreak \prox_{f_{\NT}} \left( \vec{x}_{\NT} \right)  \big]$.
\end{theorem}
\begin{proof}
Following the proximal operator Definition \ref{definition:prox_operator}, the minimization of the separable function is equivalent to minimization of respective functions $\{f_k\}$ independently~\cite{Beck2017,Parikh2013}.
\end{proof}

\begin{theorem}[projection onto Frobenius norm ball {\cite[Lemma~6.26]{Beck2017}}] \label{theorem:l2_ball_proj_operator}
    Let $E \subseteq \Cm^{p \times q}$ and $E \neq \emptyset$ be given by 
    $E \!\coloneqq \! \mathcal{B}_{\|\cdot\|_2}\left[\mat{C},  r \right] \!=\! \left\{ \mat{X} \!\in\!  \Cm^{p \times q}: \left\| \mat{X} \! -\! \mat{C} \right\|_F \!\leq\! r \right\}$,
    then the proximal or orthogonal projection operator for the Frobenius (or $\ell_2$) norm ball, \ie, $\mathcal{B}_{\|\cdot\|_F}\left[\mat{C},  r \right]$ with a given center $\mat{C}$ and the radius $r$, is
    \vspace{-1.5mm}
    \begin{align*}
        \prox_{ \lambda \Ind_{E} } \left( \mat{X} \right) 
        &= \!  \proj_{ \mathcal{B}_{\|\cdot\|_F}\left[\mat{C},  r \right]}  \left( \mat{X} \right) \nonumber \\ 
        &= \! \mat{C} \! + \!  \left( \! \frac{r}{\max\left\{  \left\| \mat{X} \!  - \!  \mat{C} \right\|_F, r \right\}} \! \right) \!  \left(  \mat{X} \!  - \!  \mat{C}  \right).
    \end{align*}
\end{theorem}

\subsection{Efficient Algorithms for \ac{EVM}-Constrained \ac{LS-MSP}}
In this section, we develop two computationally efficient algorithms to solve the aforementioned \ac{EVM}-constrained problem {\eqref{eqn:generic_min_problem_with_sum_of_indicator_functions_evm_and_sum_of_rank1_qc}}. 

\begin{table*} 
\begin{center}
\caption{Complexity comparison of various \ac{EVM}-constrained and mask-compliant algorithms}\label{table:complexity_comparison_evm_constrained}
\vspace*{-2.5mm}
{
\scalebox{1}{
\resizebox{\textwidth}{!}{%
	\begin{tabular}{|l||c|c|} \hline
	\textbf{Method}   & \textbf{Complexity: Real Multiplications} & \textbf{Complexity: Real Additions}\\ \hline
    \ac{EMSP}    & $\BigOh\left(N^{4.5}\NT\right)$  & --  \\ \hline
    \ac{EADMM}                            & {$\left(I\left(9MN \!+\! M \!+\!6N\!+\! 2\right)\NT \right)$}   & {$\left(I\left(14MN \!+\! 10N \!+\!M\!-\!1\right)\!\NT\right)$} \\ \hline
    \ac{ESSP}                            & $\left( I \left[\left(5M\!+\!7\right)N\!+\!1 \!+\! I^\prime \!\left( 4M^2N^2 \!+\! 12MN^2 \!+\! 5M^2N \!+\! 3MN \!+\!M \right) \right]   \NT  \right)$  & $\left( I \left[2MN\!+\!14N\!-\!2\!+\! I^\prime \!\left( 2M^2N^2 \!+\! 6MN^2 \!+\! 6M^2N \!-\! 2MN \!+\!M^2 \right) \right]   \NT  \right)$ \\ \hline
	\end{tabular}
}
}
}
\vspace*{-4mm}
\end{center}
\end{table*}

\subsubsection{\ac{EVM}-Constrained \ac{ADMM} \ac{LS-MSP} solution (referred to as {EADMM})}
We firstly express~\eqref{eqn:generic_min_problem_with_sum_of_indicator_functions_evm_and_sum_of_rank1_qc} amenable to \ac{ADMM}: 
\begin{align*}
&\underset{\overline{\mat{X}}, \overline{\mat{Y}}_m \in \mathbb{C}^{\NT \times N}}{\text{minimize}} 
\ \Ind_{\mathbfcal{E}}\left(\overline{\mat{X}}\right) \! + \! \sum_{m=1}^M \Ind_{\mathbfcal{C}_m}\left(\overline{\mat{Y}}_m\right) \nonumber \\
&{\text{subject to}}
\qquad \overline{\mat{Y}}_m \! = \! \overline{\mat{X}} \quad  \forall m \!=\! 1,\ldots,M \! ,
\end{align*}

\noindent The scaled-form consensus \ac{ADMM} for the above problem can be expressed as \cite{Combettes2011, Boyd2011}
\begin{subequations}
\begin{align}
    \label{eqn:scaled_consensus_eadmmm_step1}
    \overline{\mat{X}} &\leftarrow \! \arg  \underset{\overline{\mat{X}} }{ \min } \   \Ind_{\mathbfcal{E}}\left(\overline{\mat{X}}\right) \!+\! \rho \sum_{m=1}^M \left\| \overline{\mat{Y}}_m \!-\! \overline{\mat{X}} \!+\! \overline{\mat{Z}}_m   \right\|_F^2 \\
    \label{eqn:scaled_consensus_eadmmm_step2}
    \overline{\mat{Y}}_m &\leftarrow \!\arg \underset{ \overline{\mat{Y}}_m }{ \min } \ \mathcal{X}_{\mathbfcal{C}_m}\left( \overline{\mat{Y}}_m\right) \!+\! \rho \left\| \overline{\mat{Y}}_m \! - \! \overline{\mat{X}} \!+\! \overline{\mat{Z}}_m   \right\|_F^2 \ \forall m \\ 
    \label{eqn:scaled_consensus_eadmmm_step3}
    \overline{\mat{Z}}_m &\leftarrow \!  \overline{\mat{Z}}_m + \overline{\mat{Y}}_m - \overline{\mat{X}}  \quad  \forall m\!=\!1,\ldots,M \ .
\end{align}
\end{subequations}

In the first step of {our proposed \ac{EADMM} \ac{LS-MSP} algorithm}, taking the derivative with respect to $\overline{\mat{X}}$ and setting to zero {yields} 
$\overline{\mat{X}} \!=\! \prox_{\left(\nicefrac{1}{M}\right) \Ind_{\mathbfcal{E}} } \left( \frac{1}{M} \sum_{m=1}^M \left(  \overline{\mat{Y}}_m + \overline{\mat{Z}}_m \right) \right)$, \ie, the orthogonal projection onto the  wideband or frequency-selective \ac{EVM} constraint~\eqref{eqn:eadmm_algo_update_Xbar_ver1__proj_on_evm}
{---}$\ell_2$ norm ball (cf. Theorem~\ref{theorem:l2_ball_proj_operator}). The second step is an orthogonal projection onto the rank~1 quadratic constraint (cf. Theorem~\ref{theorem:proj_rank1_quadratic_constraint}) yielding \eqref{eqn:eadmm_algo_update_Ybarm_ver1}. Algorithm~\ref{alg:consensus_eadmm_msp} summarizes the proposed recipe for the \ac{EADMM} based mask-compliant spectral precoding. {The convergence analysis of the \ac{EADMM} algorithm is given in Appendix~\ref{sec:proof_of_consensus_admm}.}

\begin{algorithm}[t] 
\caption{\ac{EADMM}}\label{alg:consensus_eadmm_msp} 
 \begin{algorithmic}[1] 
 	\algrenewcommand\algorithmicrequire{\textbf{Inputs:}}
    \algrenewcommand\algorithmicensure{\textbf{Output(s):}}
 	\Require $\mat{X}$, $\left\{ \gamma_m; \ \vec{a}\left(\nu_m\right) \right\}_{m=1}^M$, and $\avgTxEVM \! \in \! \Rm$; $\left\{ \TxEVM[k] \! \in \! \Rm \right\}_{k=1}^N$
 	\Ensure  $\overline{\mat{X}}^{\left(I\right)} \in \mathbb{C}^{\NT \times N}$ 
 \\ \textit{Initialization}: $\overline{\mat{Y}}_m^{(0)} = \mat{0}_{\NT \times N}$ and $\overline{\mat{Z}}_m^{(0)} = \mat{0}_{\NT \times N}$ 
  \For {$i = 1,2,\ldots,I$ } 
  \vspace{-2mm}
  \begin{subequations} \label{eqn:eadmm_algo_ver1}
  \begin{flalign} 
    \mat{U} &=  \frac{1}{M}  \sum_{m=1}^M \left(\overline{\mat{Y}}_m^{\left(i-1\right)} + \overline{\mat{Z}}_m^{\left(i-1\right)} \right) \nonumber \\ %
    \label{eqn:eadmm_algo_update_Xbar_ver1__proj_on_evm}
  	\overline{\mat{X}}^{\left(i\right)} \!&=\! \proj_{\mathbfcal{E}} \left( \mat{U} \right) \equiv \texttt{{select}}
	\left\{ 
		 \begin{matrix}
		  \proj_{ \mathcal{B}_{\|\cdot\|_F}\left[\mat{X},  \avgTxEVM \right]}  \left( \mat{U} \right)  \\
		  \hspace{-4mm} \proj_{ \mathcal{B}_{\|\cdot\|_2}\left[\mat{X},  \TxEVM \right]}  \left( \mat{U} \right)
		\end{matrix} 
	\right.
  	\end{flalign}
  	\Statex 
    \ParFor {$m = 1,\ldots,M$ } \Comment{\% {run parallel}} 
  	\begin{flalign}
   \label{eqn:eadmm_algo_update_Ybarm_ver1}
  	\overline{\mat{Y}}_m^{\left(i\right)} 
  	   &=  {\proj}_{\mathcal{C}_m}\left(\overline{\mat{X}}^{\left(i\right)} - \overline{\mat{Z}}_m^{\left(i-1\right)} \right) \\ 
  \overline{\mat{Z}}_m^{\left(i\right)}  &= \overline{\mat{Z}}_m^{\left(i-1\right)} + \overline{\mat{Y}}_m^{\left(i\right)} - \overline{\mat{X}}^{\left(i\right)}  
  \end{flalign} 
  \EndParFor
  \end{subequations}
  \EndFor
 \State \textbf{return}  $\overline{\mat{X}}^{(I)}$
\end{algorithmic} 
\end{algorithm}

\setlength{\textfloatsep}{3mm}
\begin{algorithm}[t] 
\caption{{ESSP}}\label{alg:cruise_evm_ls_msp__sp_first} 
 \begin{algorithmic}[1] 
 	\algrenewcommand\algorithmicrequire{\textbf{Inputs:}}
    \algrenewcommand\algorithmicensure{\textbf{Output(s):}}
 	\Require $\mat{X}$, $\left\{ \gamma_m; \ \vec{a}\left(\nu_m\right) \right\}_{m=1}^M$, and $\left\{ \TxEVM \in \Rm^{\Nsc} \right\}$ or $\avgTxEVM \in \Rm$
 	\Ensure  $\overline{\mat{X}}^{\left(I\right)} \in \mathbb{C}^{\NT \times N}$ 
 \\ \textit{Initialization}: $\overline{\mat{X}}_m^{(0)} = \mat{X}$ and $\overline{\mat{Z}}_m^{(0)} = \mat{0}_{\NT \times N}$ 
  \For {$i = 1,2,\ldots,I$ } 
  \begin{subequations} \label{eqn:cruise_algo_ver2}
  \begin{flalign}
   \label{eqn:cruise_algo_update_Ybar_ver2}
  	\overline{\mat{Y}}^{\left(i\right)} 
  	   &= \! {\prox}_{\mathcal{X}_{\mathbfcal{C}}}\!\left( 2\overline{\mat{X}}^{\left(i-1\right)} \!-\! \overline{\mat{Z}}^{\left(i-1\right)} \right) \ \texttt{(solve \ac{SSP} Alg.~\ref{alg:solution_ssp_ver1})} \\
  	   \label{eqn:cruise_algo_update_Zbar_ver2}
  \overline{\mat{Z}}^{\left(i\right)}  &= \overline{\mat{Z}}^{\left(i-1\right)} +  \lambda_i \left( \overline{\mat{Y}}^{\left(i\right)} - \overline{\mat{X}}^{\left(i-1\right)}  \right)\\
   \label{eqn:cruise_algo_update_Xbar_ver2}
  	\overline{\mat{X}}^{\left(i\right)} \!&=\! \proj_{\mathbfcal{E}} \left( \overline{\mat{Z}}^{\left(i\right)} \right) \equiv \texttt{{select}}
	\left\{ 
		 \begin{matrix}
		  \proj_{ \mathcal{B}_{\|\cdot\|_F}\left[\mat{X},  \avgTxEVM \right]}  \left( \overline{\mat{Z}}^{\left(i\right)} \right)  \\
		  \hspace{-4mm} \proj_{ \mathcal{B}_{\|\cdot\|_2}\left[\mat{X},  \TxEVM \right]}  \left( \overline{\mat{Z}}^{\left(i\right)} \right)
		\end{matrix} 
	\right.
  	\end{flalign}
  \end{subequations}
  \EndFor
 \State \textbf{return}  $\overline{\mat{X}}^{(I)}$
\end{algorithmic} 
\end{algorithm}

\subsubsection{\ac{EVM}-Constrained \ac{SSP} \ac{LS-MSP} solution (referred to as \ac{ESSP})}

We layout the definition of Douglas-Rachford algorithm and subsequently propose the modifications to incorporate \ac{SSP} algorithm for the mask constraint.
\begin{theorem}[Douglas-Rachford algorithm] \label{thm:definition_of_dr_algorith__feasible_problem}
Consider the following problem 
\begin{equation}  \label{eqn:generic_min_problem_with_sum_of_two_functions}
\underset{ \overline{\mat{X}}  \in \Cm^{\NT \times N }}{\text{minimize}} 
\quad \mathcal{G}\left( \overline{\mat{X}}  \right) + \mathcal{H}\left( \overline{\mat{X}}  \right),
\end{equation}
where $\mathcal{G}$ and $\mathcal{H}$ are proper closed convex functions, and which has at least one solution. Consider $\tau \in \left(0, \infty\right)$ and a sequence of relaxation parameters $\lambda_i \in \left(0, 2 \right) \ \forall i \geq 0$ {and satisfy $\sum_i \lambda_i \left( 2 \!-\! \lambda_i\right) \!=\! + \infty $} with some arbitrary initial $ \overline{\mat{Z}}$, then the following iterative scheme 
\begin{subequations} \label{eqn:dr_generic_algo__two_operator}
\begin{align}
    \overline{\mat{X}} &\leftarrow \prox_{\tau \mathcal{G}}\left( \overline{\mat{Z}} \right)  \\
    \overline{\mat{Z}} &\leftarrow \overline{\mat{Z}} + \lambda_i \left( \prox_{\tau \mathcal{H}}\left( 2 \overline{\mat{X}} -   \overline{\mat{Z}} \right)   -  \overline{\mat{X}} \right)
\end{align} 
\end{subequations}
converges weakly to a solution to \eqref{eqn:generic_min_problem_with_sum_of_two_functions}.
\begin{proof}
See, \eg, \cite{Eckstein_DRS:92}{\cite[Corollary~5.2]{Combettes2004}}.
\end{proof}
\vspace{-3.5mm}
\end{theorem}

\noindent Strikingly, if the problem \eqref{eqn:generic_min_problem_with_sum_of_two_functions} is infeasible, then for some cases one could still find an approximate solution through Douglas-Rachford method---see, \eg, \cite{Liu2017,Ryu2019}.

The Douglas-Rachford algorithm is described for two functions. However, the problem \eqref{eqn:generic_min_problem_with_sum_of_indicator_functions_evm_and_sum_of_rank1_qc} at hand has more than two functions\footnote{One could reformulate Douglas-Rachford as consensus \ac{ADMM}. However, we would like to employ \ac{SSP} algorithm that offers a solution to the \ac{EVM}-constrained \ac{LS-MSP}.}. Thus, we {reformulate} \eqref{eqn:generic_min_problem_with_sum_of_indicator_functions_evm_and_sum_of_rank1_qc} as
\begin{align}  \label{eqn:cruise_problem_formulation_for_dr}
\underset{ \overline{\mat{X}}  \in \Cm^{\NT \times N }}{\text{minimize}} 
\quad \mathcal{X}_{\mathbfcal{E}}\left( \overline{\mat{X}} \right) + \mathcal{X}_{\mathbfcal{C}} \left( \overline{\mat{X}} \right) ,
\end{align}
where $\mathcal{X}_{\mathbfcal{C}}\left( \overline{\mat{X}} \right) \coloneqq \sum_{m=1}^M \mathcal{X}_{\mathbfcal{C}_m}\left( \overline{\mat{X}} \right)$ such that the two-operator Douglas-Rachford splitting can be employed. Since the proximal operator corresponding to the sum of $M$ indicator functions $\mathcal{X}_{\mathbfcal{C}} \left( \overline{\mat{X}} \right) = \sum_{m=1}^M \mathcal{X}_{\mathbfcal{C}_m}\left( \overline{\mat{X}} \right)$, \ie, $ {\prox}_{\mathcal{X}_{\mathbfcal{C}}}$ doesn't yield a closed-form, we approximate it by employing \ac{SSP}. We will show numerically that \ac{ESSP} framework also requires relatively less number of iterations compared to \ac{EADMM} to reach desired level of performance in terms of \ac{EVM} and \ac{ACLR} metrics at the cost of extra computational complexity compared to \ac{EADMM}, yet offering lower cost compared to the generic interior-point based solvers.

The proximal operator corresponding to the indicator function for \ac{EVM} constraint, either wideband $\mathcal{X}_{\mathbfcal{E}_{\rm wb}}\left( \overline{\mat{X}} \right)$ or frequency-selective constraint $\mathcal{X}_{\mathbfcal{E}_{\rm fs}}\left( \overline{\mat{X}} \right)$, is an orthogonal projection onto Euclidean norm ball (cf. Theorem \ref{theorem:l2_ball_proj_operator}). 

The \ac{ESSP} framework is summarized in Algorithm \ref{alg:cruise_evm_ls_msp__sp_first}, where we have performed a cyclic rotation of the Douglas-Rachford algorithm steps such that the proximal operator corresponding to the mask constraint occurs first in the given iteration cycle.

\squeezeup
\subsection{Complexity Analysis}
In this section, we analyze the worst-case run-time complexity, in terms of required {real-valued multiplications and real-valued additions but ignore the offline complexity as described in Section~\ref{sec:complexity_analysis_of_evm_unconstrained}---see the summary in Table \ref{table:complexity_comparison_evm_constrained}.
}

{
{\it \ac{EMSP}}: The computational complexity of solving the optimization problems \eqref{eqn:wbevm_constrained_msp_problem_with_rank1_constraints__cvx} and \eqref{eqn:fsevm_constrained_msp_problem_with_rank1_constraints__cvx}---almost similar to \ac{MSP} and using results in \cite{Kumar2016}---is  $\BigOh\left(N^{4.5}\NT \right)$.
}

{\it \ac{EADMM}}: The step \eqref{eqn:eadmm_algo_update_Xbar_ver1__proj_on_evm} is additional compared to \ac{ADMM} Algorithm \ref{alg:consensus_admm_msp}. So, the $\prox$ operator in the step \eqref{eqn:eadmm_algo_update_Xbar_ver1__proj_on_evm} corresponding to the \ac{EVM} constraint requires total complex multiplications $\left(IN\NT\right)$, {$\left(I\left(3N\!-\!1\right)\NT\right)$, and $N\!+\!1$ real multiplications}. Hence, ignoring parallelization, the total run-time {complexities are $\left(I\left(9MN \!+\! M \!+\!6N\!+\! 2\right)\NT \right)$ and $\left(I\left(14MN \!+\! 10N \!+\!M\!-\!1\right)\!\NT\right)$ in terms of real multiplications and real additions, respectively}.

{
{\it \ac{ESSP}}: {In step~\eqref{eqn:cruise_algo_update_Ybar_ver2}, there are $N$ complex multiplications and $N$ real multiplications besides the computational complexity of (iterative) $\prox$ operator \eqref{eqn:cruise_algo_update_Ybar_ver2} approximated by the \ac{SSP} algorithm. We need $2N$ complex additions and $N$ real multiplications in the step~\eqref{eqn:cruise_algo_update_Zbar_ver2}. Finally, in the $\prox$ operator \eqref{eqn:cruise_algo_update_Xbar_ver2} corresponding to the \ac{EVM} constraint same as in \ac{EADMM} need to be considered. Therefore, the total online complexities in terms of real multiplications and additions are given in Table \ref{table:complexity_comparison_evm_constrained}, where $I^\prime$ corresponds to inner iterations using \ac{SSP} Algorithm \ref{alg:solution_ssp_ver1}.} 
}

\begin{table} 
\begin{center}
\vspace{-2mm}
\skblack{
\caption{Simulation Parameters for FDD NR (Rel-15) PDSCH Type-A}\label{tasim}
\scalebox{0.5}{
\resizebox{\textwidth}{!}{%
	\begin{tabular}{|l||c|c|c|} \hline
	\textbf{Parameters}                                       & \textbf{Test $1$}                        & \textbf{Test $2$}        & \textbf{Test $3$}        \\ \hline
	\multicolumn{1}{|l||}{Subcarrier Spacing}                 & \multicolumn{3}{c|}{15 kHz}                                                                    \\ \hline
	\multicolumn{1}{|l||}{Carrier Bandwidth (PRB alloc.)}            & \multicolumn{3}{c|}{$5$ MHz ($25$ PRBs)}                                                       \\ \hline
	\multicolumn{1}{|l||}{{Carrier Spacing for ACLR}}            & \multicolumn{3}{c|}{{$5$ MHz  upper and lower adjacent channels \cite{3GPPTS38.1042018NRReception}}}                                                 \\ \hline
	\multicolumn{1}{|l||}{DL {SU-MIMO} $\NT, \NR$}             & $2{\rm Tx}, 2 {\rm Rx}$                  & $8{\rm Tx}, 2 {\rm Rx}$  & $2/8{\rm Tx}, 2 {\rm Rx}$ \\ \hline
	\multicolumn{1}{|l||}{ Spatial Layers ($\rank$)  }        & \multicolumn{2}{c|}{Fixed $\rank$ 1}     & \multicolumn{1}{c|}{adaptive (10\% BLER)}            \\ \hline
	\multicolumn{1}{|l||}{Spatial Precoding} & \multicolumn{3}{c|}{(codebook-based) adaptive (10\% BLER)}                                                       \\ \hline
	\multicolumn{1}{|l||}{Modulation}                         & \multicolumn{2}{c|}{64QAM}               & {adaptive (10\% BLER)}                               \\ \hline 
	\multicolumn{1}{|l||}{Code-rate}                          & \multicolumn{2}{c|}{$\nicefrac{1}{2}$  $\nicefrac{5}{6}$} &  {adaptive (10\% BLER)}             \\ \hline
	\multicolumn{1}{|l||}{Channel Model}              & \multicolumn{3}{c|}{TDL-A ($300$ns, $10$Hz) \& spatial correlation Low \cite{3GPPTS38.101-4_NRUEdemod} }   \\ \hline
	\multicolumn{1}{|l||}{Channel \& Noise power}     & \multicolumn{3}{c|}{Practical LMMSE based} \\ \hline
	\multicolumn{1}{|l||}{HARQ max transmissions}     & \multicolumn{3}{c|}{4 (3 max retransmissions with rv $\{0,2,3,1\}$) \cite{3GPPTS38.212NR}}  \\ \hline
    \multicolumn{1}{|l||}{Other Information}          & \multicolumn{3}{c|}{LDPC; LMMSE-IRC receiver; no other impairments}   \\ \hline	
	\end{tabular}
	}
}}
\end{center}
\end{table}

\section{Performance Evaluation} \label{sec:simulation_results}

In this section, we evaluate the performance of the proposed algorithms for mask-compliant spectral precoding that are both \ac{EVM}-unconstrained and \ac{EVM}-constrained utilizing a 5G NR (Rel-15) compliant {inhouse} link-level simulator. Moreover, we compare the performance of the proposed algorithms with the conventional spectral precoding algorithms, accordingly. 

\subsection{Performance Measures}
We analyze the spectral precoding performance in terms of two figure-of-merits, namely \ac{OOBE} and in-band distortions, in particular assuming base station supporting sub-6 GHz, \eg, frequency ranges between 410 MHz and 7.125 GHz---referred to as FR1 in 5G NR \cite[Section~5.1]{3GPPTS38.1042018NRReception}. 

\subsubsection{Out-of-Band Distortion}
As mentioned in Section \ref{sec:performance_metrics_oobe}, mask and (conducted) \ac{ACLR} are typically the performance metrics to quantify the operating band unwanted emissions.

In practical systems, the (digital) spectrum shaping is followed by other (non-linear) digital and analog processing as illustrated in Fig. \ref{fig:block_diag_mimo_ofdm_tx_rx}. Consequently, there is some spectral regrowth phenomenon due to such (non-linear) components in the transmitter after spectrum shaping. Thus, an implementation margin in terms of \ac{ACLR} and mask requirements are necessary to cope with spectral regrowth. Hence, \ac{OOBE} performance must render better performance than the (overall) stipulated mask in the standard due to spectrum shaping to account for the  margin.

{{We have considered \ac{ACLR} corresponding to the 1st adjacent carrier in both upper and lower frequencies, where the minimum requirement is 45 dB---worst-case of measured \ac{ACLR} in the upper and lower channels  \cite[Section~6.6.3]{3GPPTS38.1042018NRReception}}. It is worth highlighting that these \ac{ACLR} requirements are for the complete radio chain, \ie, measurements need to be performed at the antenna connector. Thus, spectrum shaping may have some aggressive mask and \ac{ACLR} requirements to meet the minimum requirements at the antenna connector.

\subsubsection{In-Band Distortion}
In these simulations, the in-band distortion is not only quantified in terms of  \ac{EVM} but also in terms of \ac{BLER} \cite{3GPPTS34.121_rel6} and {throughput}~\cite{3GPPTS38.306_rel15}. 

For our link simulations, we present normalized throughput, \ie, normalizing the throughput results by the maximum achievable throughput without any spectral precoding or \ac{OOBE} reduction and other hardware impairments or imperfections. 

\begin{table} 
\begin{center}
\vspace{-2mm}
\skblack{
\caption{\ac{EVM} Requirements \cite{3GPPTS38.1042018NRReception}}\label{table:nr_wbevm_requirements}
\scalebox{0.275}{
\resizebox{\textwidth}{!}{%
	\begin{tabular}{|l||c|} \hline
	\textbf{Modulation Scheme}     & \textbf{EVM Threshold}          \\ \hline
	\multicolumn{1}{|l||}{QPSK}    & 17.5 \%                        \\ \hline
	\multicolumn{1}{|l||}{16QAM}   & 12.5 \%                        \\ \hline
    \multicolumn{1}{|l||}{64QAM}   & 8.0 \%                         \\ \hline
	\end{tabular}
	}
} }
\end{center}
\vspace{-3mm}
\end{table}

\subsection{Simulation Parameters and Assumptions}
The key simulation parameters for the \ac{PDSCH} with \mbox{type-A}\footnote{These data types refer to different \ac{PDSCH} demodulation reference signals allocation~\cite{3GPPTS38.214NR}.} and the three investigated test scenarios are summarized in Table \ref{tasim}, see, \eg,~\cite{3GPPTS38.2112018NRModulation,3GPPTS38.1042018NRReception}, for the detailed NR physical layer and performance requirements. We have considered 15 kHz subcarrier spacing for the NR numerology unless otherwise mentioned. Furthermore, no supporting signals are transmitted besides \ac{PDSCH} along with the \ac{DMRS} for the practical channel and noise variance estimation at the \ac{UE} side. Note that for simulations purpose, we have considered relatively narrow 5 MHz channel bandwidth with 25 \ac{PRB}s allocation, even though the proposed methods can be employed for arbitrary bandwidths. Furthermore, we have employed a low spatial correlation model, but the complexities of the proposed algorithms are oblivious of spatial correlation---cf. Table~\ref{table:complexity_comparison} and~\ref{table:complexity_comparison_evm_constrained}. However, the inband performance may degrade with increasing correlation, notably for high spatial rank setup. {In Table~\ref{table:nr_wbevm_requirements}, we have provided 3GPP \ac{NR} wideband average \ac{EVM}\footnote{{The NR standard puts a requirement on the equalized \ac{EVM}, cf.~\cite[Section~6.5.2.2]{3GPPTS38.1042018NRReception}, \eg, for 64QAM, the transmitter is allowed to induce \ac{EVM} up to 8\%. As mentioned in \cref{footnote:evm_equalized_vs_unequalized} (\Cpageref{footnote:evm_equalized_vs_unequalized}), the unequalized \ac{EVM} is an upper bound for the equalized \ac{EVM}, \ie, unequalized \ac{EVM} is tougher requirement than equalized, see, \eg, \cite{Karl_PhDThesis2019}.}} requirements \cite{3GPPTS38.1042018NRReception} as a reference according to the considered modulation alphabet.}

\iftrue
\begin{figure*}[!htbp]
  \centering
  \begin{minipage}{.3125\linewidth}
    \centering
    \subcaptionbox{\footnotesize \ac{ACLR}-vs.-Iterations. \label{fig:fig1_aclr1_vs_iter_sem2_sem1}}
      {\includegraphics[width=\linewidth,trim=0mm 0mm 0mm 0mm,clip]{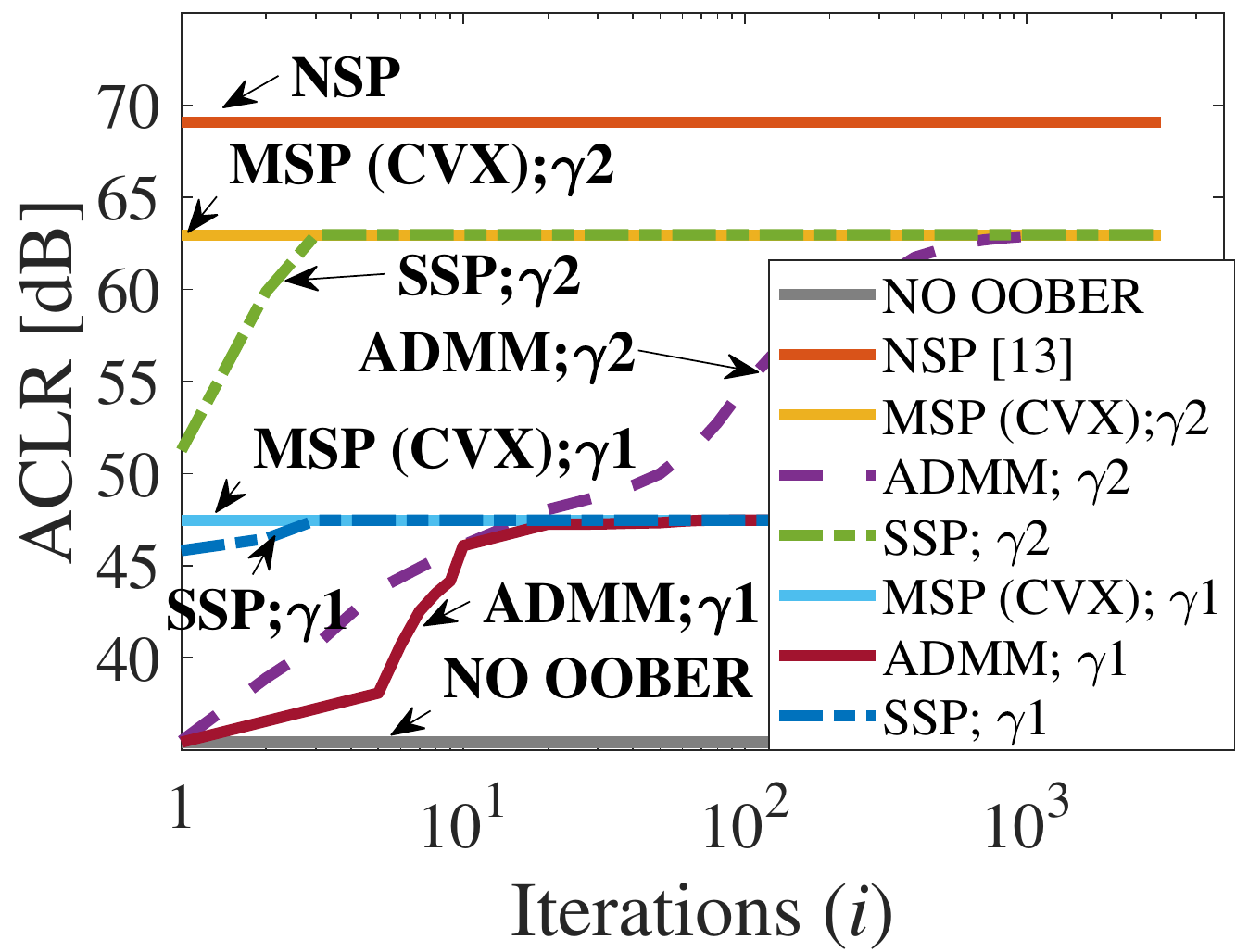}}
  \end{minipage}\quad
  \begin{minipage}{.3125\linewidth}
    \centering
    \subcaptionbox{\footnotesize {PSD} at the final iteration. \label{fig:fig3_psd_sem2_sem1_final_iter}}
      {\includegraphics[width=\linewidth,trim=0mm 0mm 0mm 0mm,clip]{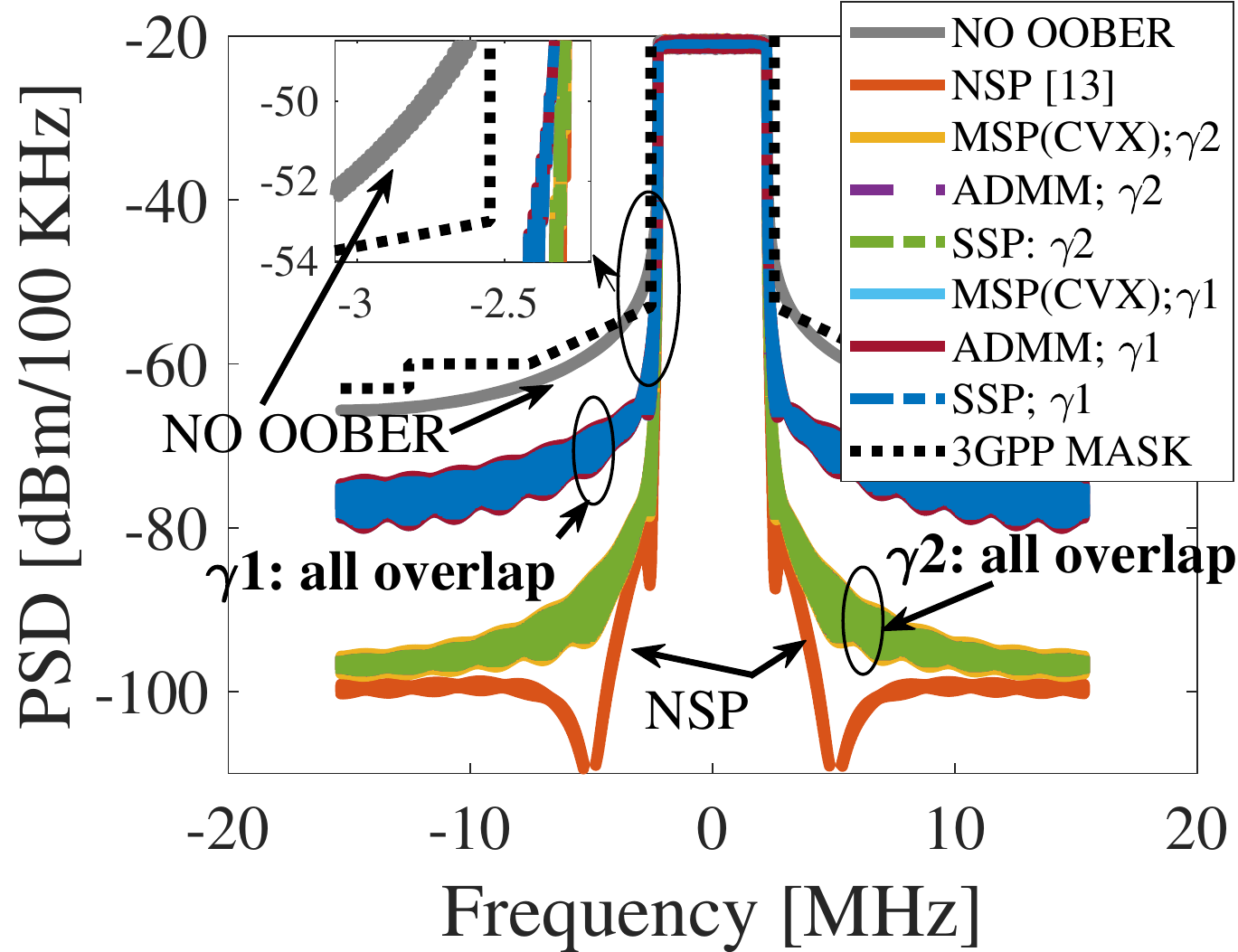}}
  \end{minipage}\quad
  \begin{minipage}{.3125\linewidth}
    \centering
    \subcaptionbox{\footnotesize \ac{EVM} [\%]~at~the~final~iteration. \label{fig:fig4_evm_in_percent_vs_prb_sem2_sem1_final_iter}}
      {\includegraphics[width=\linewidth,trim=0mm 0mm 0mm 0mm,clip]{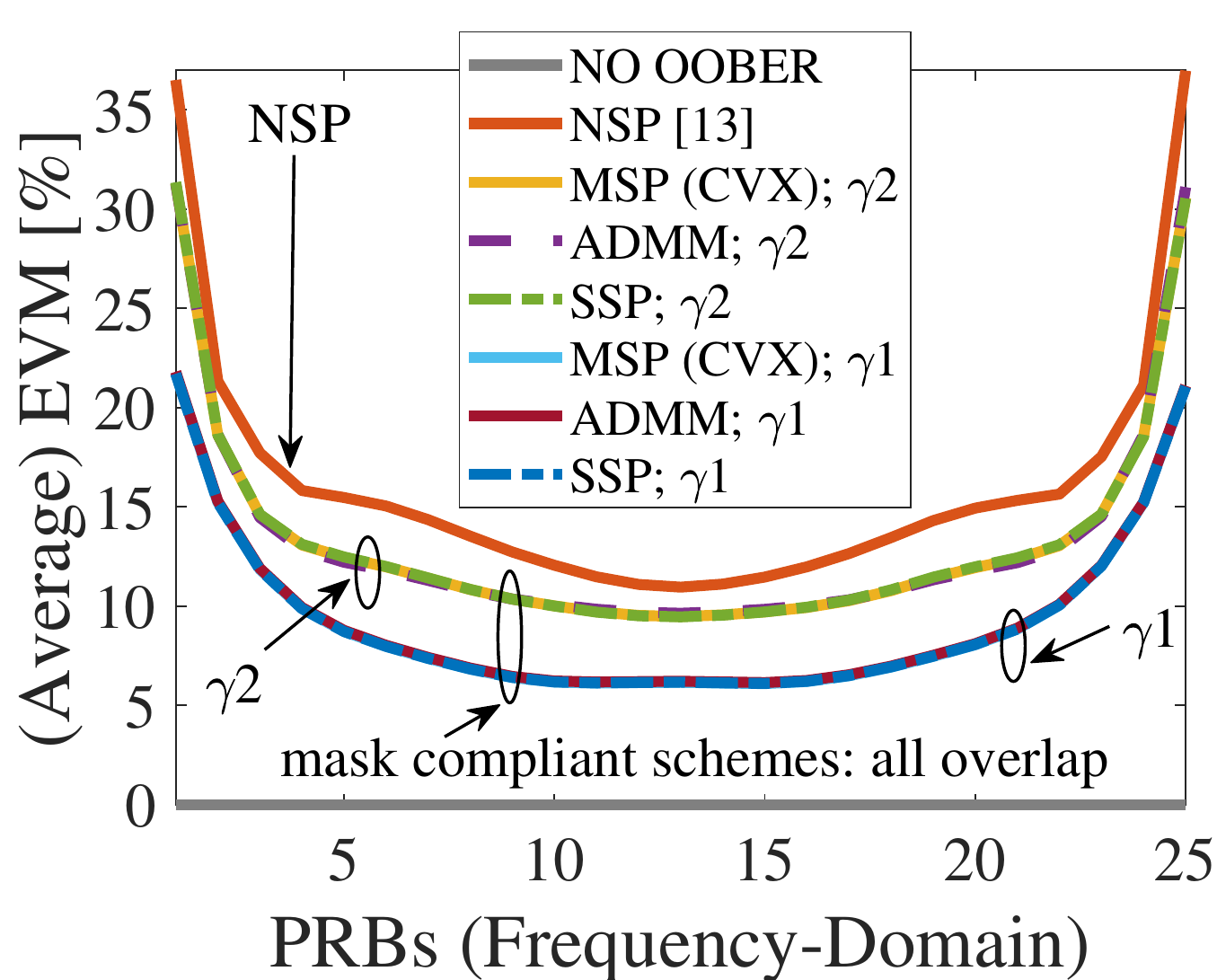}}
  \end{minipage}
  \vspace{-2mm}
  \caption{\footnotesize (a) Convergence behaviour of \ac{EVM}-unconstrained algorithms considering two different target masks, while (b) and (c) are achieved \ac{PSD} and \ac{EVM}  at the final iteration.}
  \label{fig:fig1__txevm_unconstrained_aclrevm_vs_iter__psd_evm_final_iter}
  \vspace{-5mm}
\end{figure*}
\fi

In addition to the parameters given in Table~\ref{tasim}, the discrete frequencies for mask compliant precoders are selected as $\nu \!\in\! \left\{ \mp 5010, \mp 4995, \mp 2565, \mp 2550 \right\}$ kHz, where the negative and positive frequencies correspond to the left and right side of the \ac{OOB} of the occupied signal spectrum, respectively. 
Notice that these discrete frequencies can be asymmetrically selected for the \ac{OOBE} suppression. The considered mask, referred to as \mbox{mask-1} is $\vec{\gamma}_{{\rm mask-1}} \!\coloneqq\! {\vec{\gamma}{1}} \!=\! \left[ -75, -75, -65, -65 \right]$ $\nicefrac{\rm{dBm}}{100 \ \rm{kHz}}$, corresponding to left/right side of the signal spectrum. Furthermore, we have considered another aggressive target mask-2, \ie, $\vec{\gamma}_{{\rm mask-2}}  \!\coloneqq\! {\vec{\gamma}{2}} \!=\! \left[ -85, -85, -75, -75 \right]$ $\nicefrac{\rm{dBm}}{100 \ \rm{kHz}}$ for \ac{ACLR} and \ac{EVM} performance, particularly. For this work, we have not optimized the selection of the discrete frequencies' set for the precoding. 

Based on our numerical grid-search utilizing Test 3 with 2 transmit antennas (cf. Table \ref{tasim}), we found a suitable $\rho = 10$ for consensus \ac{ADMM} algorithm (cf. Algorithm \ref{alg:consensus_admm_msp}). For the \ac{MSP}~\eqref{eqn:original_msp_optimization_problem} solution, we have employed CVX wrapper with SDPT3 solver \cite{Grant2014CVX:Beta}. Note that we do not have any additional radio hardware impairments at the transmitter and/or receiver besides the distortion generated at the transmitter by the considered spectral precoders, if enabled. 

Furthermore, for benchmarking purpose, we have also evaluated \ac{NSP} \cite{DeBeek2009SculptingPrecoder} and \ac{ENSP} \cite{VanDeBeek2009EVM-constrainedEmission}, where for each transmit antenna branch \ac{NSP} and \ac{ENSP} are performed independently.  More specifically, for each $j$-th transmit antenna, we perform $\overline{\vec{d}}_j \! = \! \mat{G} \vec{d}_j$, where $\vec{G} \! = \! \vec{I}_N \! - \! \alpha \mat{A}^{\rm H}\!\left(\!\mat{A}\!  \mat{A}^{\rm H}\right)^{-1}\!\mat{A}$ and $\mat{A}\!\left[m,k\right]$ in \eqref{eqn:definition_of_a_nu_k_simplified_without_diric_function}. For \ac{NSP} $\alpha\!=\!1$, and \ac{ENSP} $\alpha \! \in \! \Rm$ is a tunable parameter to meet the desired \ac{EVM}. 

\subsection{Simulation Results} \label{subsec:simulation_results}
In this section, we present simulation results for both unconstrained and constrained \ac{EVM} spectral precoding methods.

In Fig. \ref{fig:fig1__txevm_unconstrained_aclrevm_vs_iter__psd_evm_final_iter}, we present the convergence behaviour of the proposed \ac{EVM}-unconstrained algorithms in terms of \ac{OOB} and in-band performance considering Test $3$ (cf. Table \ref{tasim}) and two different target masks, \ie, mask-1 and mask-2.  

Figure \ref{fig:fig1__txevm_unconstrained_aclrevm_vs_iter__psd_evm_final_iter}\subref{fig:fig1_aclr1_vs_iter_sem2_sem1}  shows the \ac{ACLR} performance against iterations. Evidently, \ac{SSP} converges in $2$-$3$ iterations for both target masks to achieve the same \ac{ACLR} results as rendered by \ac{MSP}(CVX). Under relaxed mask-1 constraint, \ac{ADMM} requires nearly $80$ iterations, while for aggressive mask-2 it requires approximately $800$ iterations. On the other hand, we have observed that CVX achieves the \ac{MSP} solution in nearly $25$ iterations, but we could not manage to extract the result for every single iteration{; hence, we show the final result rendered by the CVX solver corresponding to the last iteration in the figures}. We have also observed similar behaviour for \ac{EVM} against iterations. For the subsequent \ac{EVM}-unconstrained results, we consider mask-1 and fix the iterations of \ac{ADMM} and \ac{SSP} as $80$, and $2$, respectively. Furthermore, we have numerically observed that these spectral precoders have a negligible impact on~the~peak~to~average~power~ratio. 

Figure \ref{fig:fig1__txevm_unconstrained_aclrevm_vs_iter__psd_evm_final_iter}\subref{fig:fig3_psd_sem2_sem1_final_iter} exhibits the average \ac{PSD} versus frequency for the proposed algorithms and both masks considering final iterations of the respective algorithms. The NR mask corresponding to a medium range BS with the maximum output of $38$ $\rm{dBm}$ \cite[Section~6.6.4]{3GPPTS38.1042018NRReception} is also shown for the completeness, but the mask is normalized according to the normalized transmit signal power of $0$ $\rm{dBm}$, \ie, approximately $-21.5$ $\nicefrac{\rm{dBm}}{100 \ \rm{kHz}}$, in the link simulations. All the proposed methods in addition to the prior art fulfils the 3GPP \ac{NR} mask. 

Figure \ref{fig:fig1__txevm_unconstrained_aclrevm_vs_iter__psd_evm_final_iter}\subref{fig:fig4_evm_in_percent_vs_prb_sem2_sem1_final_iter} depicts the average \ac{EVM} distribution per \ac{PRB}s \cite{3GPPTS38.2112018NRModulation}. The edge \ac{PRB}s have relatively high distortion power compared to the central \ac{PRB}s. Notice that the NR bandwidth is slightly asymmetric with respect to the direct-current carrier. Moreover, the discrete selected frequency points (\vec{\nu}) are not symmetric with respect to the direct-current carrier. Therefore, one could consequently observe that the \ac{EVM} distribution in the frequency-domain is asymmetric.

Subsequently, we present the performance of \ac{EVM}-constrained while considering mask-1 constraint. In Fig. \ref{fig:fig10__convergence_behaviour__of_wideband_evm_constrained}, we exhibit the convergence behaviour of the proposed wideband \ac{EVM}-constrained algorithms in terms of the in-band, namely, \ac{EVM}, and \ac{OOB} performance, namely, \ac{ACLR} and \ac{PSD}. For these simulations, we have considered $8\%$ wideband unequalized \ac{EVM}.

\begin{figure*}[!htbp]
  \centering
  \begin{minipage}{.3125\linewidth}
    \centering
    \subcaptionbox{\footnotesize Ratio of \ac{OOBE} power to mask-1. \label{fig:fig9_ratio_of_oob_power_to_target_SEM_cruise_ssp} }
      {\includegraphics[width=\textwidth,trim=0mm 0mm 0mm 0mm,clip]{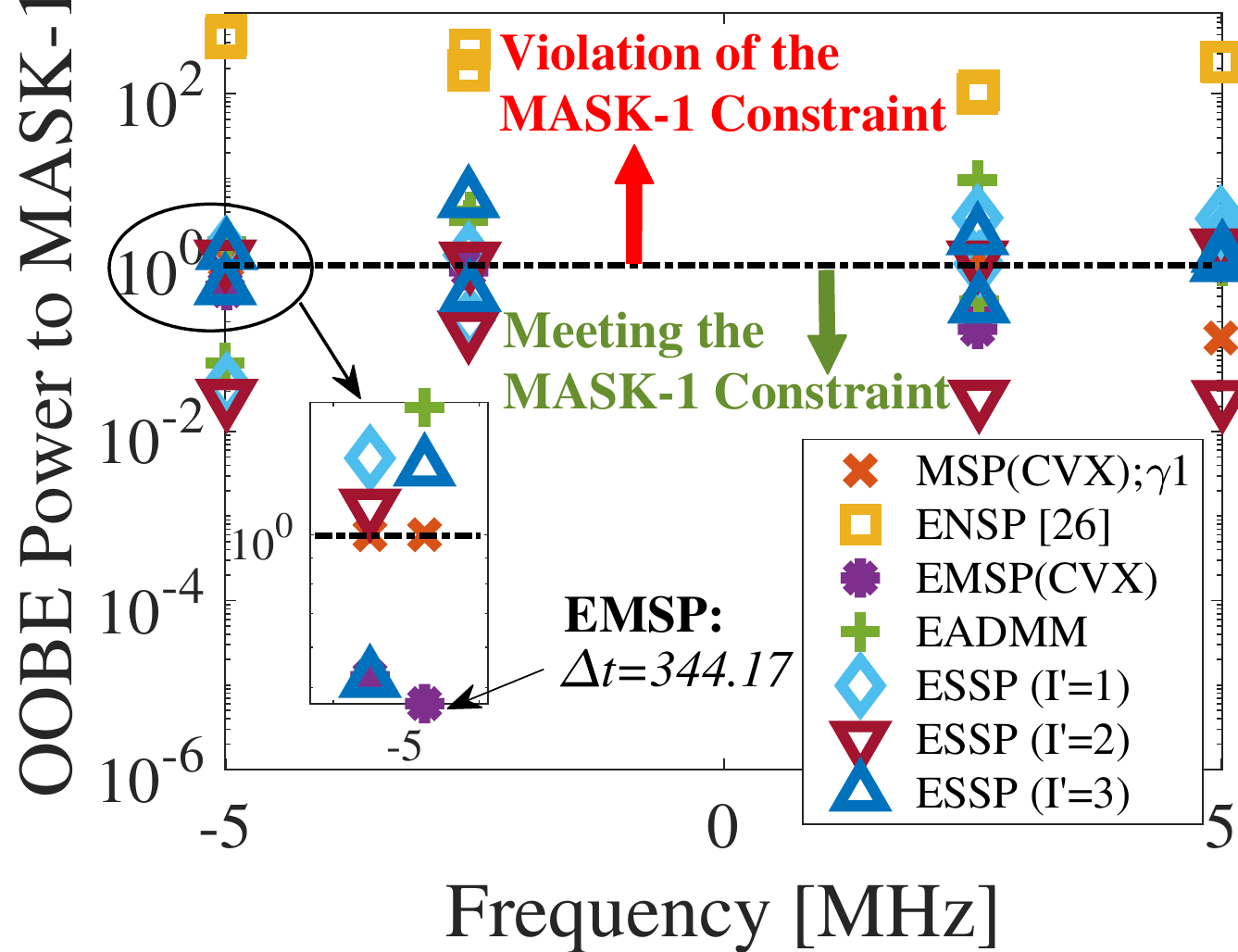}}
  \end{minipage}\quad
   \begin{minipage}{.3125\linewidth}
    \centering
    \subcaptionbox{\footnotesize \ac{ACLR}-vs.-Iterations of \ac{ESSP}. \label{fig:fig9_aclr1_vs_iter_sem1_cruise_ssp_wbevm_shows_divergence_final} }
    {\includegraphics[width=\textwidth,trim=0mm 0mm 0mm 0mm,clip]{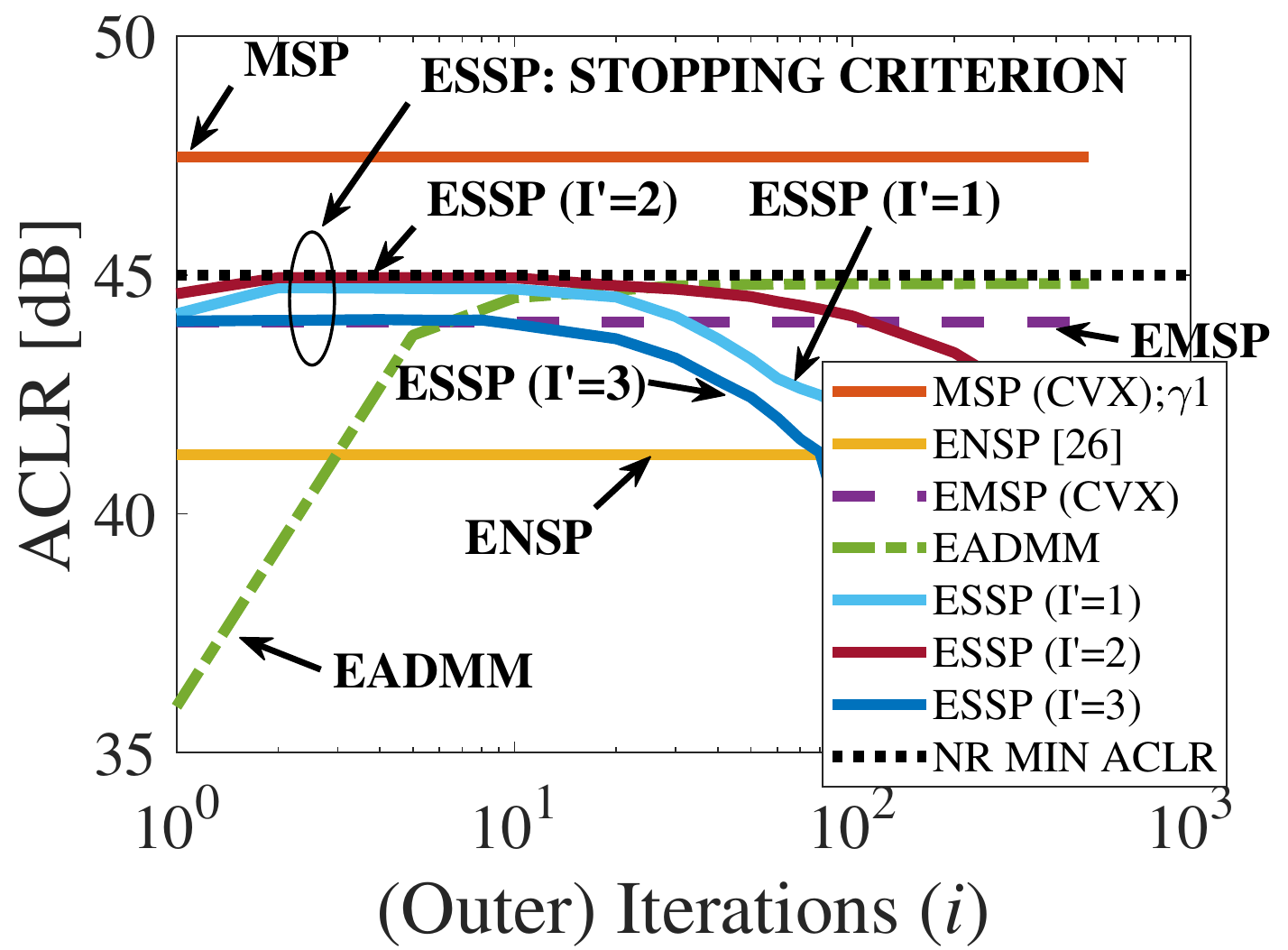}}
  \end{minipage}\quad
  \begin{minipage}{.3125\linewidth}
    \centering
    \subcaptionbox{\footnotesize \ac{EVM}-vs.-Iterations.  \label{fig:fig9_evm_vs_iter_sem1_cruise_shows_divergence_final}}
      {\includegraphics[width=\textwidth,trim=0mm 0mm 0mm 0mm,clip]{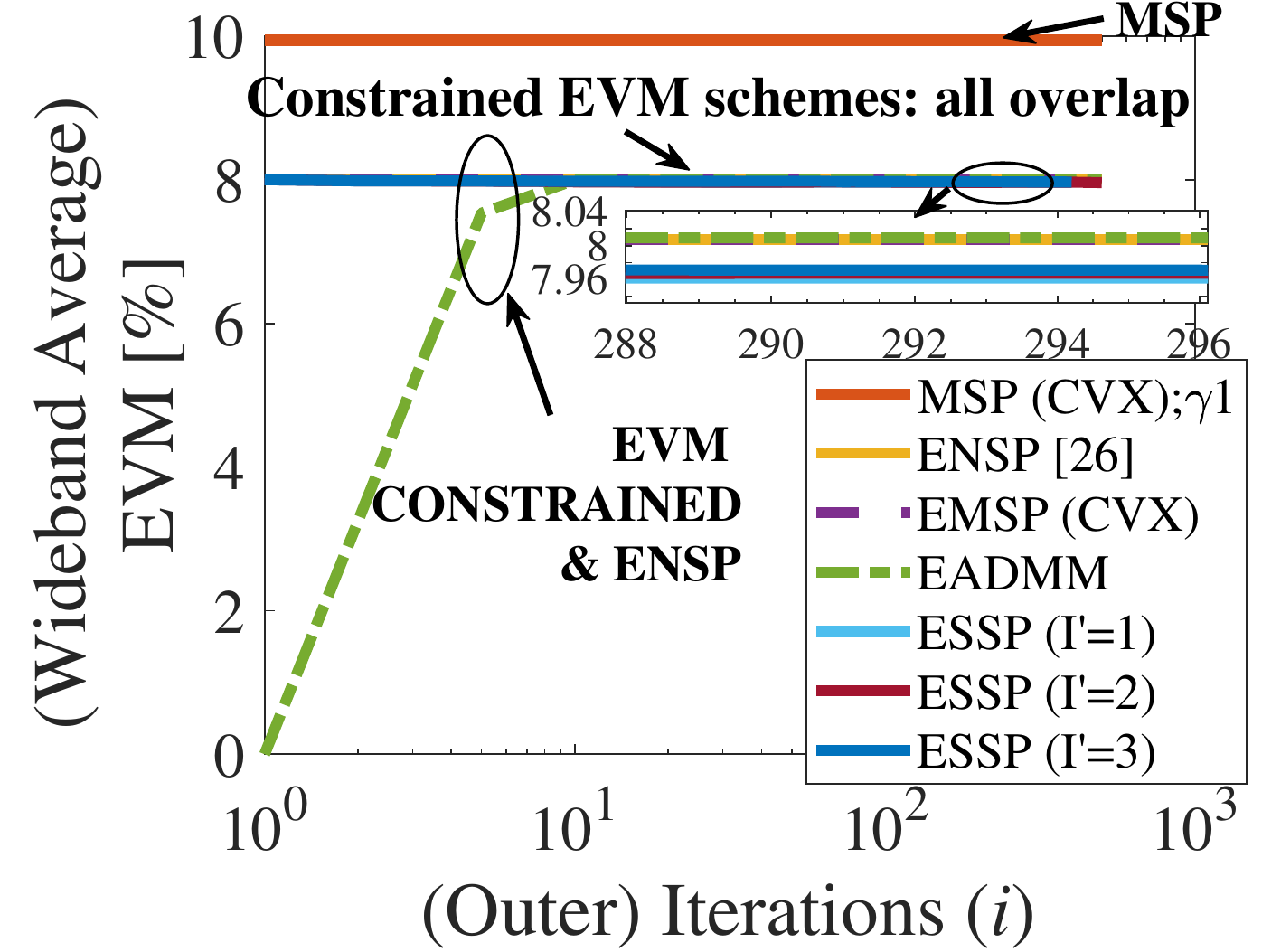}}
  \end{minipage}
  \vspace{-2mm}
  \caption{\footnotesize Convergence behaviour of wideband $8\%$ \ac{EVM}-constrained spectral precoding. }
  \label{fig:fig10__convergence_behaviour__of_wideband_evm_constrained} 
  \vspace{-5.5mm}
\end{figure*}

Figure \ref{fig:fig10__convergence_behaviour__of_wideband_evm_constrained}\subref{fig:fig9_ratio_of_oob_power_to_target_SEM_cruise_ssp} illustrates that the problem \eqref{eqn:generic_min_problem_with_sum_of_indicator_functions_evm_and_sum_of_rank1_qc} is infeasible, \ie, the intersection of the two convex sets corresponding to mask and \ac{EVM} constraints are empty, by considering a single \ac{OFDM} symbol and by randomly selecting any of the transmit antenna branch corresponding to Test $3$. More specifically, we show \ac{OOBE} power to the target mask, \ie, $\left({\left| \mat{A} \overline{\vec{d}}_j \right|^2}\right) \oslash \left({ \Delta t \bm{\gamma} } \right)$, where $\oslash$ denotes elementwise division. If this ratio is larger than $1$, then it implies that the target mask constraint is violated at the considered $\nu$-th discrete frequency point. On the contrast, if this ratio is less than and equal to $1$, the mask constraint is met by the spectrally precoded \ac{OFDM} symbol. {Notice that, in the case of \ac{EMSP}, $\Delta t = 344.17$ was obtained from CVX \cite{Grant2014CVX:Beta} for the given \ac{OFDM} symbol and transmit antenna branch, which implies that the original problem is infeasible. As it can be observed from the figure, the solution rendered by the \ac{MSP} naturally meets the mask constraint since it is an inequality constrained minimization problem, \ie, this ratio will always be less than {or} equal to $1$ as it is a feasible problem}. The spectrally precoded symbol generated by \ac{ENSP} algorithm does not meet any of the $8$ mask constraints while the \ac{EVM} constraint is kept to $8\%$. The symbol generated by \ac{EMSP} meets the target mask constraints but after considering $\Delta t = 344.17$, which evidently penalizes the achieved \ac{ACLR} compared to the \ac{ACLR} achieved through \ac{MSP} algorithm in order to meet the wideband $8\%$ \ac{EVM} constraint. Furthermore, the precoded symbol rendered by \ac{EADMM} and \ac{ESSP} apparently show that for some discrete frequency points mask constraint is met while for some points it fails. Therefore, long-term average of \ac{PSD}s or \ac{ACLR}s over several slots of \ac{OFDM} signals generated by the proposed algorithms outperform significantly heuristic \ac{ENSP} method in terms of out-of-band performance.  In other words, we can construe from the figure that the \ac{EVM} constraint penalizes the achievable \ac{ACLR} compared to the \ac{ACLR} achieved by unconstrained \ac{EVM} problem.

Figure \ref{fig:fig10__convergence_behaviour__of_wideband_evm_constrained}\subref{fig:fig9_aclr1_vs_iter_sem1_cruise_ssp_wbevm_shows_divergence_final} demonstrate the \ac{ACLR} performance against iterations for all the proposed \ac{EVM}-constrained spectral precoding algorithms. In the case of \ac{ESSP}, we illustrate the performance with three different internal iterations corresponding to the $\prox$ operator for the sum of the indicator functions to the mask constraints, \ie, calling \ac{SSP} algorithm. Strikingly, the \ac{ACLR} performance accomplished by \ac{ESSP} without stopping criterion diverges after typically $2$ iterations, although the wideband \ac{EVM} performance convergence behaviour is consistent over the iterations, \ie, meeting the wideband \ac{EVM} constraint---see Fig. \ref{fig:fig10__convergence_behaviour__of_wideband_evm_constrained}\subref{fig:fig9_evm_vs_iter_sem1_cruise_shows_divergence_final}. Remarkably, the result rendered by \ac{EADMM} do not digress or diverge over iterations, unlike observed in Douglas-Rachford-based \ac{ESSP} algorithm without early stopping criterion. Thus, \ac{ESSP} requires a stopping criterion to render an approximate solution. We choose to terminate the outer iterations within \ac{ESSP} when the \ac{ACLR} performance degrades with the increasing iterations. Hence, we employ an early stopping criterion for the results presented in the sequel, which typically terminates the algorithm at around $2$ outer iterations (for all the considered test cases). 

\iftrue
\begin{figure*}[tp!]
  \centering
  \begin{minipage}{.3125\linewidth}
    \centering
    \subcaptionbox{\footnotesize \ac{EVM} [\%] over subcarriers. \label{fig:fig10_evm_in_percent_vs_sc_sem1_final} } 
      {\includegraphics[width=\textwidth,trim=0mm 0mm 0mm 0mm,clip]{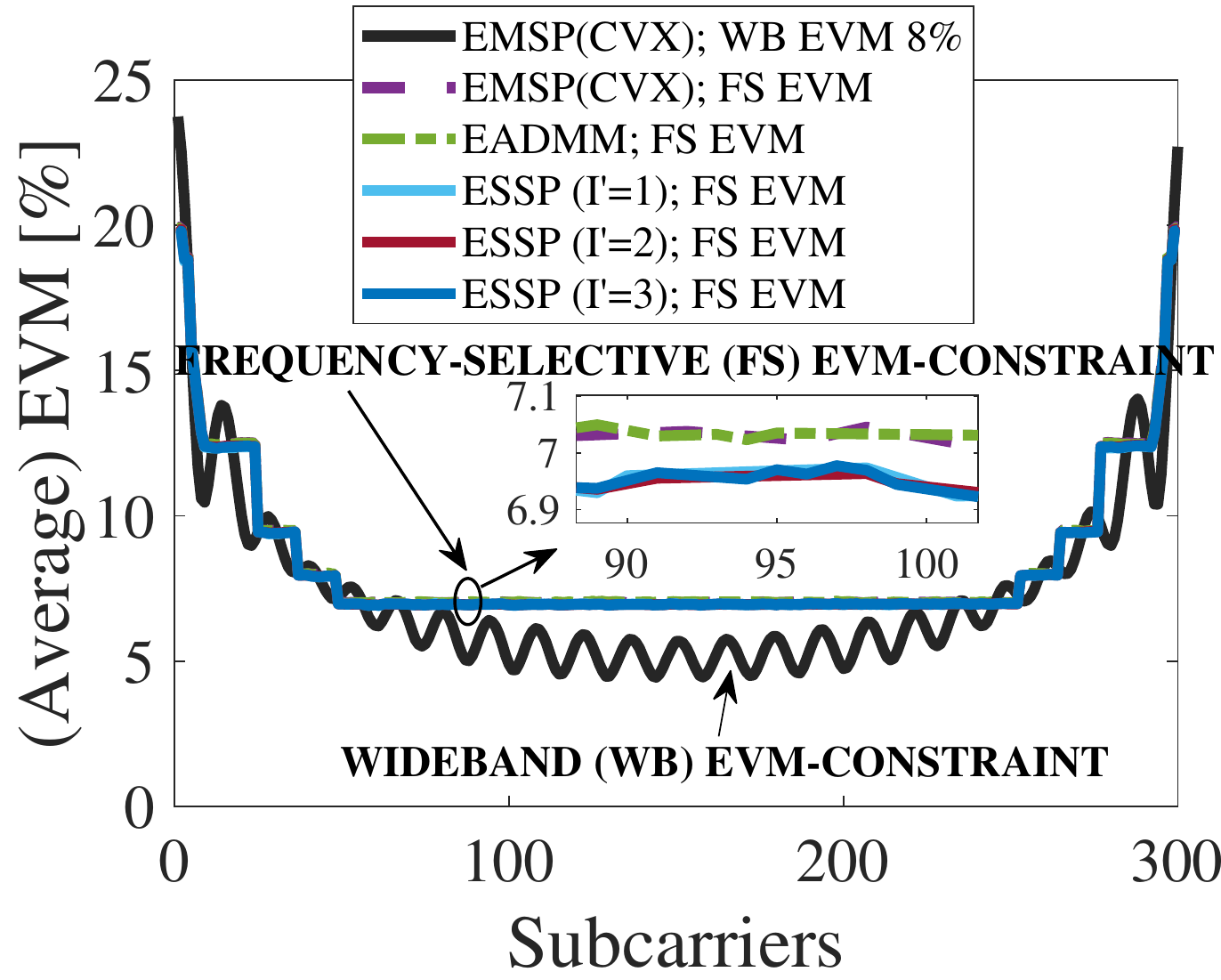}}
  \end{minipage}\quad
   \centering
   \begin{minipage}{.3125\linewidth}
    \centering
    \subcaptionbox{\footnotesize \ac{ACLR}-vs.-Iterations of \ac{ESSP}. \label{fig:fig10_aclr1_vs_iter_sem1_fs_evm_shows_divergence_final} }
    {\includegraphics[width=\textwidth,trim=0mm 0mm 0mm 0mm,clip]{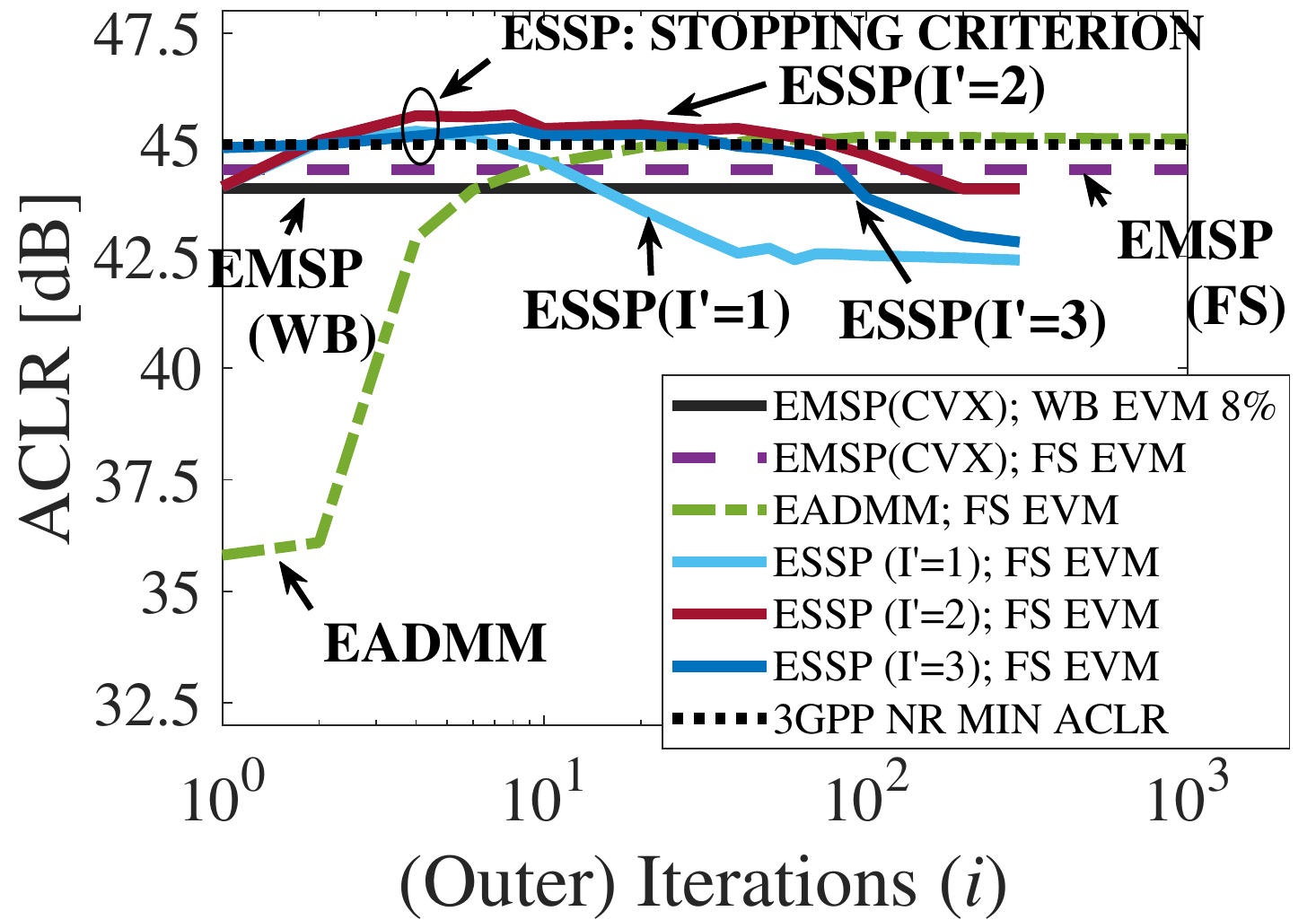}}
  \end{minipage}\quad
   \centering
  \begin{minipage}{.3125\linewidth}
    \centering
   \subcaptionbox{\footnotesize \ac{EVM}-vs.-iterations. \label{fig:fig10_wbevm_vs_iter_sem1_fs_evm_shows_divergence_final}}
      {\includegraphics[width=\textwidth,trim=0mm 0mm 0mm 0mm,clip]{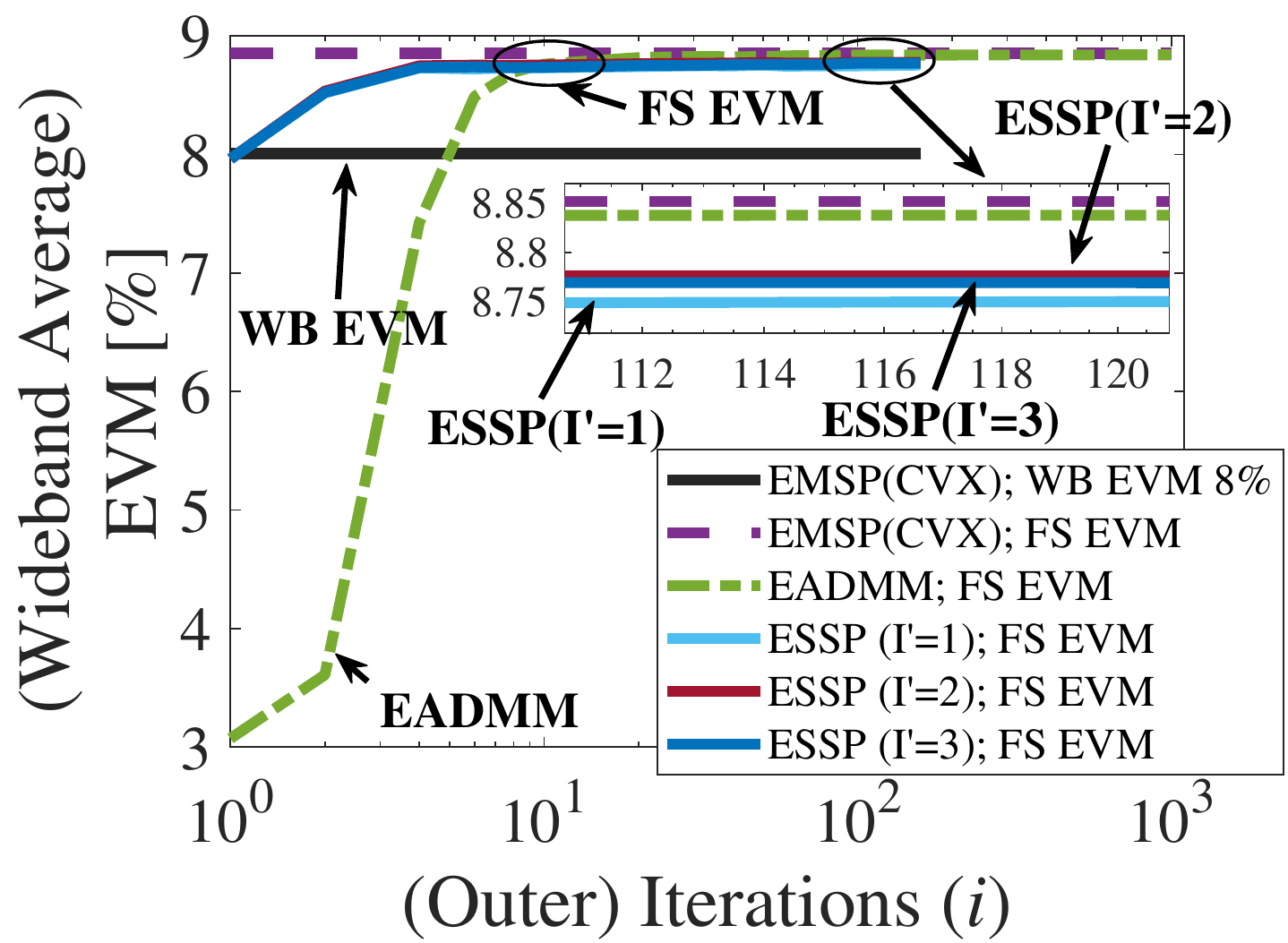}}
  \end{minipage}
  \vspace{-2.5mm}
  \caption{\footnotesize Convergence behaviour of frequency-selective \ac{EVM}-constrained spectral precoding.}
  \label{fig:fig10__convergence_behaviour__of_frequency_sel_evm_constrained} 
  \vspace{-2.5mm}
\end{figure*}
\fi

Figure \ref{fig:fig10__convergence_behaviour__of_frequency_sel_evm_constrained} demonstrates the convergence behaviour of frequency-selective \ac{EVM} algorithms. The frequency-selective \ac{EVM} constraints are arbitrarily chosen to exhibit the efficacy of the proposed frequency-selective algorithms. The first lower and upper edge \ac{PRB}s were set to have \ac{EVM} requirement $\big[ 20\%, 20\%, 19\%, 19\%, 16\%, 15\%, 14\%, 13\%, 12.5\%,  12.5\%, \allowbreak  12.5\%,  12.5\% \big]$ such that the edge most subcarrier of the edge \ac{PRB} corresponds to  $20\%$ \ac{EVM} and the inner most subcarrier corresponds to $12.5\%$ \ac{EVM}. Furthermore, all the subcarriers of the second, third, and fourth edge \ac{PRB}s are set to have the flat target of $12\%$, $9.5\%$, and $8\%$ \ac{EVM}, respectively. The remaining central subcarriers are set to $7\%$. This represents a wideband average \ac{EVM} of $\sim\!8.9\%$ across the subcarriers. The \ac{EVM} distribution over subcarriers are shown in Fig.~\ref{fig:fig10__convergence_behaviour__of_frequency_sel_evm_constrained}\subref{fig:fig10_evm_in_percent_vs_sc_sem1_final}, where all the proposed methods meet the frequency-selective \ac{EVM} requirement. In Fig.~\ref{fig:fig10__convergence_behaviour__of_frequency_sel_evm_constrained}\subref{fig:fig10_wbevm_vs_iter_sem1_fs_evm_shows_divergence_final}, all the frequency-selective \ac{EVM} algorithms have similar wideband \ac{EVM} performance. The imposed mask and frequency-selective \ac{EVM} constraints make the problem infeasible which is evident from the divergence of \ac{ACLR} metric rendered by \ac{ESSP} algorithms without any stopping criterion as depicted in Fig.~\ref{fig:fig10__convergence_behaviour__of_frequency_sel_evm_constrained}\subref{fig:fig10_aclr1_vs_iter_sem1_fs_evm_shows_divergence_final}. We observe that the \ac{ACLR} performance delivered by \ac{EADMM} are quite stable over iterations even though the problem is infeasible. Reiterating that the \ac{ENSP} can not be employed or extended to support frequency-selective \ac{EVM} requirement.

In Fig.~\ref{fig:fig10__oob_in_band_performance__of_wideband_evm_constrained}, we exhibit in-band and \ac{OOBE} performance of the proposed $8\%$ wideband \ac{EVM}-constrained algorithms, where \ac{ESSP} employs $2$ inner and $2$ outer iterations (with stopping criterion) for all the 3 test cases. Further, \ac{EADMM} employs $40$ iterations for the performance evaluation.  More specifically, we show the in-band performance, in terms of \ac{BLER} and (normalized) throughput against received \ac{SNR} at the user-equipment and \ac{EVM} distribution for the final chosen iterations of the respective algorithms. We also depict the \ac{OOBE} performance, in terms of \ac{ACLR} versus (outer) iterations and average \ac{PSD}s of the the final chosen iterations of the respective algorithms.

\iftrue
\begin{figure*}[tp!]
  \centering
  \begin{minipage}{.3125\linewidth}
    \centering
    \subcaptionbox{\footnotesize \ac{ACLR}-vs.-Iterations for Test 3 (8Tx2Rx). \label{fig:fig9_aclr1_vs_iter_sem1_wbevm_final}}
      [0.85\linewidth]{\includegraphics[width=\textwidth,trim=0mm 0mm 0mm 1mm,clip]{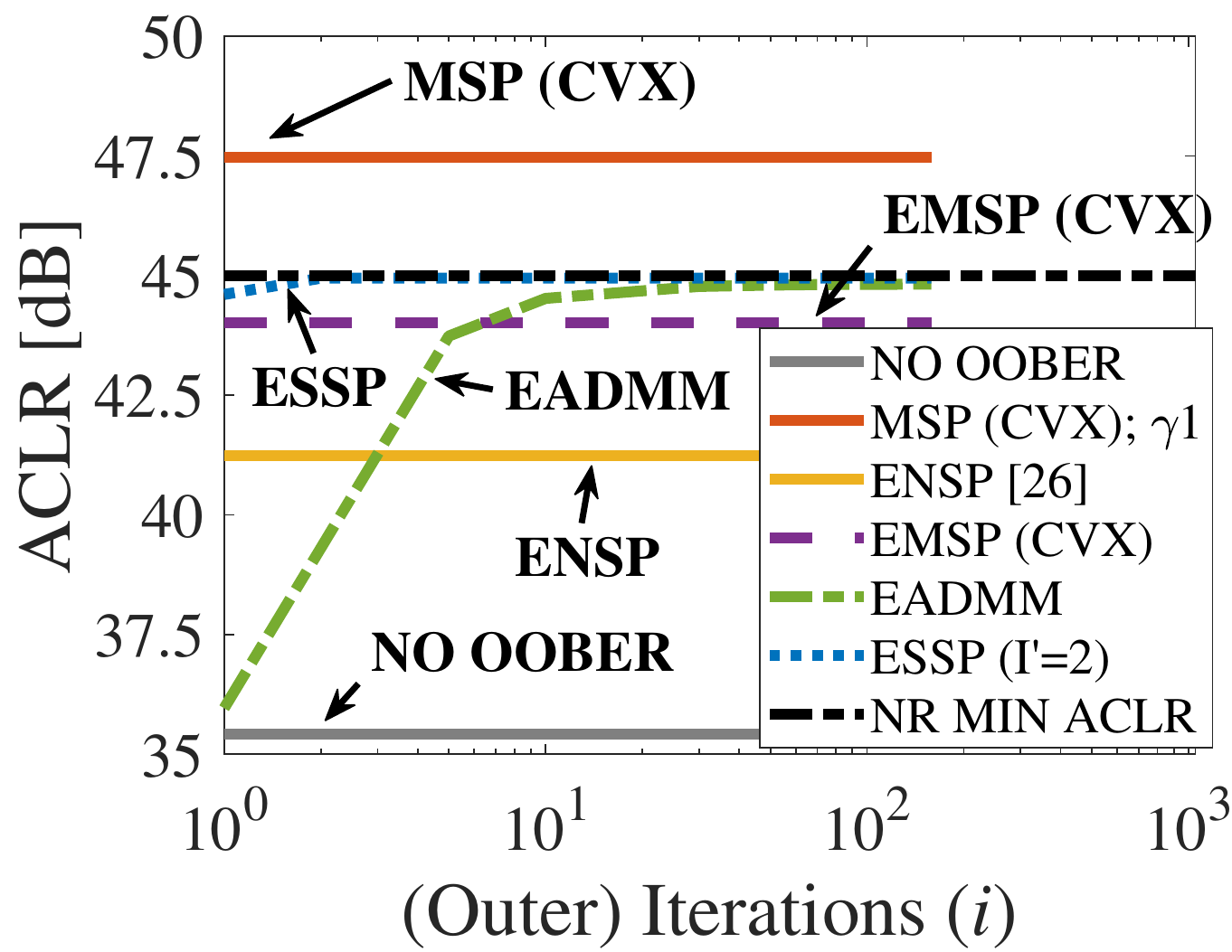}}
  \end{minipage}\quad
   \centering
  \begin{minipage}{.3125\linewidth}
    \centering
    \subcaptionbox{\footnotesize \ac{PSD} for Test 3 (8Tx2Rx). \label{fig:fig9_psd_sem1_final_iter} }
    [0.85\linewidth]{\includegraphics[width=\textwidth,trim=0mm 0mm 0mm 1mm,clip]{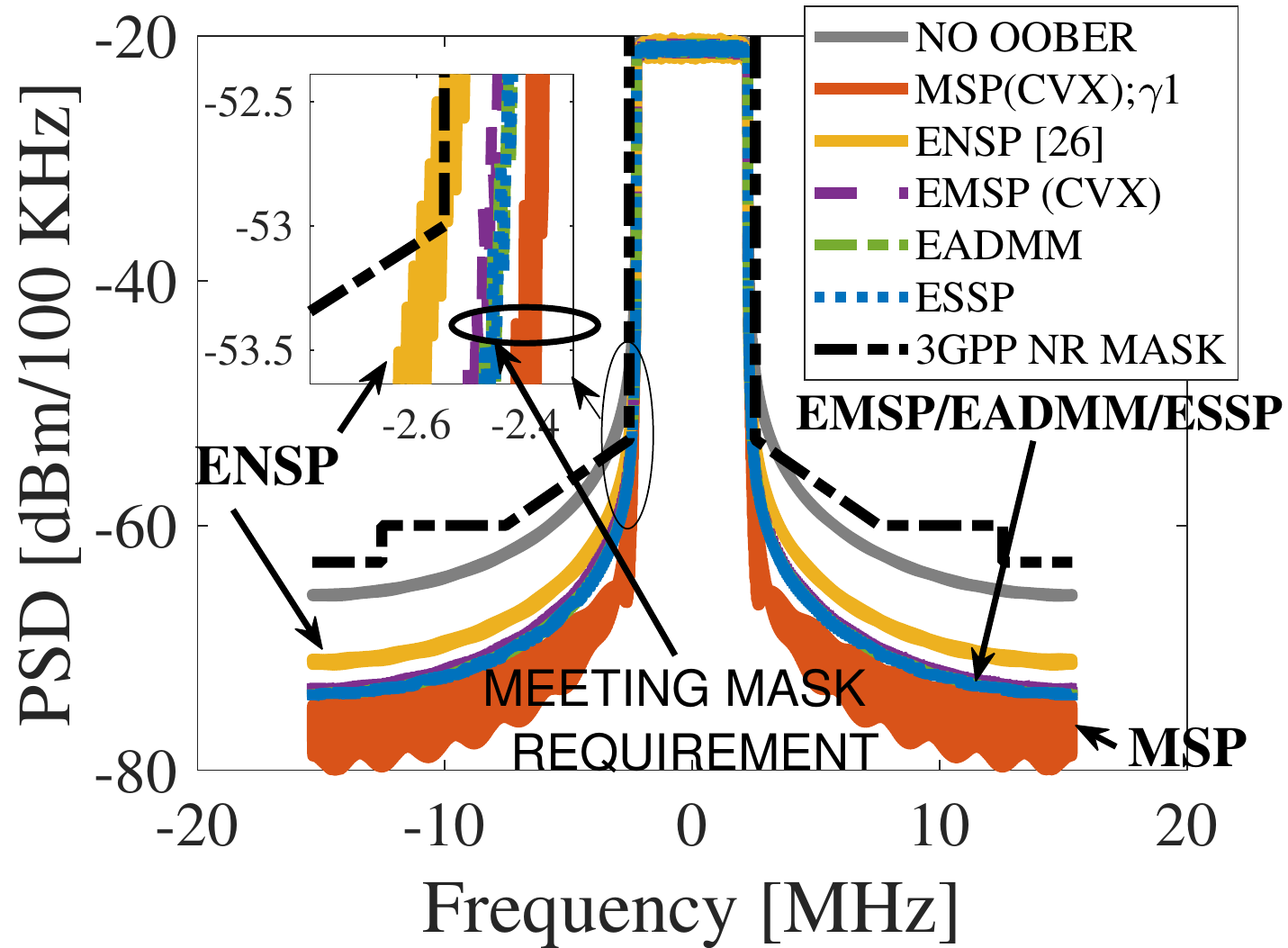}}
  \end{minipage}\quad
  \begin{minipage}{.3125\linewidth}
    \centering
    \subcaptionbox{\footnotesize \ac{EVM} distribution for Test 3 (8Tx2Rx). \label{fig:fig9_evm_in_percent_vs_sc_sem1_final}}
     [0.85\linewidth] {\includegraphics[width=\textwidth,trim=0mm 0mm 0mm 1mm,clip]{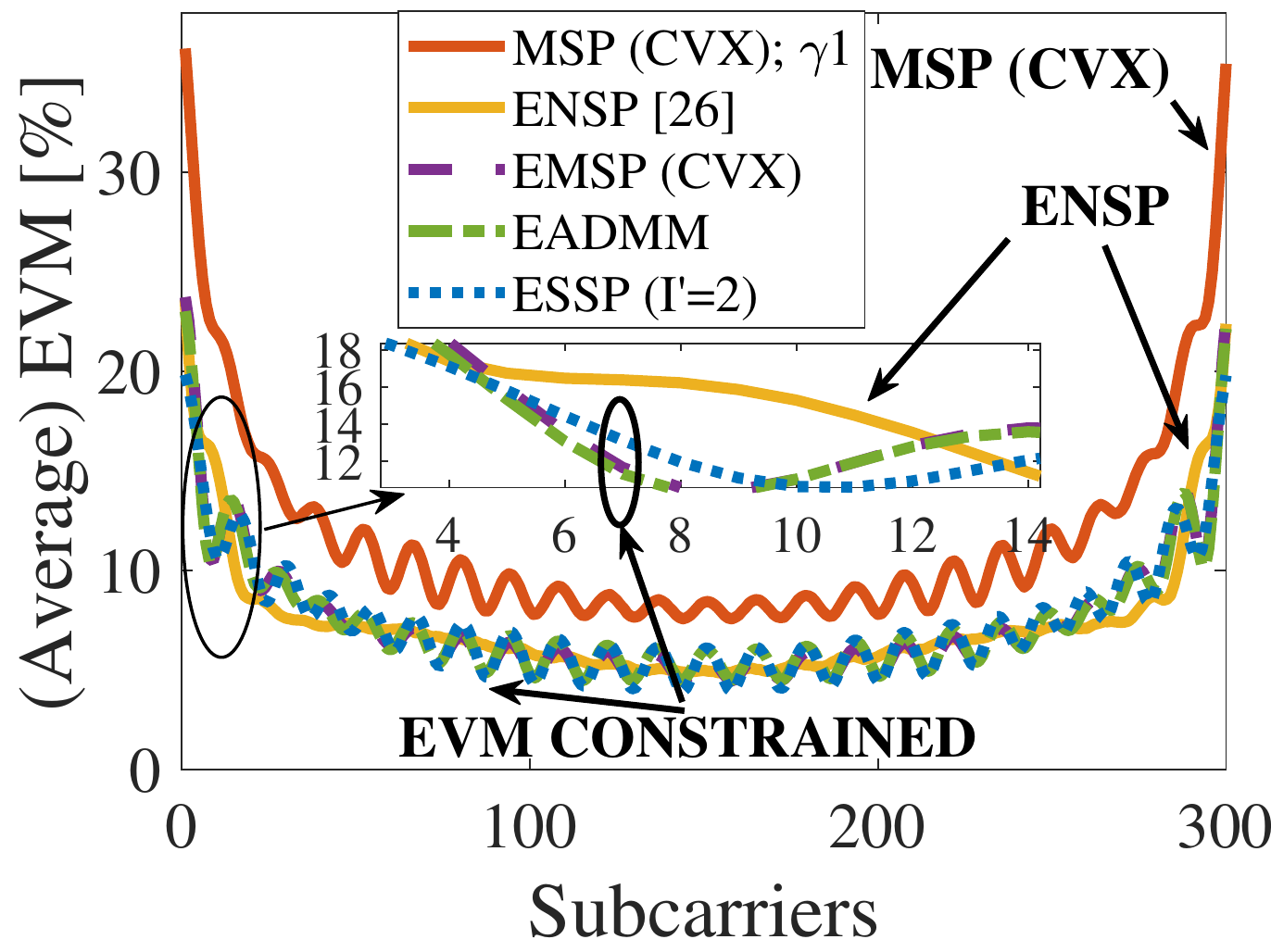}}
  \end{minipage}
  
  \vspace{1mm}
  
  \begin{minipage}{.3125\linewidth}
    \centering
    \subcaptionbox{\footnotesize \ac{BLER}-vs.-\ac{SNR} for \mbox{Test~1} (2Tx2Rx, Rank~1). \label{fig:fig10_test2_2Tx_bler_vs_snr__qam64} }
    [0.85\linewidth]  {\includegraphics[width=\textwidth,trim=0mm 0mm 0mm 0mm,clip]{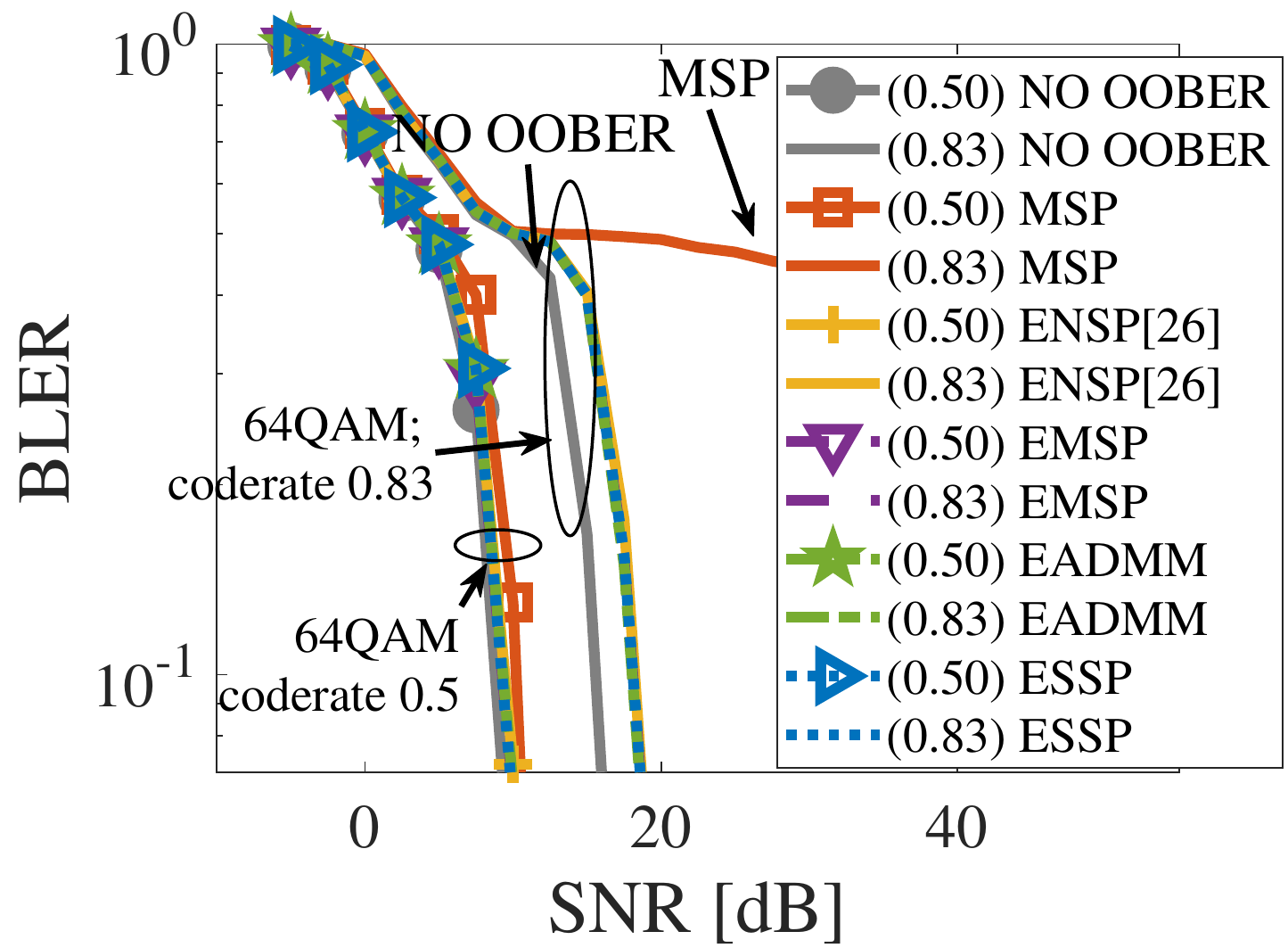}}
  \end{minipage}\quad
  \begin{minipage}{.3125\linewidth}
    \centering
    \subcaptionbox{\footnotesize \ac{BLER}-vs.-\ac{SNR} for \mbox{Test~2} (8Tx2Rx, Rank 1). \label{fig:fig10_test4_8Tx_bler_vs_snr__qam64} }
    [0.85\linewidth]  {\includegraphics[width=\textwidth,trim=0mm 0mm 0mm 0mm,clip]{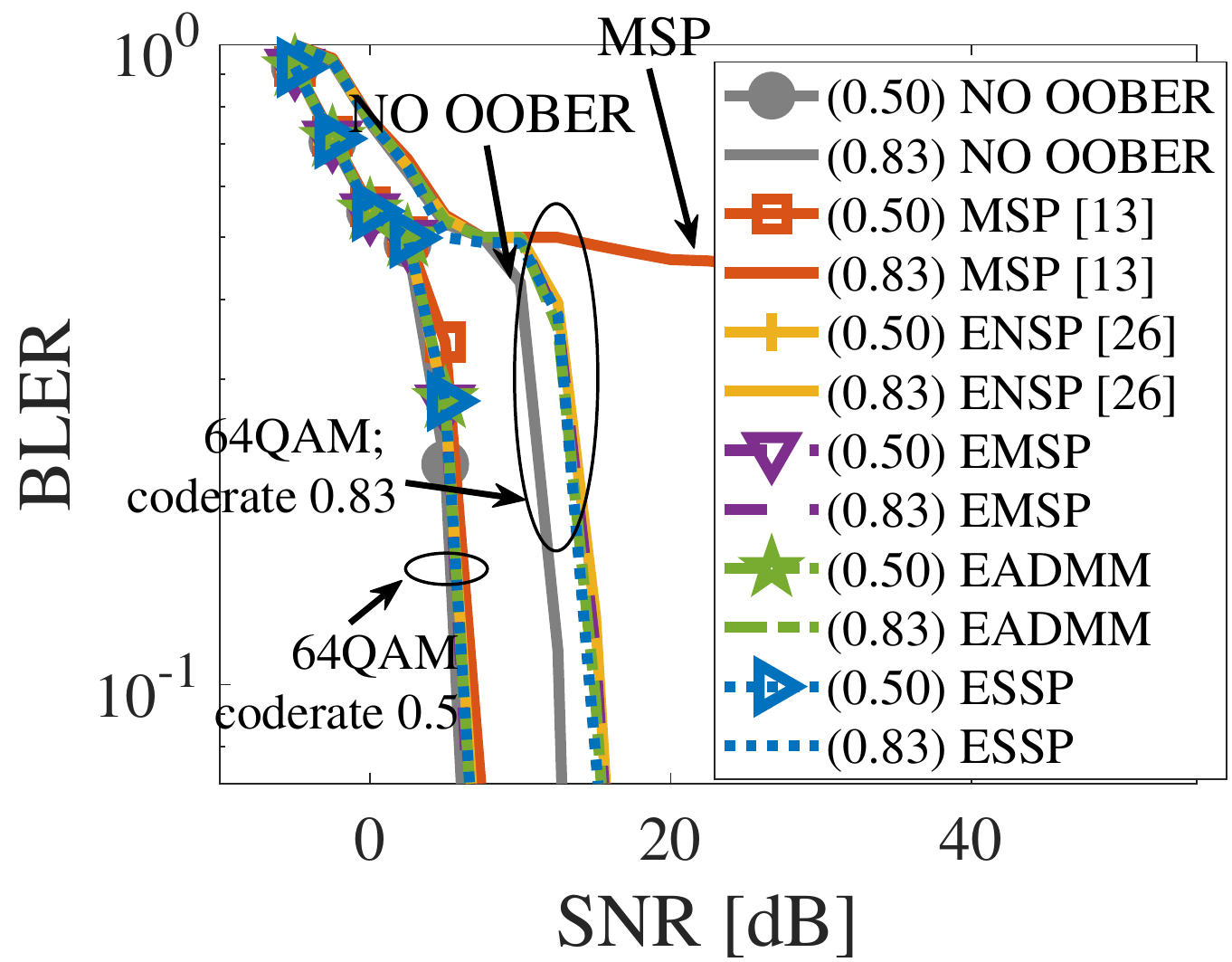}}
  \end{minipage}\quad
  \centering
  \begin{minipage}{.3125\linewidth}
    \centering
    \subcaptionbox{\footnotesize Throughput-vs.-\ac{SNR} for Test~3 (8Tx2Rx). \label{fig:fig10_test7_8Tx2Rx_la_tp_vs_snr}}
    [0.85\linewidth]  {\includegraphics[width=\textwidth,trim=0mm 0mm 0mm 0mm,clip]{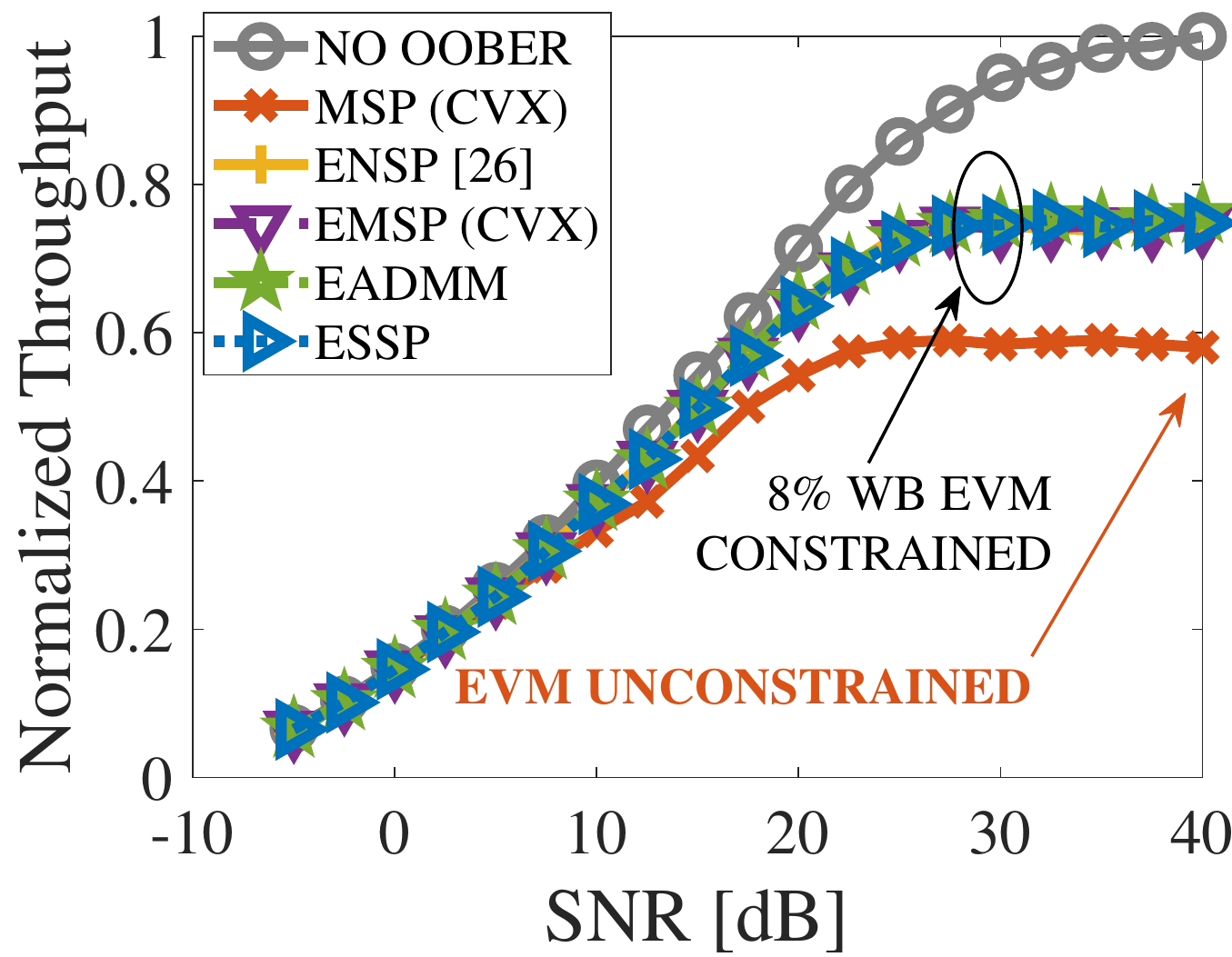}}
  \end{minipage}
  \vspace{-2mm}
  \caption{\footnotesize Performance of wideband $8\%$ \ac{EVM}-constrained spectral precoding including \ac{ESSP} with stopping criterion.}
  \label{fig:fig10__oob_in_band_performance__of_wideband_evm_constrained}
  \vspace{-1mm}
\end{figure*}
\fi

Particularly, Figure~\ref{fig:fig10__oob_in_band_performance__of_wideband_evm_constrained}\subref{fig:fig9_aclr1_vs_iter_sem1_wbevm_final} depicts \ac{ACLR} performance against iterations utilizing Test~$3$ (similar results are observed for other test scenarios). Since the problem was infeasible, the \ac{EVM} constraint penalizes the achievable \ac{ACLR} compared to the \ac{ACLR} achieved by unconstrained \ac{EVM} problem, which can be corroborated by the \ac{EMSP} performance comparing with the performance achieved by the \ac{MSP} (CVX).  Moreover, it can be noticed that all the algorithms have different \ac{ACLR} performance and unfortunately do not meet the minimum $45$~dB \ac{ACLR} requirement for the 1st channel. Still, all the proposed algorithms render the \ac{ACLR} between $44$~dB and $45$~dB. Conspicuously, \ac{ESSP} achieves $44.95$~dB \ac{ACLR}, which is $0.05$ dB below the minimum $45$~dB requirement. Moreover, it can be observed that all the proposed algorithms meet the 3GPP mask requirement as illustrated in the \ac{PSD}---see Fig.~\ref{fig:fig10__oob_in_band_performance__of_wideband_evm_constrained}\subref{fig:fig9_psd_sem1_final_iter}. Therefore, one could employ a relaxed additional spectrum shaping method, for instance, a transmit windowing or filtering that would use less than $9\%$ of the cyclic prefix length to meet the \ac{ACLR} requirement with some implementation margin. In other words, one can combine spectral precoding with other (time-domain) spectrum shaping methods, see, \eg, \cite{Mohamad_PhDThesis2019} and thereby improve \ac{OOBE} performance. For completeness, Fig.~\ref{fig:fig10__oob_in_band_performance__of_wideband_evm_constrained}\subref{fig:fig9_evm_in_percent_vs_sc_sem1_final} depicts the distribution of \ac{EVM} over subcarriers (in frequency-domain) after the final iterations of the respective algorithms.

In Fig.~\ref{fig:fig10__oob_in_band_performance__of_wideband_evm_constrained}\subref{fig:fig9_aclr1_vs_iter_sem1_wbevm_final}, the \ac{ACLR} performance rendered by the heuristic  \ac{ENSP} method is, inadequately low, $41.24$~dB compared to the proposed principled methods that are between $44$~dB and $45$~dB,~\ie, \ac{ENSP} has $3$-$4$~dB loss in \ac{ACLR} but similar wideband \ac{EVM}. Furthermore, \ac{ENSP} fails to meet the mask requirement, cf. Fig.~\ref{fig:fig10__oob_in_band_performance__of_wideband_evm_constrained}\subref{fig:fig9_psd_sem1_final_iter}. Thus, \ac{ENSP} would need an aggressive additional spectrum shaping method, for instance transmit windowing with nearly more than $30\%$-$40\%$ cyclic prefix length, to meet both mask and \ac{ACLR} requirement. Hence, the gain from such spectral precoding vanishes and may be futile to employ in a realistic system. 

In Fig. \ref{fig:fig10__oob_in_band_performance__of_wideband_evm_constrained}, we also manifest the in-band performance of the proposed algorithms considering wideband $8\%$ \ac{EVM} constraint. In particular, we show \ac{BLER} against received \ac{SNR} at the user-equipment for the Test $1$ and Test $2$ since there is no link adaptation in these tests except for the spatial precoding matrix indicator (PMI) adaptation, where all the proposed \ac{EVM}-constrained algorithms have similar performance. Furthermore, we present normalized throughput versus received \ac{SNR} for the Test $3$ with $8$ transmit antennas since this test employs fast link adaptation, \ie, spatial precoder (PMI), rank/layers, and modulation and code rate adaptation (referred to as channel quality indicator in the 3GPP standard) with $10\%$ target \ac{BLER} on long-term average. From these figures, firstly, we observe that the distribution of the \ac{EVM} in the frequency domain matters since meeting minimum $8\%$ wideband \ac{EVM} requirement for 64QAM may not be sufficient to meet the receiver demodulation performance. Secondly, we observe that increasing number of the transmit antennas from $2$ to $8$ for the low-rank scenario does not improve the in-band performance gain. We have observed even with 8 transmit antennas that the \ac{EVM} emanating from spectral precoding is beamformed in the same direction as the signal, similar observations for (large scale) \ac{MIMO} for the \ac{EVM} stemming from other sources, see, \eg,~\cite{Moghadam_distortion_corr_mimo:12,Sienkiewicz_spatially_dependent:14,Larsson:18,Moghadam:18}. The \ac{EVM}-constrained algorithms improve the receiver demodulation performance compared to the \ac{EVM}-unconstrained spectral precoding. However, there is a potential to improve receiver performance while meeting the \ac{OOBE} requirements when the transmitter can be aware of the full channel state information of the receiver(s), which will be addressed in future work.

\vspace{-1.25mm}
\section{Conclusion} \label{sec:conclusion_future_work}

{In this paper we {investigated the problem of spectral precoding with mask compliant properties. The problem is formally posed as some  optimizations} that offer an explicit trade-off between \ac{EVM} and \ac{OOBE} suppression for large-scale MIMO-OFDM systems, but {have the disadvantage of not having} closed-form solutions. To mitigate the complexity of such problems, we proposed a divide-and-conquer approach that decomposes the {spectral precoding problem} into smaller problems {having each} either a closed-form or an efficient solution. More specifically, in the first part of the paper, we developed two computationally efficient algorithms for the \ac{EVM}-unconstrained spectral precoding method, namely 1) \ac{ADMM} and 2) \ac{SSP}---which is derived by utilizing \ac{KKT} conditions, capitalizing on the closed-form solution of rank~1 quadratic form, and a coordinate descent scheme. In the second part---in stark contrast to the {unsystematic} prior art on wideband \ac{EVM}-constrained spectral precoding---we formulated the optimal wideband and frequency-selective constrained \ac{EVM} problems in conjunction with mask-compliant spectral precoding. Subsequently, employing the divide-and-conquer approach, we extended and developed ``hardware-friendly" algorithms for the \ac{EVM}-constrained spectral precoding method, referred to as 1) \ac{EADMM} and 2) \ac{ESSP}---which uses the Douglas-Rachford operator splitting technique to meet an \ac{EVM} constraint while internally utilizes \ac{SSP} for mask constraint. {The proposed algorithms should be treated as vendor-specific transmitter module like filtering, which implies that they are 3GPP NR standard transparent to the transmitter and the receiver.} Numerical results corroborate that the proposed low-complexity algorithms can meet the target EVM constraints and the 3GPP NR mask by suppressing \ac{OOBE}.} 

{This is arguably the first work that proposes computationally affordable \ac{EVM}-constrained and yet mask-compliant spectral precoding.}

\vspace{-2mm}
\begin{appendices}
\section{{Convergence Analysis: \ac{ADMM}---Algorithm~\ref{alg:consensus_admm_msp} and \ac{EADMM}---Algorithm~\ref{alg:consensus_eadmm_msp}}} \label{sec:proof_of_consensus_admm}
{We present the convergence of Algorithm~\ref{alg:consensus_admm_msp} and Algorithm~\ref{alg:consensus_eadmm_msp} in a consolidated manner.}
\begin{theorem}[{Global convergence of consensus \ac{ADMM} or \ac{EADMM}}]  \label{theorem:two_operator_general_consensus_admm}
{Consider an either \ac{EVM} unconstrained \eqref{eqn:msp_problem_with_rank1_constraints} or constrained problem \eqref{eqn:generic_min_problem_with_sum_of_indicator_functions_evm_and_sum_of_rank1_qc} that can be unified as, \ie, 
\begin{equation}  \label{eqn:two_operator_general_consensus_admm}
\begin{aligned} 
&\underset{ \left\{\overline{\mat{Y}}_m \in \Cm^{\NT \times N } \right\},  {\overline{\mat{X}}} \in \Cm^{\NT \times N }}{\text{minimize}} \ 
 \sum_{m=1}^M f_m \left( {\overline{\mat{Y}}_m} \right)  + g\left( {\overline{\mat{X}}} \right)   \\ 
&\hspace{1.1cm} \text{subject to}  \qquad \qquad
 {\overline{\mat{Y}}_m} = {\overline{\mat{X}}}, 
\end{aligned}
\vspace{-1mm}
\end{equation}
where non-differentiable indicator function $f_m \!\coloneqq\! \mathcal{X}_{\mathbfcal{C}_m}$ to the mask constraint set (see \eqref{eqn:ls_msp_constraint_set__mimo_X}) for all $m\!=\!1,\ldots,M$ is closed convex proper (common to both \ac{ADMM} and \ac{EADMM}). The closed convex proper function is either differentiable $g \! \coloneqq \! \left\|\overline{\mat{X}} \! - \!  {\mat{X}} \right\|_F$ (\ac{EVM} unconstrained) or non-differentiable indicator function $g \!\coloneqq\! \mathcal{X}_{\mathbfcal{E}}$ (for either wideband \eqref{eqn:wb_evm_constraint_set} or frequency-selective~\eqref{eqn:fs_evm_constraint_set} \ac{EVM} constraint). Suppose~\eqref{eqn:two_operator_general_consensus_admm} has at least one solution. Now, assume subproblems of \ac{ADMM} and \ac{EADMM} have solutions, and so-called dual residual $\lim_{i \rightarrow +\infty} \left( \overline{\mat{X}}^{\left( i+1 \right)} \! - \! \overline{\mat{X}}^{\left( i \right)} \right) \! = \! 0$ and primal residual 
$\lim_{i \rightarrow +\infty} \allowbreak \left( \overline{\mat{Y}}_m^{\left( i+1 \right)} \! - \! \overline{\mat{X}}^{\left( i+1 \right)} \right) \! = \! 0$, $\forall m \! = \! 1,\ldots,M$, and $\rho \! \in \! \Rm_{>0}$ with some arbitrary initial $\left\{\overline{\mat{X}}^{(0)}, \overline{\mat{Y}}_m^{(0)}, \overline{\mat{Z}}_m^{(0)}\right\}$. {Then,} Algorithm~\ref{alg:consensus_admm_msp} and Algorithm~\ref{alg:consensus_eadmm_msp},}
{at any limit point, $\left\{\overline{\mat{X}}^{\left(i\right)}\right\}$ converge to a \ac{KKT} point~of~\eqref{eqn:two_operator_general_consensus_admm}.}
\begin{proof}
{
Towards the convergence analysis goal of the proposed \ac{ADMM} and \ac{EADMM} algorithms, we follow the proof given in, \eg, \cite{Boyd2011}. Now, we form the (unaugmented) Lagrangian of the unified problem~\eqref{eqn:two_operator_general_consensus_admm} proposed to be solved by \ac{ADMM}/\ac{EADMM} such that
\vspace{-1mm}
\begin{align} \label{eqn:augmented_Lagrangian_admm__generic_form}
    & \hspace{-0.45cm}  \mathcal{L}\!\left(\!\left\{ \overline{\mat{Y}}_m \right\}_{m=1}^M,\!  \overline{\mat{X}},\! \left\{ \overline{\mat{Z}}_m \!\right\}_{m=1}^M\! \right) \nonumber \\
    \! \coloneqq& \! \sum_{m=1}^M f_m \left( {\overline{\mat{Y}}_m} \right) \! + \! g\left( {\overline{\mat{X}}} \right) \! + \! \sum_{m=1}^M 2\Re \left\{ {\trace}\left( \overline{\mat{Z}}_m^{\herm} \left( \overline{\mat{Y}}_m \! - \! \overline{\mat{X}} \right) \right) \right\} , 
    \vspace{-3mm}
\end{align}
where ${\trace}$ is a trace operator.
} {
Using \eqref{eqn:augmented_Lagrangian_admm__generic_form}, then according to \ac{KKT} optimality conditions---see, \eg, \cite{Boyd2004ConvexOptimization}, in particular, stationarity condition at the optimal primal values $\left\{ \overline{\mat{Y}}_m^\star \right\}_{m=1}^M$ and $\overline{\mat{X}}^\star$, and dual variable $\left\{ \overline{\mat{Z}}_m^\star \right\}_{m=1}^M$ satisfy
\begin{align}
\label{eqn:stationarity_condition__v_var_with_dual_var_in_admm_ver1}
    &\mat{0} \in \frac{\partial}{\partial {\left({\overline{\mat{X}}^\star}\right)}^*} \mathcal{L}\left(\left\{ \overline{\vec{Y}}_m^\star \right\}_{m=1}^M,  {\overline{\mat{X}}^\star}, \left\{ \overline{\vec{Z}}_m^\star \right\}_{m=1}^M \right) \nonumber \\ 
    \Longleftrightarrow&
    \mat{0}
    \in  \partial g\left( {\overline{\mat{X}}^\star} \right) -  \sum_{m=1}^M \overline{\mat{Z}}_m^\star \vspace{-3mm} \\
    \label{eqn:stationarity_condition__u_var_with_dual_var_in_admm_ver1}
    &\mat{0} \in \frac{\partial}{\partial {\left({\overline{\mat{Y}}_m^\star}\right)}^*} \mathcal{L}\left(\left\{ \overline{\vec{Y}}_m^\star \right\}_{m=1}^M,  {\overline{\mat{X}}^\star}, \left\{ \overline{\vec{Z}}_m^\star \right\}_{m=1}^M \right) \nonumber \\  
    \Longleftrightarrow&
    \mat{0} 
    \in  \partial f_m\left( {\overline{\mat{Y}}_m^\star} \right)   + \overline{\mat{Z}}_m^\star, 
\end{align}
where, for \ac{EVM} unconstrained, $\partial g\left( {\overline{\mat{X}}^\star} \right) \! = \! \left\{ \left( \overline{\mat{X}}^\star \! - \!  {\mat{X}} \right) \right\}$  and, for \ac{EVM} constrained, $\partial g\left( {\overline{\mat{X}}^\star} \right) \! = \! \partial \mathcal{X}_{\mathbfcal{E}} \! \equiv \! \mathcal{N}_{\mathbfcal{E}}$, where $ \mathcal{N}_{\mathbfcal{E}}$ corresponds to a normal cone (see, \eg, \cite[Chapter~3]{Beck2017}); and $\partial f_m\left( {\overline{\mat{Y}}_m^\star} \right) \! = \! \partial \mathcal{X}_{\mathbfcal{C}_m} \! \equiv \! \mathcal{N}_{\mathbfcal{C}_m}$.}
{
The primal feasibility satisfies
\begin{equation} \label{eqn:primal_feasibility_condition_in_admm}
    \overline{\mat{Y}}_m^\star - \overline{\mat{X}}^\star = \mat{0} \quad \forall m = 1,\ldots,M.
\end{equation}
}
{Towards this end, we analyze the proposed Algorithm~\ref{alg:consensus_admm_msp} and Algorithm~\ref{alg:consensus_eadmm_msp}, which for sufficiently large iterations satisfy the abovementioned optimality conditions using the stated assumptions.}
{
In the first step of (E)\ac{ADMM}, $\overline{\mat{X}}^{\left( i+1 \right)}$ minimizes the update step-1 (cf. \eqref{eqn:scaled_consensus_eadmmm_step1}), \ie,}
{
$\mat{0} \! \in \! \partial g\!\left( \overline{\mat{X}}^{\left( i+1 \right)} \!\right) \! - \! \sum_{m=1}^M \rho \!\left( \! \overline{\mat{Y}}_m^{\left( i  \right)} \! - \! \overline{\mat{X}}^{\left( i+1 \right)} \! + \! \overline{\mat{Z}}_m^{\left( i \right)} \right)$. Using dual variable update \eqref{eqn:scaled_consensus_eadmmm_step3}  and rearranging the terms yields $\mat{0} \! \in \! \partial g\left( \overline{\mat{X}}^{\left( i+1 \right)} \right) \! - \! \sum_{m=1}^M  \rho \overline{\mat{Z}}_m^{\left( i+1 \right)} \! + \! \sum_{m=1}^M  \rho \!\left( \overline{\mat{Y}}_m^{\left( i +1 \right)} \! - \! \overline{\mat{Y}}_m^{\left( i \right)} \right)$.}
{
Now, we state $\left(\!\overline{\mat{Y}}_m^{\left( i+1 \right)} \! - \! \overline{\mat{Y}}_m^{\left( i \right)}\right) \! \rightarrow \! 0$ when $i \! \rightarrow \! \infty$ because of the assumption that the dual residual $\left( \! \overline{\mat{X}}^{\left( i+1 \right)} \! - \! \overline{\mat{X}}^{\left( i \right)} \right) \! \rightarrow \! 0$ and the primal residual $\left( \! \overline{\mat{Y}}_m^{\left( i+1 \right)} \! - \! \overline{\mat{X}}^{\left( i+1 \right)} \right) \! \rightarrow \! 0$. Thus, asymptotically, $\overline{\mat{X}}^{\left( i+1 \right)}$ update satisfies the stationarity condition~\eqref{eqn:stationarity_condition__v_var_with_dual_var_in_admm_ver1}.} 
{Similarly, we have
$\mat{0} \! \in \!  \partial f_m\!\left( \!  \overline{\mat{Y}}_m^{\left( i + 1\right)} \right)  \! + \! \rho \! \left( \!  \overline{\mat{Y}}_m^{\left( i + 1\right)} \! - \! \overline{\mat{X}}^{\left( i + 1\right)} \! + \! \overline{\mat{Z}}_m^{\left( i \right)} \right) \!=\! \partial f_m\!\left( \!  \overline{\mat{Y}}_m^{\left( i + 1\right)} \right) \!+\! \rho \overline{\mat{Z}}_m^{\left( i+1 \right)} $ in the step-2 update (cf. \eqref{eqn:scaled_consensus_eadmmm_step2}), where in the last equality have used the dual variable update \eqref{eqn:scaled_consensus_eadmmm_step3}. Thus, step-2 always satisfies the stationarity condition \eqref{eqn:stationarity_condition__u_var_with_dual_var_in_admm_ver1} for~sufficiently~large~$i$.} 

{Finally, primal feasibility \eqref{eqn:primal_feasibility_condition_in_admm} is satisfied by the assumption $\left( \overline{\mat{Y}}_m^{\left( i+1 \right)} \! - \! \overline{\mat{X}}^{\left( i+1 \right)} \right) \! = \! 0$, when $\!{i \rightarrow +\infty} $.} 
\end{proof}
\end{theorem}

\section{Proof of \eqref{eqn:projection_operator_rank1_quadratic_constraint}: projection onto rank~1 ellipsoid} \label{sec:proof_of_proj_onto_rank1_matrix}
The projection problem reads $\underset{\vec{z}}{\text{minimize}}  \left\| \vec{x} - \vec{z} \right\|_2^2$ $\text{subject to} \
\vec{z}^\herm {\widetilde{\mat{A}}} \vec{z} - b \leq  0$
where ${\widetilde{\mat{A}}} = \vec{u} \vec{u}^\herm$. So, the Lagrangian can be formed as $
L\left(\vec{z}, \mu \right)  = \left\| \vec{x} - \vec{z} \right\|_2^2  + \mu \left( \left| \vec{u}^\herm \vec{z} \right|^2 - b \right)$.
According to the \ac{KKT} conditions \cite{Boyd2004ConvexOptimization}, in particular due to complementary slackness condition, if the given $\vec{x}$ is feasible, that is, fulfils the constraint then $\vec{z} = \vec{x}$ and the Lagrange multiplier would correspond to $\mu = 0$. However, if the given $\vec{x}$ does not fulfil the constraint, then $\mu > 0$ and the inequality constraint can be converted to the equality constraint such that constraint $\left| \vec{u}^\herm \vec{z} \right|^2 = b \Longleftrightarrow \vec{u}^\herm \vec{z} = \sqrt{b} \exp\left( \iota \theta \right)$ for some unknown angle $\theta$. Now, we assume that the angle $\theta$ is known, then the Lagrangian can be expressed as $
L\left(\vec{z}, \mu^\prime \right)  = \left\| \vec{x} - \vec{z} \right\|_2^2  + \mu^\prime \left(  \vec{u}^\herm \vec{z}  - \sqrt{b} \exp\left( \iota \theta \right) \right).
$

Now, according to the \ac{KKT} conditions, we set gradient of $L\left(\vec{z}, \mu^\prime \right)$ with respect to $\vec{z}$ to~$0$, such that $\vec{z} = \vec{x} - \mu^\prime \vec{u}$. Plugging this in the constraint yields $\mu^\prime = \frac{\vec{u}^\herm \vec{x} - \sqrt{b} \exp\left( \iota \theta \right)}{\left\| \vec{u} \right\|_2^2}$. Thus,
$\vec{z} = \vec{x} - \mu^\prime \vec{u} = \vec{x} - \left( \frac{\vec{u}^\herm \vec{x} - \sqrt{b} \exp\left( \iota \theta \right)}{\left\| \vec{u} \right\|_2^2} \right) \vec{u}$. Following \cite{Huang2016Consensus-ADMMProgramming}, the optimal $\theta$ that minimizes the objective $\left\| \vec{x} - \vec{z} \right\|_2^2 = \left\| \left( \frac{\vec{u}^\herm \vec{x} - \sqrt{b} \exp\left( \iota \theta \right)}{\left\| \vec{u} \right\|_2^2} \right) \vec{u} \right\|_2^2$ is an angle of $\vec{x}^\herm \vec{u}$. Thus, $\exp\left( \iota \theta \right) = \frac{\vec{u}^\herm \vec{x}}{\left| \vec{u}^\herm \vec{x} \right|}$. Hence, the projection result follows and given in \eqref{eqn:projection_operator_rank1_quadratic_constraint}.

\section{Derivation of \ac{SSP} Algorithm} \label{sec:derivation_of_ssp_algorithm}

We present the derivations of the \ac{SSP} algorithm, whose pseudo-code is outlined in Algorithm~\ref{alg:solution_ssp_ver1}. In the \ac{SSP} algorithm, the set of Lagrange multipliers $\left\{ \mu_m^{(0)} \right\}_{m=1}^M$ are initialized assuming that the considered $m$-th multiplier is present while others are absent, \ie, boiling down to rank~1 case where the multiplier is computed in closed-form by following \cite{Huang2016Consensus-ADMMProgramming}. 
\begin{lemma} \label{lemma:mu_value_for_rank1_closed_form}
Let $M = 1$, then a closed-form solution to the Lagrange multiplier $\mu_m$ is 
\be \label{eqn:mu_value_for_rank1_closed_form}
    \mu_m = \frac{1}{\lambda_1^m} \left(\left| \vec{a}\left(\nu_m\right)^{\tran} \vec{d}_j \right| \sqrt{\left( \frac{\lambda_1^m}{\gamma_m} \right)} - 1 \right),
\ee
where $\mu \geq 0$ and $\lambda_1^m = \|\vec{a}\left(\nu_m\right)\|_2^2$. 
\end{lemma}
\begin{proof}
Following \cite[Section~IIIB]{Huang2016Consensus-ADMMProgramming} and performing algebraic manipulations, we derive the closed-form solution.
\end{proof}

After initialization, at every given iteration, we apply coordinate descent iterative scheme for the \ac{SSP} algorithm that essentially utilizes the closed-form solution for rank~1 scenario.

Using Lagrangian $L\left(\overline{\vec{d}}_j, \left\{\mu_m\right\}\right)$ \eqref{eqn:lagrangian_for_ssp_algo}, the stationarity condition of the \ac{KKT} conditions, \ie, setting the gradient of $L\left(\cdot\right)$ with respect to $\overline{\vec{d}}_j$ to $0$, yields
\be \label{eqn:closed_form_rank1_quadratic_problem_but_unknown_phasor}
\overline{\vec{d}}_j = \left( \mat{I}_N + \sum \limits_{m=1}^{M} \mu_m {\overline{\mat{A}}}_m \right)^{-1} \vec{d}_j \coloneqq \mat{G}^{-1} \vec{d}_j.
\ee
For brevity, we define 
\begin{align} \label{eqn:definition_of_M_matrix__1}
 \mat{G} \! \equiv \! \mat{G}\left(\mu_m \right) \!
 \coloneqq \! \left( \! \mat{I}_N \! + \! \sum \limits_{m=1}^{M} \mu_m \overline{\mat{A}}_m \! \right) 
 \! = \! \left( \mat{G}_{\backslash m} \! + \! \mu_m {\overline{\mat{A}}}_m \right) ,
\end{align} 
where $\mat{G}_{\backslash m} \coloneqq \mat{I}_N + \sum \limits_{n=1 \backslash m}^{M} \mu_n {\overline{\mat{A}}}_n$.

The matrix inversion of $\mat{G}(\mu_m)$ utilizing Sherman-Morrison formula \cite{Meyer2000MatrixAlgebra} is
\begin{align} \label{eqn:rank_one_matrix_inv}
    \mat{G}(\mu_m)^{-1} \!= \!\left( \mat{G}_{\backslash m}^{-1}  \! -  \! \frac{\mu_m \mat{G}_{\backslash m}^{-1} \overline{\mat{A}}_m \mat{G}_{\backslash m}^{-1}}{1 \! + \! \mu_m \vec{a}\left(\nu_m\right)^{\trans} \mat{G}_{\backslash m}^{-1} \vec{a}\left(\nu_m\right)^{*} } \right),
\end{align}
and noting the fact that the matrix $\mat{G}(\mu_m)$ is a sum of rank one matrices then the matrix inversion can be performed iteratively, cf. Lemma \ref{lemma:matrix_inversion_sum_of_rank_one_matrices}. In order to obtain the set of $\left\{\mu_m\right\}$ Lagrange multipliers, we employ coordinate descent scheme \cite{Tseng2001Convergence1}. Let us say we compute $\mu_m$ in a given cycle, then we fix other multipliers $\mu_n \ \forall n$ but excluding $\mu_m$. We propose to compute and update all the $M$ Lagrange multipliers cyclically. In a given iteration cycle, we compute $\mu_m$ Lagrange multiplier and fixing other multipliers by utilizing the complementary slackness condition of the \ac{KKT} conditions and the dual feasibility condition: 1) If $\mu_m = 0$, then $\overline{\vec{d}}_j =  \mat{G}_{\backslash m} \vec{d}_j$ (cf. \eqref{eqn:definition_of_M_matrix__1}), and 2) If $\mu_m > 0$, then the inequality constraint should be an equality constraint. By plugging~\eqref{eqn:closed_form_rank1_quadratic_problem_but_unknown_phasor} and \eqref{eqn:definition_of_M_matrix__1} in the constraint such that
\begin{align}
\label{eqn:primal_solution__kkt__complementary_slackness__3}
&\overline{\vec{d}}_j^{\rm H} \overline{\mat{A}}_m  \overline{\vec{d}}_j \!-\! \gamma_m  \! = \! \overline{\vec{d}}_j^{\rm H} \vec{a}\left(\nu_m\right)^{\rm *}\vec{a}\left(\nu_m\right)^{\trans} \overline{\vec{d}}_j - \gamma_m  \! = \! 0 \nonumber \\ \Longleftrightarrow&  \vec{a}\left(\nu_m\right)^{\trans} \mat{G}\left(\mu_m\right)^{-1} \vec{d}_j \!=\! \sqrt{\gamma_m} \exp\left(\iota \phi\right) ,
\end{align}
where $\phi$ is a free parameter. We now obtain the Lagrange multiplier $\mu_m$ from \eqref{eqn:primal_solution__kkt__complementary_slackness__3} utilizing \eqref{eqn:rank_one_matrix_inv} after some algebraic manipulation 
\begin{align}
&\vec{a}\left(\nu_m\right)^{\trans} \left( \mat{G}_{\backslash m} + \mu_m {\overline{\mat{A}}}_m \right)^{-1} \vec{d}_j = \sqrt{\gamma_m} \exp\left(\iota \phi\right)\\
 \label{eqn:primal_solution__kkt__complementary_slackness__4}
\Longleftrightarrow&\left( \vec{a}\left(\nu_m\right)^{\trans} \! \mat{G}_{\backslash m}^{-1} \vec{d}_j \!-\! \mu_m \frac{\vec{a}\left(\nu_m\right)^{\trans} \!\mat{G}_{\backslash m}^{-1} {\overline{\mat{A}}}_m \mat{G}_{\backslash m}^{-1} \vec{d}_j}{1 \! + \! \mu_m \vec{a}\left(\nu_m\right)^{\trans} \mat{G}_{\backslash m}^{-1} \vec{a}\left(\nu_m\right)^{*} } \right) \nonumber \\  
 &= \sqrt{\gamma_m} \exp\left(\iota \phi\right) .
\end{align}

For brevity, let $	\alpha_1 \! \coloneqq \! \vec{a}\left(\nu_m\right)^{\trans} \mat{G}_{\backslash m}^{-1} \vec{d}_j$ and $\alpha_2 \! \coloneqq \! \vec{a}\left(\nu_m\right)^{\trans} \mat{G}_{\backslash m}^{-1} \vec{a}\left(\nu_m\right)^{*}$ such that \eqref{eqn:primal_solution__kkt__complementary_slackness__4} can be rewritten as
\begin{align*}
&\alpha_1 - \frac{\mu_m \alpha_1 \alpha_2}{1 + \mu_m \alpha_2} = \frac{\alpha_1 }{1 + \mu_m \alpha_2} = \sqrt{\gamma_m} \exp\left(\iota \phi\right)  \\
\Rightarrow& \mu_m = \Re \left\{ \frac{\alpha_1 \exp\left(-\iota \phi\right) - \sqrt{\gamma_m}}{\sqrt{\gamma_m} \ \alpha_2} \right\}\ .
\end{align*}

\iftrue
\section{Useful Lemma}
\begin{lemma}[Matrix inversion with sum of rank one matrices] \label{lemma:matrix_inversion_sum_of_rank_one_matrices}
The matrix inversion of $\mat{A}_{r+1} = \mat{G} + \sum_{i=1}^{r} \mat{H}_i$ with sum of $r$ rank one matrices $\mat{H}_i$ can be obtained iteratively
$\mat{A}_{k+1}^{-1} = \mat{A}_{k}^{-1} - g_k \mat{A}_{k}^{-1}\mat{H}_{k} \mat{A}_{k}^{-1} \ \forall k=1,\ldots,r$,
where $\mat{A}_{1} = \mat{G}$ and $g_k = \nicefrac{1}{\left(1 + \trace\left(\mat{A}_{k}^{-1} \mat{H}_{k}\right)\right)}$.
\end{lemma}
\begin{proof}
Let $\mat{G}$ and $\mat{G} + \mat{H}$ be invertible matrices, and $\mat{H} = \sum_{i=1}^{r} \mat{H}_i$ with $\rank{\left(\mat{H}_i\right)}= 1$. 
Let $\mat{A}_{k+1} = \mat{G} + \sum_{i=1}^{k} \mat{H}_i$ be invertible. By initializing $\mat{A}_{1} = \mat{G}$, then utilizing \cite{Miller1981OnMatrices} result, we achieve
$\mat{A}_{k+1}^{-1} = \mat{A}_{k}^{-1} - g_k \mat{A}_{k}^{-1}\mat{H}_{k} \mat{A}_{k}^{-1} \ , \ {\mathrm{ for }} \ k = 1,\ldots,r$ where $g_k = \nicefrac{1}{\left(1 + \trace\left(\mat{A}_{k}^{-1} \mat{H}_{k}\right)\right)}$. 
\end{proof}
\fi

\end{appendices}

\section*{Acknowledgement}
This paper is dedicated to the memory of Prof. Peter H\"{a}ndel who prematurely passed away.

\bibliographystyle{IEEEtran}
\bibliography{IEEEabrv,ref}


\end{document}